%% file: et1.tex
\documentclass[pre,twocolumn,showkeys,showpacs,superscriptaddress,preprintnumbers,floatfix]{revtex4-1}

\usepackage{etex}
\usepackage{ifpdf}
\usepackage{dcolumn}
\usepackage{url}
\usepackage{amsmath}
\usepackage{amscd}
\usepackage{amsfonts}
\usepackage{amssymb}
\usepackage{bm}   % bold math
\usepackage{bbm}   
\usepackage{braket}
\usepackage{verbatim}
\usepackage{stmaryrd}
\usepackage{amsthm}
\usepackage{xcolor}
\usepackage{setspace}
\usepackage{subfigure}

\ifpdf
  \usepackage[pdftex]{graphicx}
\else
  \usepackage[dvips]{graphicx}
\fi

\theoremstyle{plain}	\newtheorem{Lem}{Lemma}
\theoremstyle{plain}	
\theoremstyle{plain} 	\newtheorem{Cor}{Corollary}
\theoremstyle{plain} 	
\theoremstyle{plain} 	
\theoremstyle{plain} 	
\theoremstyle{plain} 	\newtheorem{Prop}{Proposition}
\theoremstyle{plain} 	
\theoremstyle{plain} 	
\theoremstyle{plain}	
\theoremstyle{plain}	\newtheorem{Def}{Definition} 
\theoremstyle{plain}

\input{cmechabbrev}

\input{oharp}
\renewcommand{\H}{\operatorname{H}}
\newcommand{\I}{\operatorname{I}}

\newcommand{\IA} {\boldsymbol{\InAlphabet}}
\newcommand{\IS} {\InputSimple}
\newcommand{\IFuture} {\orharp{\IS}}
\newcommand{\IPast} {\olharp{\IS}}

\newcommand{\isym} {\insymbol}
\newcommand{\ifuture} {\orharp{\isym}}
\newcommand{\ipast} {\olharp{\isym}}
\newcommand{\IAll} {\olrharp{\IS}}
\newcommand{\iall} {\olrharp{\isym}}

\newcommand{\OA} { \boldsymbol{\OutAlphabet}}
\newcommand{\OS} {\OutputSimple}
\newcommand{\OFuture} {\orharp{\OS}}
\newcommand{\OPast} {\olharp{\OS}}
\newcommand{\OPastSet} {\olharp{\OA}}

\newcommand{\osym} {\outsymbol}

\newcommand{\opast} {\olharp{\osym}}
\newcommand{\OAll} {\olrharp{\OS}}
\newcommand{\oall} {\olrharp{\osym}}

\newcommand{\JA} { \boldsymbol{\InAlphabet} \times  \boldsymbol{\OutAlphabet}}
\newcommand{\JS} {(\InputSimple,\OutputSimple)}

\newcommand{\JPast} {\olharplow{\JS}}
\newcommand{\JPastSet} {\olharplow{(\IA,\OA)}}
\newcommand{\jsym} {(\insymbol,\outsymbol)}

\newcommand{\jpast} {\olharplow{\jsym}}
\newcommand{\JAll} {\olrharplow{\JS}}

\newcommand{\CSSet} {\CausalStateSet}
\newcommand{\CS} {S} %overloaded this to fit with the convention 'capital=variable' 'bold,mathcal=set.' Do not use \CausalState.
\newcommand{\cs} {\causalstate}
\newcommand{\CSTSet} { \boldsymbol{\mathcal{T}} }
\newcommand{\ASSet} {\boldsymbol{\mathcal{R}}}
\newcommand{\AS} {R}

\newcommand{\PSSet} {\widehat{\boldsymbol{\AlternateStateSet}}}
\newcommand{\PS} {\widehat{R}}
\newcommand{\ps} {\widehat{\rho}}
\newcommand{\SAll} {\olrharp{\CS}}

\newcommand{\SPast} {\olharp{\CS}}

\newtheorem{definition}{Definition}

\newtheorem{theorem}{Theorem}

 \usepackage{tikz}
  \usetikzlibrary{external}
  \usetikzlibrary{arrows}
  \usetikzlibrary{automata}
  \usetikzlibrary{positioning}

  \tikzstyle{vaucanson}=[
    node distance=2.5cm,
    bend angle=15,
    auto,
    every loop/.style={},
    every edge/.style={->,draw=black,line width=1.2,>=latex,shorten <=1pt,shorten >=1pt},
    every state/.style={draw=black,line width=2,font=\large},
    loop right/.style={right,out=22,in=-22,loop},
    loop above/.style={above,out=112,in=68,loop},
    loop left/.style={left,out=202,in=158,loop},
    loop below/.style={below,out=292,in=248,loop},
    loop above right/.style={right,out=67,in=23,loop},
    loop above left/.style={left,out=157,in=113,loop},
    loop below left/.style={left,out=247,in=203,loop},
    loop below right/.style={right,out=337,in=293,loop},
  ]

\begin{document}

\title{Computational Mechanics of Input-Output Processes:\\
Structured transformations and the $\epsilon$-transducer}

\author{Nix Barnett}
\email{nix@math.ucdavis.edu}
\homepage{http://math.ucdavis.edu/~nix/}
\affiliation{Complexity Sciences Center}
\affiliation{Mathematics Department}

\author{James P. Crutchfield}
\email{chaos@ucdavis.edu}
\homepage{http://csc.ucdavis.edu/~chaos/}
\affiliation{Complexity Sciences Center}
\affiliation{Mathematics Department}
\affiliation{Physics Department\\
University of California at Davis\\
One Shields Avenue, Davis, CA 95616}

\date{\today}

\bibliographystyle{unsrt}

% ************************* ABSTRACT *************************
\begin{abstract}
Computational mechanics quantifies structure in a stochastic process via its
causal states, leading to the process's minimal, optimal predictor---the \eM.
We extend computational mechanics to communication channels coupling two
processes, obtaining an analogous optimal model---the \eT---of the stochastic
mapping between them. Here, we lay the foundation of a structural analysis of
communication channels, treating joint processes and processes with input. The
result is a principled structural analysis of mechanisms that support
information flow between processes. It is the first in a series on the
structural information theory of memoryful channels, channel composition, and
allied conditional information measures.

\noindent

\end{abstract}

\preprint{Santa Fe Institute Working Paper 14-12-046}
\preprint{arxiv.org:1412.2690 [cond-mat.stat-mech]}

\maketitle
\
% ****************************************************************

\begin{spacing}{0.8}
\tableofcontents
\end{spacing}

\section{Introduction}

Arguably, the distinctive character of natural and engineered systems lies in
their organization. This is in contrast to differences, say, in how random they
are or in their temperature. Computational mechanics provides an analytical and
constructive way to determine a system's organization \cite{Crut12a},
supplementing the well developed statistical physics view of systems in terms
of disorder---via thermodynamic entropy and free energy. The contrast
begs a question, though, how does organization arise?

As a step to an answer, we extend computational mechanics from describing
individual systems to analyze organization in transformations between systems.
We build on its well established methods to describe a calculus for detecting
the emergence of structure. Indeed, natural systems are nothing, if not the
result of transformations of energy, information, or both. There is no lack of
examples: Filtering measurement time series to discover the temporal behavior
of hidden states \cite{Pack80}; Maxwell's Demon translating measurement
information to control actions that rectify thermal fluctuations into work,
locally defeating Thermodynamics' Second Law
\cite{Mand012a,Boyd14b,Boyd15a}; sensory
transduction in the retina that converts light intensity and spectra into
neural spike trains \cite{Riek99}; perception-action cycles in which an agent
decides its future behavior based on its interpretation of its environment's
state  \cite{Cuts11a,Gord11a}; and, finally, firms that transform raw materials into finished goods
\cite{Padg03a}.

We refer to our objects of study as \emph{structured transformations} to
emphasize the particular focus on basic questions of organization. How complex
is a transformation? What and how many resources are required to implement it?
What does it add to an input in producing an output? Randomness, structure, or
both? What is thrown away? How is its own organization reflected in that of its
output?

The framework addresses these questions, both quantitatively and from first
principles. It is the first in a series. Foremost, it's burden is to lay the
groundwork necessary to answer these questions. Sequels introduce information
measures to classify the kinds of information processing in joint input-output
processes and in structured transformations. To delineate the underlying
mechanisms that produce organization, they analyze a range of examples and
describe a number of interesting, even counterintuitive, properties of
structured transformations.

The questions posed are basic, so there is a wide range of applications and
of historical precedents. Due to this diversity and to avoid distracting from
the development, we defer reviewing them and related work until later, once the
context has been set and the focus, benefits, and limitations of our approach
are clear.

The following analyzes communication channels and channel composition in terms
of intrinsic computation \cite{Crut88a,Crut12a}. As such, it and the entire
series, for that matter, assume familiarity with stochastic
processes at the level of Ref. \cite{Gray09a}, information theory at the
level of Refs.  \cite{Cove06a,Yeun08a}, and computational mechanics at the
level of Refs. \cite{Crut08b,Crut10a}. These serve as the default
sources for statements in our review.

We first present a brief overview and relevant notation for how computational
mechanics describes processes. We extend this to describe controlled
processes---processes with input. Several examples that illustrate the basic
kinds of input-output process are then given, by way of outlining an
organizational classification scheme. At that point, we describe the global
\eM\ for joint input-output processes. Using this we introduce the \eT,
defining its structural complexity and establishing its optimalities. We close
with a thorough analysis of the example input-output processes and compare and
contrast \eTs\ with prior work.

\section{Processes and their Presentations}

\subsection{Stationary, Ergodic Processes}

The temporal \emph{stochastic process} we consider is a one-dimensional chain
$\ldots \OS_{-2} \OS_{-1} \OS_0 \OS_1 \OS_2 \ldots$ of discrete random variables
$\{ \OS_t \}_{t \in \mathbb{Z}}$ that takes values $\{ \osym_t \}_{t \in \mathbb{Z}}$ over a finite or countable alphabet $\OA$.
A finite block $\OS_i \OS_{i+1} \ldots \OS_{j-1}$ of variables with
$t \in [i,j)$ is denoted $\OS_{i:j}$. The left index is always inclusive,
the right always exclusive. Let $\OAll$ denote the bi-infinite chain
$\OS_{-\infty:\infty}$, and let $\olrharp{\OA}$ denote the set
of all bi-infinite sequences $\oall$ with alphabet $\OA$. $\OPast_{t} = \ldots \OS_{t-2} \OS_{t-1}$ is the
\emph{past} leading up to time $t$, not including $t$, and
$\OFuture_{t} = \OS_t \OS_{t+1} \ldots$ is the \emph{future} leading from it,
including $t$. 

We can also use the time index notation above to specify the time origin of a
process or bi-infinite sequence. This is useful when we are comparing two such
processes or sequences. For example, if we wish to say that $\OAll$ is $\IAll$
delayed by three time steps, this can be written as $\OAll_3 = \IAll_0$. To
indicate the time origin in \emph{specific} realized symbol values in a chain,
we place a period before the symbol occurring at $t = 0$: e.g., $\oall = \ldots
acbb.caba \ldots$, where $\osym_0 = c$. Finally, if the word $abcd$ follows a
sequence of random variables $\OS_{i:j}$ we denote this by $\OS_{i:j+4} =
\OS_{i:j}abcd$, for example. A word occurring before a sequence of random
variables is denoted by a similar concatenation rule; e.g., $\OS_{i-4:j} =
abcd\OS_{i:j}$.

A stochastic process is defined by its \emph{word distributions}:
\begin{align}
\Prob(\OS_{t:t+L}) \equiv \{ \Prob(\OS_{t:t+L} = \osym_{t:t+L}) \}_{\osym_{t:t+L} \in \OA^L},
\label{eq:WordDistributions}
\end{align}
for all $L \in \mathbb{Z}^+$ and $t \in \mathbb{Z}$. 

We will often use an equivalent definition of a stochastic process as a random
variable defined over the set of bi-infinite sequences $\olrharp{\OA}$. In
this case, a stochastic process is defined by an indexed set of probabilities of the form:
\begin{align}
\Prob(\OAll) \equiv \{ \Prob(\OAll \in \sigma ) 
 \}_{ \sigma \subseteq \scriptsize \olrharp{\OA}} ~,
\label{eq:CylinderSetDistributions}
\end{align}
where $\sigma$ is a measurable set of bi-infinite sequences. We can obtain a
process' word probabilities by selecting appropriate measurable
subsets---\emph{cylinder sets}---that correspond to holding the values of a contiguous subset of random variables fixed.

In the following, we
consider only \emph{stationary processes}---those invariant under time
translation:
\begin{align*}
\Prob(\OS_{t:t+L}) &= \Prob(\OS_{0:L}) ~\text{and} \\
\Prob(\OAll_t) &= \Prob(\OAll_0)
  ~,
\end{align*}
for all $t \in \mathbb{Z}$ and $L \in \mathbb{Z}^+$.
This property ensures that a process's behavior has no
\emph{explicit} dependence on time origin.

We will also primarily limit the discussion to \emph{ergodic} stationary
processes---processes where any realization $\oall$ gives good empirical
estimates $\widehat{\Prob}(\OS_{0:L})$ of the process's true word probabilities
$\Prob(\OS_{0:L})$ \cite{Gray74a}. That is, for any finite realization
$\osym_{0:M}$, the empirical estimate $\widehat{\Prob}(w)$ of a word $w=w_0 w_1
\ldots w_{L-1}$ of length $L$, converges almost surely to the true process
probability $\Prob(\OS_{0:L}=w)$ as $M \to \infty$, where:
\begin{align*}
\widehat{\Prob}(w) &
  \equiv \sum_{t=0}^{M-L} \frac{\mathbf{I}_w (\osym_{t:t+L})}{M-L+1}
  ~,
\end{align*}
and the indicator function $\mathbf{I}_w (\osym_{t:t+L})$ equals $1$ when
$\osym_{t:t+L}=w$, and $0$ otherwise. This property ensures, among other
things, that any particular realization of the process reflects the behavior of
the process in general---a useful property when we do not have freedom to
re-initialize the system at will.

\subsection{Examples}

To illustrate key ideas in the development, we use several example processes,
all with the binary alphabet $\OA = \{ 0, 1 \}$. They are used repeatedly in
the following and in the sequels.

\subsubsection{Biased Coin Process}

The \emph{Biased Coin Process} is an independent, identically distributed (IID)
process, where word probabilities factor into a product of single-variable
probabilities:
\begin{align*}
\Prob(\OS_{0:L}) = \Prob(\OS_{0}) \Prob(\OS_{1}) \cdots \Prob(\OS_{L-1}) ~,
\end{align*}
where $\Prob(\OS_{t}) = \Prob(\OS_{0})$ for all $t$.
If $n$ is the number of $0$s in $\osym_{0:L}$, then
$\Prob(\osym_{0:L}) = p^n (1-p)^{L-n}$,
where $p = \Prob(\OS_{0} = 0)$.

\subsubsection{Period-$2$ Process}

The \emph{Period-$2$ Process} endlessly repeats the word $01$, starting with
either a $0$ or a $1$ with equal probability. It is specified by the word
distributions:
\begin{align*}
\Prob(\OS_{i:j} & = (\ldots 1010.1010 \ldots)_{i:j} ) = \tfrac{1}{2} ~\text{and}\\
\Prob(\OS_{i:j} & = (\ldots 0101.0101 \ldots)_{i:j} ) = \tfrac{1}{2}
  ~,
\end{align*}
where $i,j \in \mathbb{Z}, i < j$.

\subsubsection{Golden Mean Process}

The \emph{Golden Mean Process} generates all binary sequences, except those
with consecutive $0$s. After a $1$ is generated, the next $0$ or $1$ appears
with equal likelihood. Its word distributions are determined by a Markov
chain with states $\OS_0=0$ and $\OS_0=1$ having probabilities
$\Prob(\OS_0 = 0) = 1/3$ and $\Prob(\OS_0 = 1) = 2/3$, respectively, and
transition probabilities:
\begin{equation}
\begin{aligned}
\Prob(\OS_1=0|\OS_0=0) &= 0 ~, \\
\Prob(\OS_1=1|\OS_0=0) &= 1 ~, \\
\Prob(\OS_1=0|\OS_0=1) &= \tfrac{1}{2} ~, \text{~and}\\
\Prob(\OS_1=1|\OS_0=1) &= \tfrac{1}{2} ~.
\end{aligned}
\label{eq:GMPCondlProbs}
\end{equation}

\subsubsection{Even Process}

The \emph{Even Process} generates all binary sequences, except that a $1$
always appears in an even-length block of $1$s bordered by $0$s. After an even
number of $1$s are generated, the next $0$ or $1$ appears with equal
likelihood. Notably, the Even Process cannot be represented by a Markov chain
of any finite order or, equivalently, by any finite list of conditional
probabilities over words. As we will see, though, it can be represented
concisely by a finite-state hidden Markov model. To see this, we must first
introduce alternative models---here called \emph{presentations}---for a given
process.

\subsection{Optimal Presentations}

We can completely, and rather prosaically, specify a stationary process by
listing its set of word probabilities, as in Eq. (\ref{eq:WordDistributions}),
or conditional probabilities, as in Eq. (\ref{eq:GMPCondlProbs}). As with the
Even Process, however, generally these sets are infinite and so one prefers to
use finite or, at least, more compact descriptions. Fortunately, there is a
canonical \emph{presentation} for any stationary process---the \eM\ of computational
mechanics \cite{Crut88a,Crut08b}---that is its unique, optimal, unifilar
generator of minimal size; which we now define.

Given a process's word distribution, its \eM\ is constructed by regarding any
two infinite pasts $\opast$ and $\opast^\prime$ as equivalent when they lead to
the same distribution over infinite futures---the same \emph{future morphs}, of
the form $\Prob(\OFuture|\opast)$. This grouping is given by the \emph{causal
equivalence relation} $\sim_\epsilon$:
\begin{align}
\opast \sim_\epsilon \opast^\prime
  \iff \Prob(\OFuture | \OPast = \opast ) = \Prob(\OFuture | \OPast = \opast^\prime )
  ~.
\label{eq:PredEqReln}
\end{align}
The equivalence classes of $\sim_\epsilon$ partition the set $\OPastSet$ of all
allowed pasts. The classes are the process's \emph{causal states}. The indexed
set of causal states is denoted by $\CausalStateSet$ and has elements $\cs_i$.
Note that $i$ is not a time index, but an element of an index set with the same
cardinality as $\CSSet$. The associated random variable over alphabet
$\CausalStateSet$ is denoted by $\CS$. The \emph{$\epsilon$-map} is a function
$\epsilon: \OPastSet \to \CausalStateSet$ that takes a given past to its
corresponding causal state or, equivalently, to the set of pasts to which it is
causally equivalent:
\begin{align*}
\epsilon(\opast) & = \sigma_i \\
  & = \{ \opast^\prime: \opast \sim_\epsilon \opast^\prime \} ~.
\end{align*}

\begin{figure}
  \centering
  \includegraphics[scale=.95]{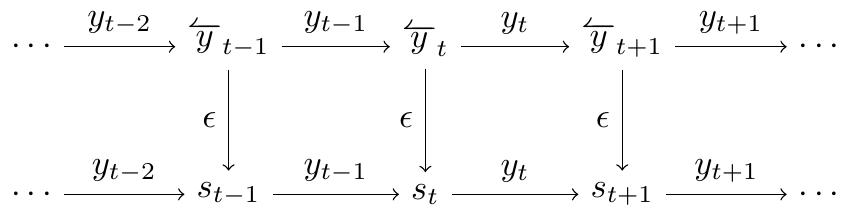}
\caption{Process lattice: The dynamic inherited by the causal states
  $s_t = \cs_i \in \CausalStateSet$ from a process's pasts via the
  $\epsilon$-map: $s_t = \epsilon (\ldots \osym_{t-2} \osym_{t-1} )
  \xrightarrow[]{\osym_t} s_{t+1} = \epsilon (\ldots \osym_{t-2} \osym_{t-1} \osym_t ) $.
  }
\label{fig:processlattice}
\end{figure} 

The $\epsilon$-map induces a dynamic over pasts: Appending a new symbol
$\osym_t$ to past $\opast_{t}$ produces a new past $\opast_{t+1} = \opast_{t}
\osym_t$. This, in turn, defines a stochastic process---the \emph{causal-state
process}---with random variable chain $\SAll = \CS_{-\infty:\infty} = \ldots
\CS_{-1} \CS_0 \CS_1 \ldots$, where at time $t$ each $\CS_t$ takes on some
value $s_t = \cs_i \in \CausalStateSet$. The relationship between a process's
pasts and its causal states is summarized in Fig.~\ref{fig:processlattice}. The
map from the observed chain $\OAll$ to the internal state chain $\SAll$ is the
\emph{causal state filter}. (We return to this mapping later on.) The dynamic
over causal states is specified by an indexed set $\CSTSet$ of
\emph{symbol-labeled transition matrices}:
\begin{align*}
\CSTSet \equiv
  \{ T^{(\osym)}\}_{\osym \in \OA}
  ~,
\end{align*}
where $T^{(\osym)}$ has elements:
\begin{align*}
T^{(\osym)}_{ij} = \Prob(\CS_1 = \cs_j, \OS_0 = \osym | \CS_0 = \cs_i). 
\end{align*}

The $\epsilon$-map also induces a distribution $\pi$ over causal states, with elements:
\begin{align*}
\pi (i) & = \Prob (\CS_0 =\cs_i) \\
  & = \Prob \big( \epsilon \big( \OPast \big) = \cs_i \big)
  ~.
\end{align*}

Since the process is stationary and the $\epsilon$-map is time independent,
$\pi$ is also stationary. We therefore refer to $\pi$ as the process's \emph{stationary distribution}.
Note that for ergodic processes, the stationary distribution can be calculated
directly from the transition matrices alone \cite{Bill61a}.

The tuple $\left( \OA, \CSSet, \CSTSet \right)$ consisting of the process'
alphabet, causal state set, and transition matrix set is the process'
\emph{\eM}.

The \eM\ is a process's unique, maximally predictive, minimal-size unifilar
presentation \cite{Crut88a,Crut98d,Shal98a}. In other words, the causal states
are as predictive as any rival partition $\ASSet$ of the pasts $\OPastSet$. In
particular, any given causal state $\cs_i$ is as predictive as any of its pasts
$\opast \in \epsilon^{-1}(\cs_i)$. Measuring predictive ability via the mutual information between states and future observations, this translates to the statement that:
\begin{align*}
\I[\CS_0 ; \OFuture_0] & = \I[\OPast_0 ; \OFuture_0] \\
                   & \geq \I[\AS_0 ; \OFuture_0] ~,
\end{align*}
where $\AS_0 \in \ASSet$.
Moreover, among all equally predictive (\emph{prescient rival}) partitions
$\PrescientStateSet$ of the past, the causal states minimize the state Shannon
entropy: $\H[\CS] = \H[\pi]\leq \H[\PrescientState]$. Due to the \eM's minimality, the
\emph{statistical complexity} $\Cmu = \H[\CS] = \H[\pi]$ measures the amount of historical
information a process stores.

A process's \eM\ presentation has several additional properties that prove
useful. First, the causal states form a Markov chain. This means that the
\eM\ is a type of \emph{hidden Markov model}. Second, the causal states are
\emph{unifilar}:
\begin{align*}
\H[\CS_{t+1}| \OS_t, \CS_t] & = 0 ~.
\end{align*}
That is, the current state and symbol uniquely determine the next state.
This is necessary for an observer to maintain its knowledge of a process's
current causal state while scanning a sequence of symbols. Third, unlike
general (that is, nonunifilar) hidden Markov models, one can calculate a
process's key informational properties directly from its \eM. For example,
a process's \emph{entropy rate}
$\hmu$ can be written in terms of the causal states:
\begin{align*}
\hmu = \H[\OS_0 | \CS_0] ~.
\end{align*}
And, using the methods of Refs. \cite{Crut08a,Crut08b}, a process's past-future
mutual information---the \emph{excess entropy} $\EE$---is given by its forward
$\FutureCausalState$ and reverse $\PastCausalState$ causal states:
\begin{align}
\EE & \equiv \I[\OPast;\OFuture] \\
  & = \I[\PastCausalState;\ForwardCausalState]
  \nonumber
  ~.
\label{eq:EE}
\end{align}
Generally, the excess entropy is only a lower bound on the information
that must be stored in order to predict a process: $\EE \leq \Cmu$. This difference is captured by the process's \emph{crypticity} $\PC =
\Cmu - \EE$.

Since they are conditioned on semi-infinite pasts, the causal states defined
above correspond to \emph{recurrent} states in a process's complete
\eM\ presentation. They capture a process's time-asymptotic behavior. An
\eM\ also has \emph{transient} causal states that arise when conditioning on
finite-length pasts, as well as a unique start state, which can be either transient or recurrent. When the underlying process is ergodic, they can be derived from the recurrent causal
states using the mixed-state method of Ref. \cite{Crut08b}. In general, they can be obtained from a modified causal equivalence relation, extended to include finite pasts. We omit them, unless otherwise
noted, from our development.

The preceding introduced what is known as the \emph{history} specification of
an \eM: From a stationary, ergodic process, one derives its \eM\ by applying equivalence relation
Eq. (\ref{eq:PredEqReln}). There is also the complementary \emph{generator}
specification: As a generator, rather than using equivalence classes over
a process's histories, the \eM\ is defined as a strongly connected hidden Markov model whose
transitions are unifilar and whose states are probabilistically distinct. Such an \eM\ generates
a unique, ergodic, stationary process, and is the same \eM\ that we would obtain by applying the
causal equivalence relation to said generated process \cite{Trav11a}.  The
following uses both history and generator \eMs; which will be clear in context.

\subsection{Example Process \EMs}

The cardinality of a process's causal state set $\CausalStateSet$ need not be
finite or even countable; see, e.g., Figs. 7, 8, 10, and 17 in Ref.
\cite{Crut92c}. For simplicity in the following, though, we restrict our study
to processes with a finite or countably infinite number of causal states. This
allows us to depict a process graphically, showing its \eM\ as an
\emph{edge-labeled directed graph}. Nodes in the graph are causal states and
edges, transitions between them. A transition from state $\cs_i$ to state
$\cs_j$ while emitting symbol $\osym$ is represented as an edge connecting the
corresponding nodes that is labeled $\mathbf{y} : p$. (Anticipating our needs later on, this differs slightly from prior notation.) Here, $p =
T_{ij}^{(y)}$ is the state transition probability and $\osym$ is the symbol
emitted on the transition. Figure \ref{fig:machines} displays \eM\
state-transition diagrams for the example processes.

\begin{figure}
\centering
\subfigure[Biased Coin Process.]
  	{ \includegraphics{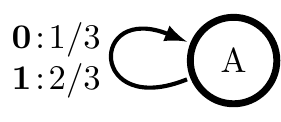} }
\qquad
  \subfigure[Period-2 Process.]
  { \includegraphics{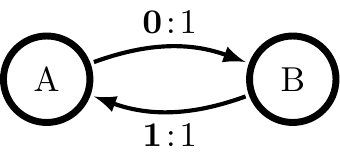} }
\qquad
  \subfigure[Golden Mean Process.]
  { \includegraphics{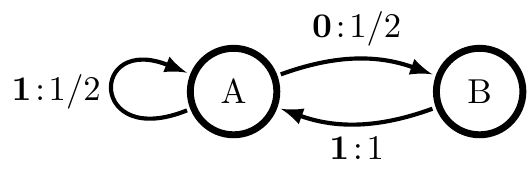} }
\qquad
    \subfigure[Even Process.]
	{ \includegraphics{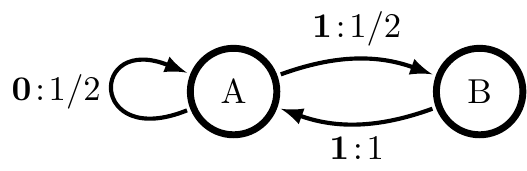} }
\caption{\EMs\ for the example processes. Transitions labeled
  $\mathbf{\osym} : p$, where $p = T_{ij}^{(y)}$ is the state transition
  probability and $\osym$ is the symbol emitted on the transition.
  }
\label{fig:machines}
\end{figure} 

Since the Biased Coin Process's current output is independent of the
past, all pasts are causally equivalent. This leads to a single causal state
$A$, occupied with probability $1$. (See Fig. \ref{fig:machines}(a).)
Therefore, the Biased Coin Process has a statistical complexity of $\Cmu = 0$
bits. It's excess entropy $\EE$ also vanishes. This example illustrates the general
property that IID processes lack causal structure. They can be quite random;
the example here has an entropy rate of $\hmu = \H(2/3) \approx 0.918$ bits per
step, where $\H(p)$ is the binary entropy function.

In contrast, the Period-$2$ Process has two causal states, call them $A$ and
$B$, that correspond to pasts that end in either a $1$ or a $0$, respectively.
(See Fig. \ref{fig:machines}(b).) Since the states are occupied with equal
probability, the Period-$2$ Process has $\Cmu = 1$ bit of stored information.
In this case, $\EE = \Cmu$. It is perfectly predictable, with $\hmu = 0$ bits
per step.

The two causal states of the Golden Mean Process also correspond to pasts
ending with either a $0$ or $1$, but for it the state transition structure
ensures that no consecutive $0$s are generated. (See Fig.
\ref{fig:machines}(c).) Causal state $A$ is occupied with $\Prob(A) =
2/3$ and state $B$ with $\Prob(B) = 1/3$, giving $\Cmu =
\H(2/3) \approx 0.918$ bits. Note that the excess entropy is substantially
less---$\EE \approx 0.2516$ bits---indicating that the process is
\emph{cryptic}. We therefore must store additional state information above and
beyond $\EE$ in order to predict the process \cite{Crut08a}. Its entropy rate
is $\hmu = \H(2/3) \approx 0.918$ bits per step.

As mentioned already, the Even Process cannot be represented by a finite Markov
chain. It can be represented by a finite hidden Markov model, however. In
particular, its \eM\ provides the most compact presentation---a two-state
hidden Markov model. Causal state $A$ corresponds to pasts that end with an
even block of $1$s, and state $B$ corresponds to pasts that end with an odd
block of $1$s. (See Fig. \ref{fig:machines}(d).) Since the probability
distribution over states is the same as that of the Golden Mean Process, the
Even process has $\Cmu \approx 0.918$ bits of statistical complexity and $\hmu
\approx 0.918$ bits per step. In contrast, the excess entropy and statistical
complexity are equal: $\EE = \Cmu$. The Even Process is not cryptic.

\section{Input-Output Processes and Channels}

Up to this point, we focused on a stochastic process and a particular canonical
presentation of a mechanism---the \eM---that can generate it. We now turn to
our main topic, generalizing the \eM\ presentation of a given process to a
presentation of a controlled process---that is, to input-output (I/O) processes.
I/O processes are related to probabilistic extensions of Moore's
sequential machines \cite{Paz71a,Hopc79}, probabilistic extensions of codes
from symbolic dynamics \cite{Lind95a}, and, perhaps more naturally, to
Shannon's communication channels \cite{Shan48a}. And so, we refer to them
simply as \emph{channels}. There are important differences from how channels
are developed in standard treatments \cite{Cove06a}, though. The principal
difference is that we consider channels with memory, while the latter in its
elementary treatment, at least, typically considers memoryless channels or
channels with restricted forms of memory \footnote{Though see Ref. \cite[see
Ch. 7]{Ash65a} and for early efforts Refs. \cite{Blac58a} and \cite{Blac61a}.}.
In addition, with an eye to applications, the framework here is adapted to
\emph{reconstruct} channels from observations of an input-output process.  (A
topic to which we return at the end.) Finally, the development places a unique
emphasis on detecting and analyzing structure inherent in I/O processes.
Our development parallels that for \eMs; see Ref. \cite[and citations
therein]{Crut12a}.

To begin, we define I/O processes. To make these concrete, we provide
several example channels, leveraging the example processes and their
\eMs\ already described. The examples naturally lead to the desired
extension---the \eT. The development then turns to the main properties of \eTs.

Loosely speaking, a channel defines a coupling between stochastic processes. It
can be memoryful, probabilistic, and even anticipate the future. In other
words, the channel's current output symbol may depend probabilistically upon
its past, current, and future input and output symbols. As such, we will be led
to generalize the memoryless and anticipation-free communication channels
primarily studied in elementary information theory \cite{Cove06a}. 

\begin{definition}
A \emph{channel} $\OAll\big|\IAll$ with \emph{input alphabet} $\IA$ and
\emph{output alphabet} $\OA$  is a collection of stochastic processes over
alphabet $\OA$, where each such process $\OAll|\iall$ corresponds to a
bi-infinite \emph{input} sequence in $\olrharp{\IA}$:
\begin{align}
\OAll | \IAll \equiv
  \{ \OAll | \iall \}_{\scriptsize \iall \in \scriptsize\olrharp{\IA}}
\end{align}
That is, each fixed realization $\iall = \ldots \isym_{-1} \isym_0
\isym_1 \ldots$ over input alphabet $\IA$ is mapped to a stochastic process $\OAll | \iall = \OAll|\ldots \isym_{-1} \isym_0 \isym_1 \ldots$ over
alphabet $\OA$.
\end{definition}

If we are given a process $\IAll$---an \emph{input process}---then a channel
maps this distribution over sequences to a joint process $\JAll$, which can
be marginalized to obtain a process $\OAll$---the \emph{output process}. We can
characterize a channel by an indexed set $\mathcal{P}$ of \emph{conditional
word probabilities}: 
\begin{align*}
\Prob(\OS_{t:t+L} |& \iall) \\
  \equiv
  & \{ \Prob\big(\OS_{t:t+L}=\osym_{t:t+L} \big| \IAll = \iall\big) \}_{\osym_{t:t+L} \in \OA^L , \scriptsize\iall \in \scriptsize\olrharp{\IA}}
  ~.
\end{align*}
Equivalently, we can represent a channel as a distribution over bi-infinite
sequences (output) conditioned on a particular bi-infinite sequence
(input). In other words, a channel can be characterized by an indexed set:
\begin{align*}
\Prob\big(\OAll \big| \iall\big)
\equiv \{ \Prob\big(\OAll \in \sigma \big| \IAll = \iall\big)
\}_{\sigma \subseteq \scriptsize\olrharp{\OA}, \scriptsize\iall \in \scriptsize\olrharp{\IA}}
  ~.
\end{align*}
Given an input process' distribution $\Prob(\IAll)$ and a channel's
distribution $\Prob\big(\OAll | \iall\big)$ for all inputs $\iall$, we obtain the output process' distribution as follows:
\begin{align*}
\Prob\big( \OAll\big) &=
\int \Prob\big(\OAll,  \iall\big) \,\mathrm{d} \iall \\
& = \int \Prob\big(\OAll | \iall\big) \Prob(\iall) \,\mathrm{d} \iall,
\end{align*}
where the first integrand shows the appearance of the intermediate joint process $\JAll$.

Let's say a few words about definitional choices made up to this point. As
defined, the channels we consider are \emph{total} (defined for every possible
input sequence). One could extend the
definition to allow for \emph{partial} channels (defined only for a subset of
possible sequences), but we do not consider
such channels in what follows. This is the primary reason for defining a
channel's domain in terms of bi-infinite sequences rather than say,
collections of finite words. Such channels would not necessarily be total.
Also, requiring that channels be defined for \emph{every finite} input word is
restrictive, as even the simplest channels may not have well defined behavior
for say, a single symbol input word. We could instead define channels over a
\emph{subset} of all finite input words and explicitly add in the restriction
that the channel be total, but this is still more restrictive than the
$\olrharp{\IA}$ definition above. Consider, for example, the channel that
outputs all $1$s if its input sequence contains at least one $1$ and outputs
all $0$s otherwise. Such a channel is undefined for any finite input word that
consists of all $0$s, but is well defined for any bi-infinite binary sequence. It is also true that any total channel defined over finite input
words can be trivially defined over bi-infinite sequences by appending an
arbitrary infinite past and future to each finite word (without changing the
channel's output behavior).

%When relating two chains, input and output in the present setting, relative time
%is important. So, we specify a
%chain's time-origin via $\iall_{t}$. For example, if we wish to say that chain
%$\oall$ is merely $\iall$ shifted by $3$ steps, we write
%$\oall_{0} = \iall_{3}$. 

Stationarity is as useful a property for channels as it is for processes.

\begin{definition}
A channel is \emph{stationary} if and only if its probability distributions are invariant under time translation:
\begin{align*}
\Prob(\OS_{t:t+L}  | \iall_{t} )
  & = \Prob (\OS_{0:L}  | \iall_{0} ) ~\text{and} \\
  \Prob(\OAll_t  | \iall_{t} ) &= \Prob (\OAll_0  | \iall_{0} ) ~,
\end{align*}
for all $t \in \mathbb{Z}$, $L \in \mathbb{Z}^+$, and every input sequence
$\iall$.
\end{definition}

An immediate consequence is that stationary channels map stationary input
processes to stationary output processes:

\begin{Prop}[Stationarity Preservation] When a stationary channel is applied to a stationary input process, the resulting output process is also stationary.

\begin{proof}
One calculates directly:
\begin{align*}
\Prob( \OAll_t) & =
  \int \Prob(\OAll_t | \iall_{t}) \Prob(\iall_{t}) \mathrm{d} \iall_{t} \\
  & = \int \Prob(\OAll_0 | \iall_{0}) \Prob(\iall_{0}) \mathrm{d} \iall_{0} \\
  & = \Prob(\OAll_0) ~,
\end{align*}
where the second equality follows from stationarity of both input process and
channel, once we note that
$\ldots \mathrm{d} \isym_{t-1} \mathrm{d} \isym_{t} \mathrm{d} \isym_{t+1} \ldots = \ldots \mathrm{d} \isym_{-1} \mathrm{d}
\isym_0 \mathrm{d} \isym_{1} \ldots$, since
under shifted indexing the relationships between $x_i$ and infinitesimals $\mathrm{d} x_i$ are preserved and the latter themselves do not change.
\end{proof}
\end{Prop}

We also primarily restrict our discussion to ergodic channels \cite{Gray90a}.

\begin{definition}
An \emph{ergodic channel} is a channel that maps any ergodic input process $\IAll$ to an ergodic joint process $\JAll$.
\end{definition}

Another important channel property is that of causality.

\begin{definition}
A \emph{causal channel} is anticipation-free:
\begin{align*}
\Prob(\OS_{t:t+L}  | \IAll ) = \Prob (\OS_{t:t+L}  | \IPast_{t+L} ) ~.
\end{align*}
That is, the channel has well defined behavior on semi-infinite input pasts
and is completely characterized by that behavior.
\end{definition}

Channel causality is a reasonable assumption when a system has no access to
future inputs. However, as a note of caution in applying the following to
analyze, say, large-scale systems with many components, the choice of
observables may lead to input-output processes that violate causality. For
example, treating spatial configurations of one-dimensional spin systems or
cellular automata as if they were time series---a somewhat common strategy---violates causality. In the following, though, we assume channel
stationarity and causality, unless stated otherwise. 

It is worth noting that causality is often not a severe restriction.
Specifically, the following results extend to channels with finite
anticipation---channels whose current output depends upon $N$ future inputs.
When both the input process and channel are stationary, one delays the
appearance of the channel output by $N$ time indices. This does not change the
output process, but converts \emph{finite anticipation} to finite channel
memory and delayed output. In this way, it is possible to apply the analysis to
follow to channels with anticipation directly, but the optimality theorems
established must be modified slightly.

\begin{figure*}
  \centering
  \subfigure[~Identity Channel]{
 \includegraphics[scale=.33]{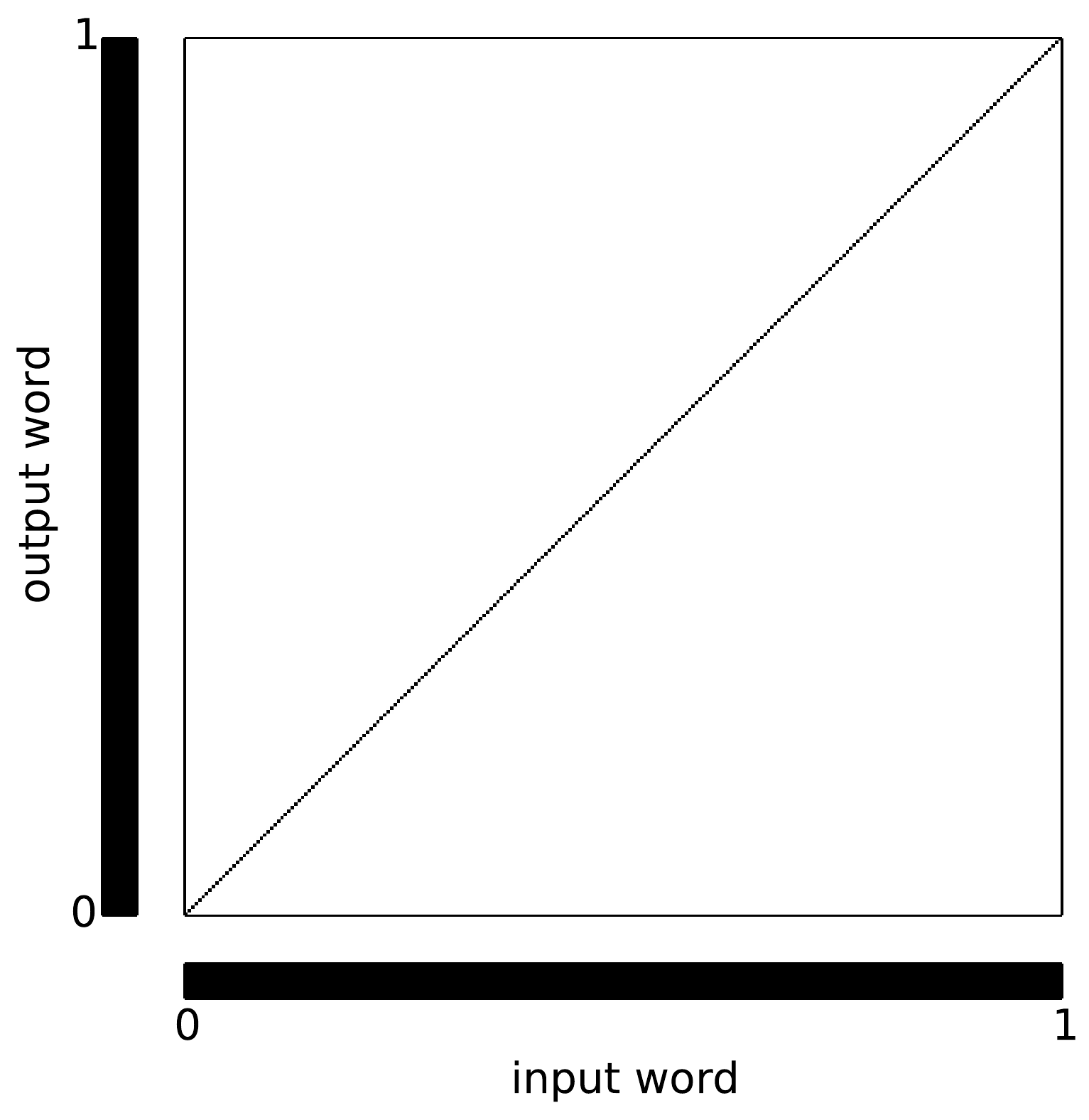}
\label{fig:identityword}
  } \qquad
  \subfigure[~All-is-Fair Channel]{
 \includegraphics[scale=.33]{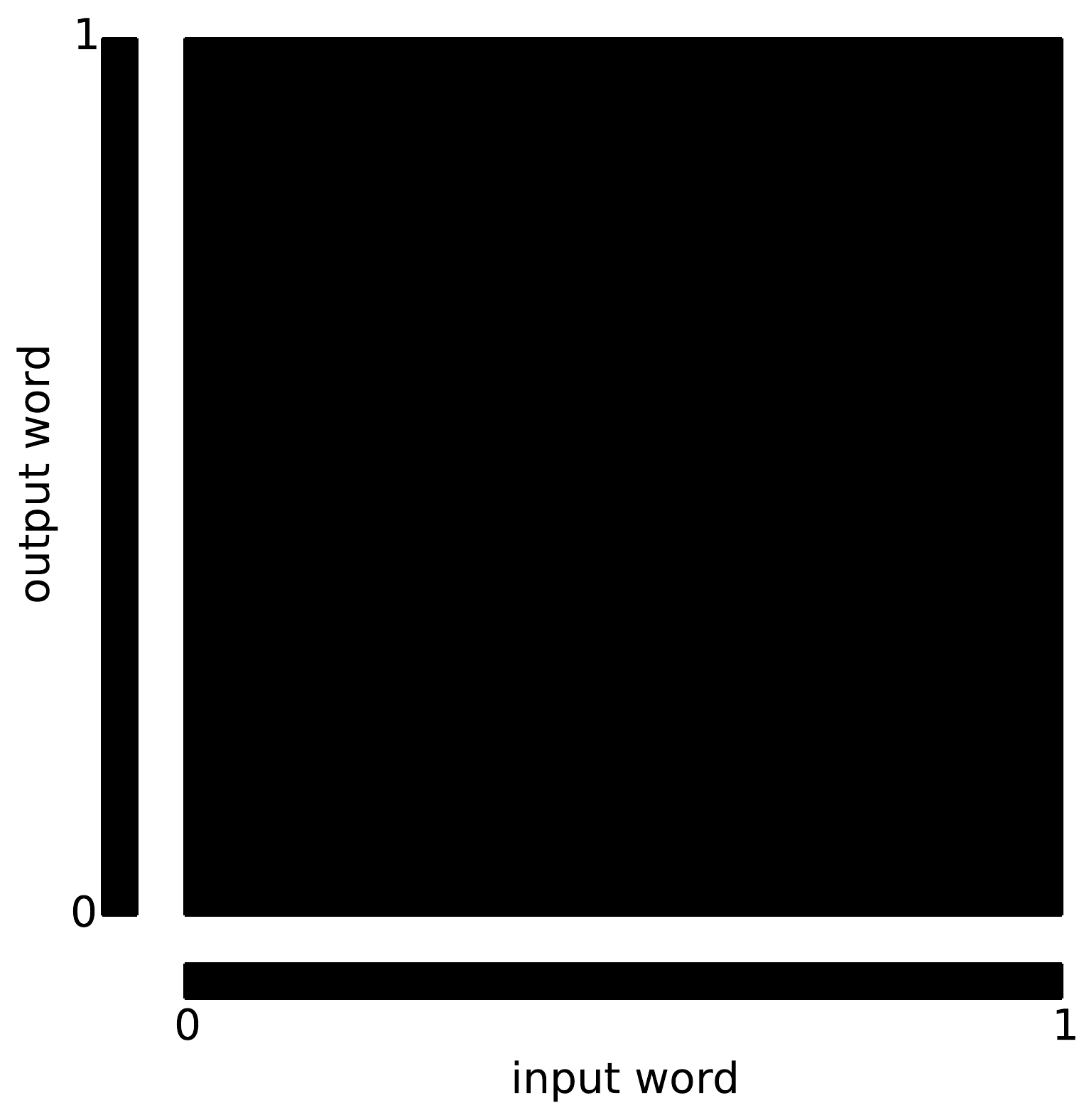}
\label{fig:alltofairword}
  } \qquad
  \subfigure[~Z Channel]{
\includegraphics[scale=.33]{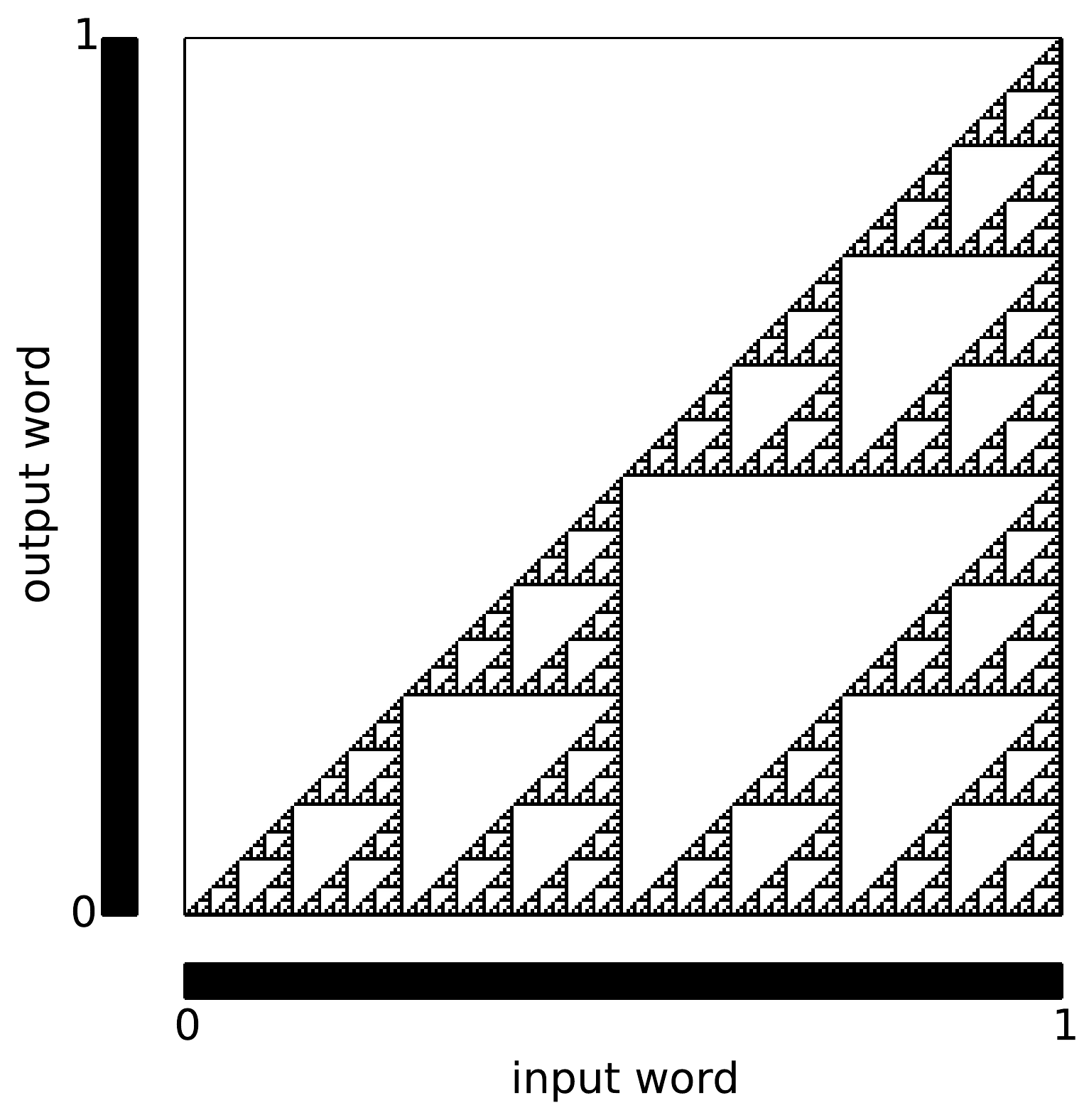}
\label{fig:zchannelword}
  } \qquad
  \subfigure[~Delay Channel]{
\includegraphics[scale=.33]{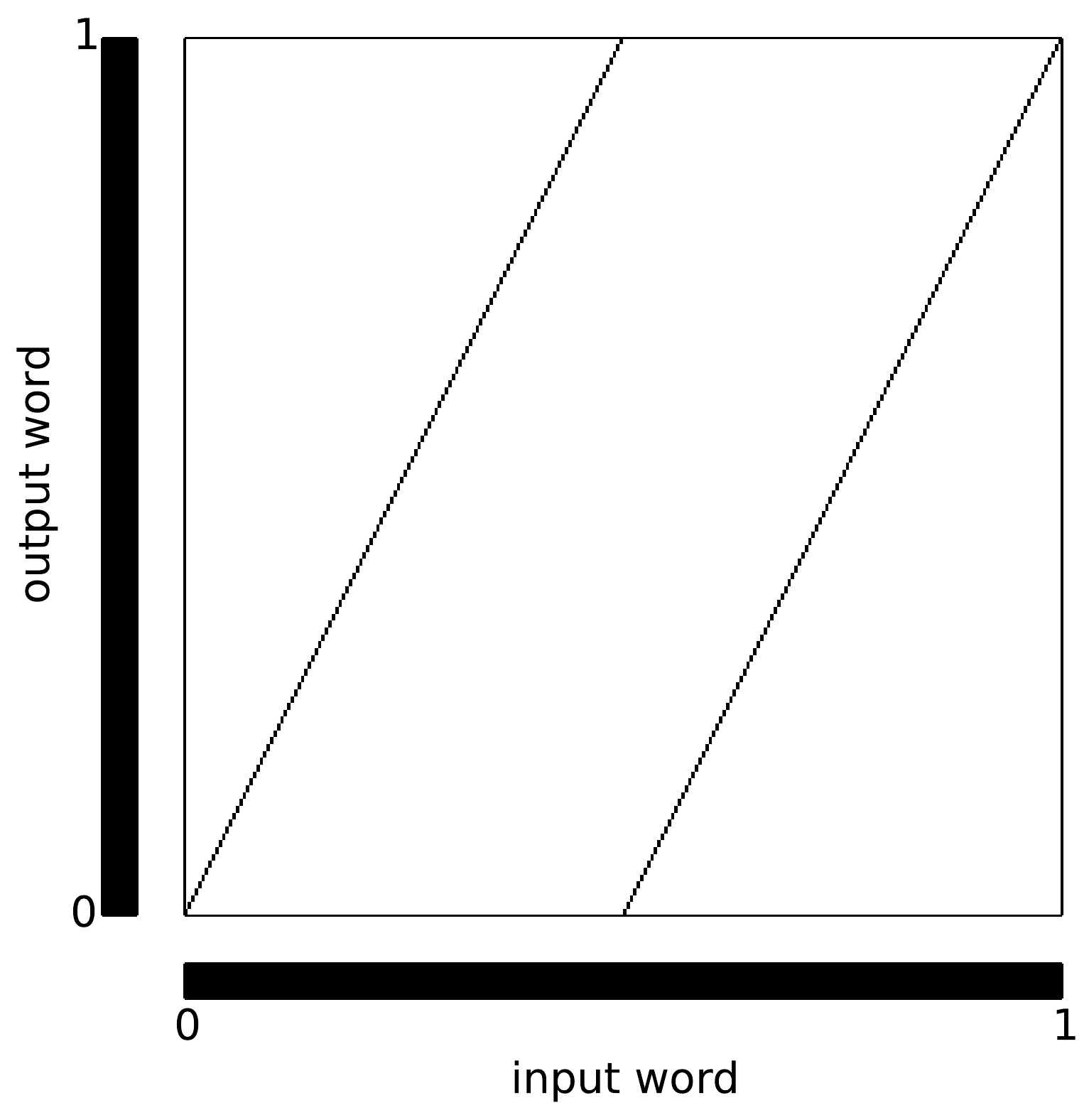}
\label{fig:delayword}
} \qquad
  \subfigure[~Feedforward NOR Channel]{
  \includegraphics[scale=.33]{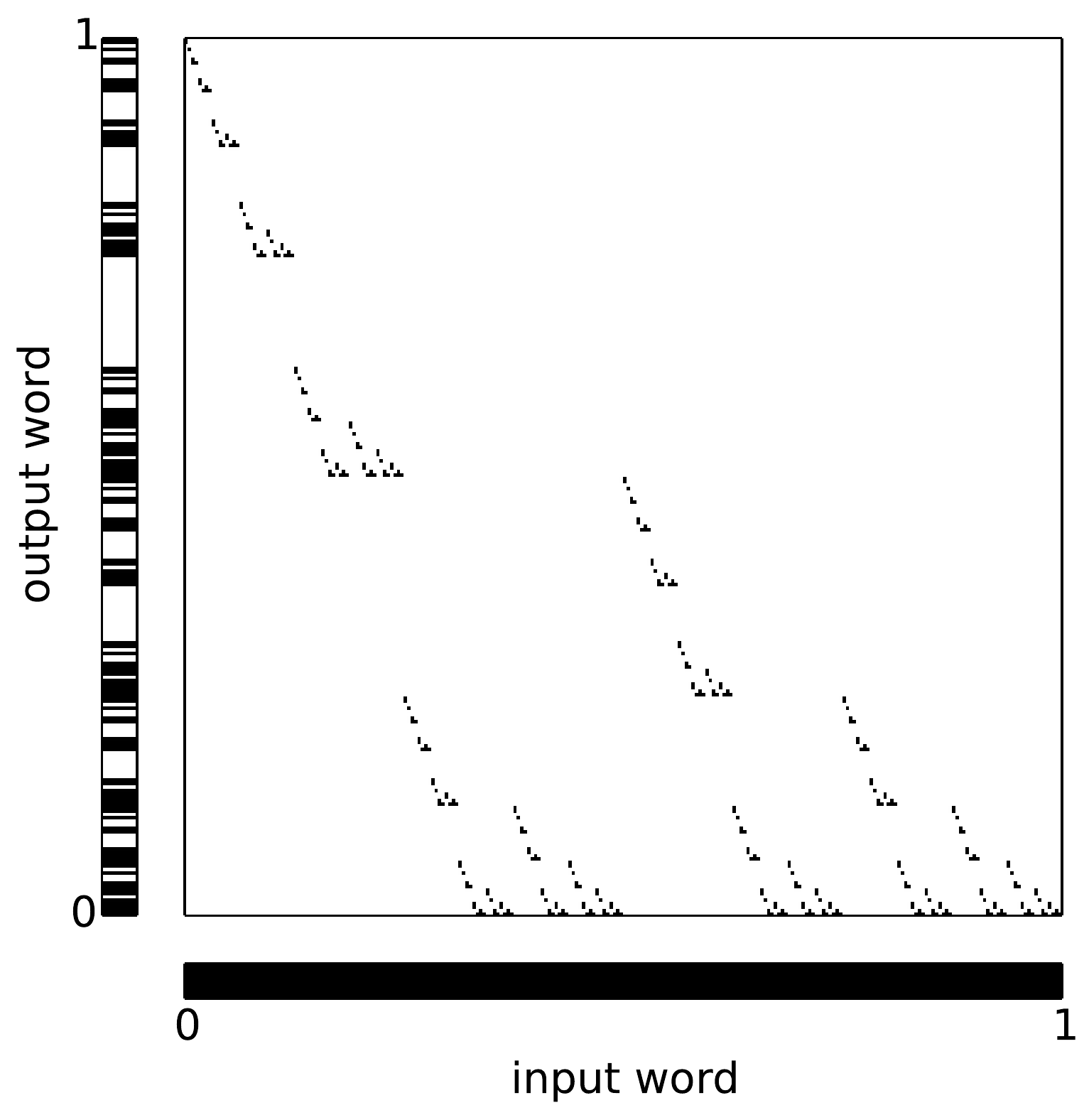}
\label{fig:feedforwardNORword}
  }
\caption{Wordmaps for simple example channels display the map from finite input
  words to finite output words: Allowed input-output pairs $(x_{0:8},y_{0:8})$
  correspond to a dot placed in the unit square at
  $(0.x_0 x_1 \cdots x_7, 0.y_0 y_1 \cdots y_7)$.
  The possible observed input (output) words are displayed below (left) of the horizontal (vertical) axis.
  }
\label{fig:wordplots}
\end{figure*} 

\section{Example Channels and their Classification}

To motivate our main results, consider several example stationary causal
channels with binary input and output. Figure \ref{fig:wordplots} illustrates
the mapping from input words $\{ x_{0:L} = x_0 x_1 \ldots x_{L-1} \}$ to output
words $\{ y_{0:L} = y_0 y_1 \ldots y_{L-1} \}$ for the example channels. The
\emph{wordmaps} there treat each input-output pair $(x_{0:L}, y_{0:L})$ as a
point $p = (p_x,p_y)$ in the unit square, where the input (output) word forms
the binary expansion of $p_x = 0.x_0 x_1 \cdots x_{L-1}$
($p_y = 0.y_0 y_1 \cdots y_{L-1}$). Since the channels are defined for all inputs, the
plots show every input word.

We organize the examples around a classification scheme paralleling that used
in signal processing to highlight memory and the nature of feedback
\cite{oppenheim2010discrete}.

We describe channel behavior using recurrence relations of the form:
\begin{align*}
\OS_t \sim r \left( \IPast_t, \OPast_t, \IS_t \right)
  ~,
\end{align*}
where $r(\cdot)$ is either a logical function, in simple cases, or a
distribution. Note that the next output $\OS_t$ depends only on the past---the
channels are causal.

In general, such a recurrence relation only captures the channel's \emph{present}
behavior. Fortunately, a causal channel's future behavior is summarized by its present
behavior:
\begin{align*}
&\Prob  (\OS_{0:L} \!=\! ab \ldots yz  | \JPast_{0}, \IS_{0:L} \!=\!
\alpha \beta \ldots \psi \omega
  ) \\[6pt]
& = \Prob (\OS_0 \!=\! a | \JPast_{0}, \IS_{0:L} \!=\! \alpha \ldots \omega) \\
&~~~\, \cdot \Prob (\OS_{1} \!=\! b |\JPast_{0}, \IS_{0:L} \!=\! \alpha \ldots \omega, \OS_{0} \!=\!a) \\
&~~\,\cdots \Prob (\OS_{L-1} \!=\! z | \JPast_{0}, \IS_{0:L} \!=\! \alpha \ldots \omega, \OS_{0:L-1} \!=\! a \ldots y) \\[6pt]
& \!\overset{(a)}{=} \! \Prob (\OS_0 \!=\! a | \JPast_{0}, \IS_{0:L} \!=\! \alpha \ldots \omega ) \\
&\hspace{3.7em}~~~\, \cdot \Prob (\OS_{1} \!=\! b |\JPast_{1}, \IS_{1:L} \!=\! \beta \ldots \omega) \\
&\hspace{7.4em}~~\,\cdots  \Prob (\OS_{L-1} \!=\! z | \JPast_{L-1}, \IS_{L-1} \!=\! \omega) \\[6pt]
& \!\overset{(b)}{=} \! \Prob (\OS_0 \!=\! a | \JPast_{0}, \IS_{0} \!=\! \alpha) \\
&\hspace{3.7em}~~~\, \cdot \Prob (\OS_{1} \!=\! b |\JPast_{1}, \IS_{1} \!=\! \beta) \\
&\hspace{7.4em}~~\,\cdots  \Prob (\OS_{L-1} \!=\! z | \JPast_{L-1}, \IS_{L-1} \!=\! \omega) \\[6pt]
& \!\overset{(c)}{=} \! \Prob (\OS_0 \!=\! a | \JPast_{0}, \IS_{0} \!=\! \alpha) \\
&\hspace{5.5em}~~~\, \cdot \Prob (\OS_{0} \!=\! b |\JPast_{0}, \IS_{0} \!=\! \beta) \\
&\hspace{11em}~~\, \cdots  \Prob (\OS_{0} \!=\! z | \JPast_{0}, \IS_{0} \!=\! \omega)
  ,
\end{align*}
where in (a) we merge individual variables into pasts to obtain new pasts, in
(b) we remove input variables that have no effect due to causality, and (c)
follows from stationarity. 

\subsection{Memorylessness}

The \emph{Memoryless Binary Channel}'s (MBC's) current output depends only on
its current input; the analog of an IID process in that its behavior at time $t$
is independent of that at other times. The MBC includes as special cases the
\emph{Binary Symmetric Channel} (BSC) and the \emph{Z Channel} \cite{Cove06a}.
We can summarize the MBC's present behavior with a simplified recurrence relation of the form
$\OS_t \sim r \left( \IS_t \right)$ and
its conditional probabilities factor as follows:
\begin{align*}
\Prob(\OS_{t:t+L} & | \IAll)  \\
 = &\Prob(\OS_{t} | \IS_{t}) \Prob(\OS_{t+1} | \IS_{t+1})
 	\cdots \Prob(\OS_{t+L-1} | \IS_{t+L-1}) ~.
\end{align*}

The first three wordmaps of Fig. \ref{fig:wordplots} illustrate the behavior of
memoryless channels. We see that the \emph{Identity Channel} (Fig.
\ref{fig:wordplots}\subref{fig:identityword}) always maps a word to itself.
Whereas, the \emph{All-is-Fair Channel} (Fig.
\ref{fig:wordplots}\subref{fig:alltofairword}) maps each input word uniformly
to every output word. We immediately see that deterministic channels have
wordmaps with a single filled pixel per plot column.

The Z Channel wordmap is shown in Fig.
\ref{fig:wordplots}\subref{fig:zchannelword}. It transmits all $0$s with no
noise, but adds noise to all $1$s transmitted. The wordmap shows the maximal
noise case, where all $1$s are replaced with the output of a fair coin. We see
that even memoryless channels have nontrivial word mappings. In this case, the
latter forms a self-similar Sierpinski right triangle \cite{Mand77a} in the
unit square.

\subsection{Finite Feedforward}
%Note the use of feedforward as a noun, just as feedback can be used as an adjective or noun.

Now, consider channels whose behavior depends only on a finite input history. Their
behavior on input is analogous to those of order-$R$ Markov chains. They can also be thought
of as stochastic, anticipation-free \emph{sliding block codes} \cite{Lind95a}
with finite memory or as generalized \emph{finite impulse response filters}
\cite{oppenheim2010discrete}. These channels' present behavior can be summarized with a recurrence relation of the
form $\OS_t \sim r(\IS_{t-M:t},\IS_t)$, where $M$ is a finite input history
length such that $M>0$.

The \emph{Delay Channel} simply stores its input at time $t-1$ and outputs it
at time $t$. Its wordmap is shown in
Fig. \ref{fig:wordplots}\subref{fig:delayword}. Those familiar with
one-dimensional iterated maps of the interval will recognize the word mapping
as the \emph{shift map}, capturing the fact that delayed output corresponds to
a binary shift applied to the input word. Note that when viewed as a function
on the space of processes, the Delay Channel acts as the identity.

The \emph{Feedforward NOR Channel}'s output at time $t$ results from
performing a logical NOR ($\downarrow$) on its inputs at times $t$ and $t-1$:
%\begin{align*}
%\OS_t \sim \neg ( \IS_{t-1} + \IS_t) ~.
%\end{align*}
\begin{align*}
\osym_t = \isym_{t-1} \downarrow \isym_t ~.
\end{align*}
The wordmap for the Feedforward NOR Channel is shown in
Fig. \ref{fig:wordplots}\subref{fig:feedforwardNORword}. There, we see that
although the channel is deterministic, the wordmap's self-similarity makes
it difficult to see that the mapping is a function---that there is, in fact,
a single output word for each input.

The additional complexity of the remaining examples does not lead to new types
of apparent graphical structure beyond that seen in the existing wordmaps.
So, wordmaps will be omitted for now. With additional theory developed, we
return to illustrate these channels, but using a more advanced form of wordmap.

\subsection{Infinite Feedforward}

Generally, channels depend on infinitely long input histories. This behavior is analogous to the long-range dependence seen in typical hidden Markov models \cite{Rabi89a,Elli95a} or in \emph{strictly sofic} subshifts \cite{Lind95a}. Channels with dependence upon infinitely long input histories \emph{alone} can also be interpreted as generalized \emph{infinite impulse response filters} \cite{oppenheim2010discrete}. The present behavior of such channels can be summarized by a recurrence relation of the form $\OS_t \sim r \left( \IPast_t , \IS_t \right)$.

The \emph{Odd NOT} Channel stores the parity (even or odd) of the number of
ones observed in its input since the last zero observed; much like the Even
Process. If the parity is currently even, it behaves as the Identity Channel.
If the parity is odd, it outputs the bitwise NOT (bit flip) of its input.
Since the channel's behavior depends on the parity of its input, it cannot
be characterized by finite input histories alone. 

A channel can depend, however, on past \emph{outputs} as well as inputs. Such \emph{feedback} can allow one to replace the infinite-history recurrence relation with one that includes only a finite
history of inputs and outputs. Consider again the Odd NOT
Channel described above. Note that its behavior is determined entirely by the
current input value, as well as the parity of the number of ones observed on
input. In fact, the parity at time $t$ can be summarized by the input and
output at time $t-1$. If $\isym_{t-1}=0$, the parity will always be even. If
$\isym_{t-1}=1$, and $\osym_{t-1}=0$ we know that the parity \emph{was} odd,
since the bit was flipped, but since a $1$ was just observed on input, the
parity is now even. Finally, if $\isym_{t-1}=1$ and $\osym_{t-1}=1$, we know
that the parity was previously even, and the newly observed $1$ makes the
current parity \emph{odd}. Summarizing, we have that:
\begin{align}
\jsym_{t-1}&=(0,0) \Leftrightarrow \text{Even input history parity},
  \nonumber \\
\jsym_{t-1}&=(0,1) \Leftrightarrow \text{Even input history parity}, 
  \nonumber \\
\jsym_{t-1}&=(1,0) \Leftrightarrow \text{Even input history parity},
~\text{and} \nonumber \\
\jsym_{t-1}&=(1,1) \Leftrightarrow \text{Odd input history parity} ~.
\label{eq:OddNotChannel}
\end{align}
By allowing feedback, we can therefore summarize the behavior of the Odd
NOT Channel with a recurrence relation that depends only upon finite history (of length $M=1$):
$\OS_t \sim r \left( \IS_{t-1},\OS_{t-1},\IS_t \right)$.

An interesting observation is that the Odd NOT channel maps the Even Process
to a bit-flipped Golden Mean Process. However, a single process-to-process
mapping does \emph{not} uniquely define a channel. There are an infinite number of channels,
in fact, that map the Even Process to the bit-flipped Golden Mean Process.

\subsection{Finite Feedback}

As seen in the previous section, allowing for even a finite amount of output
feedback can can lead to substantial simplifications in the description of a channel. Let's
consider channels that depend on a finite output history on their own terms.

The most trivial case would be channels that depend solely on output histories, with no dependence on
inputs. Since there is effectively no input, these channels reduce to the
output-only stochastic processes (generators) discussed earlier. Consider, for
example, the \emph{All is Golden Channel} that outputs the Golden Mean Process,
regardless of what input it receives.

A less trivial example is the \emph{Feedback NOR Channel}, similar to the
Feedforward NOR Channel, except the output at time $t$ is the logical NOR
of its current input and \emph{previous output}:
\begin{align*}
\osym_t = \isym_{t} \downarrow \osym_{t-1}
  ~.
\end{align*}   
This channel's behavior is clear in this feedback form. It might be desirable,
however, to find a purely feedforward presentation for the channel. The
recurrence relation for the Feedback NOR Channel can be solved recursively to
give a feedforward presentation that is defined for almost every input history.
We recurse in the following way:
%\begin{align*}
%\OS_t & \sim \neg ( \IS_{t} + \OS_{t-1} ) \\
%  & \sim \neg ( \IS_{t} + \neg (\IS_{t-1} + \OS_{t-2}) )\\
%  & \sim \neg ( \IS_{t} + \neg (\IS_{t-1} + \neg (\IS_{t-2} + \OS_{t-3}) ) ) \\
%  & \sim \cdots
%  ~,
%\end{align*}
\begin{align*}
\osym_t & =  \isym_{t} \downarrow \osym_{t-1} \\
  & =  \isym_{t} \downarrow (\isym_{t-1} \downarrow \osym_{t-2}) \\
  & = \isym_{t} \downarrow (\isym_{t-1} \downarrow (\isym_{t-2} \downarrow \osym_{t-3}) )  \\
  & = \cdots
  ~,
\end{align*}
and so on, until reaching sufficiently far into the input past that a $1$ is
observed. When this happens, the recursion terminates as the output of
the NOR function is always $0$ when either argument is $1$. We can therefore
construct a purely feedforward recurrence relation, but the input histories
can be arbitrarily long---corresponding to arbitrarily long input histories
consisting entirely of $0$s. The resulting feedforward recurrence relation
is defined for all histories, except for the infinite history of all $0$s. 

Note that such ill-defined behavior for certain infinite histories is typical
in systems that have infinite memory lengths and is not a problem specific to channels.
One can be careful to explicitly \emph{define} behavior for such cases, but
this is beyond the scope of our current work, and these pathological histories
typically occur with zero probability.

In contrast, consider replacing the logical NOR in the Feedback
NOR Channel with an exclusive OR (XOR or $\oplus$), thus giving the \emph{Feedback XOR Channel}:\begin{align}
\osym_t = \isym_{t} \oplus \osym_{t-1}.
\end{align}
In this case, solving for a pure-feedforward relation fails since the output of
a logical XOR is never determined by a single argument. This channel
illustrates the fact that a presentation which includes feedback cannot always
be reduced to a pure-feedforward presentation.

\subsection{Infinite Feedback}

Just as we can define channels whose behavior depends upon infinite input
histories, we can define channels whose behavior depends upon infinite
\emph{output} histories. In fact, we have already studied a channel that can
be represented this way. Consider a channel that stores the parity (even or
odd) of the number of ones observed in its \emph{output} since the last zero
observed. If this parity is currently even, it behaves as the Identity
Channel. If the parity is odd, it outputs the bitwise NOT (bit flip) of its
input. This appears to be very similar to the Odd NOT Channel defined above,
but with a dependence on infinite output histories and the present input, rather than infinite input histories. In fact, this is simply a different \emph{presentation} of the Odd NOT Channel. 

It suffices to show that the feedback presentation of the Odd NOT Channel can
be reduced to the same finite history presentation as with its feedforward
presentation. Proceeding as before, we observe that if $\osym_{t-1}=0$, the
output parity is always even. If $\isym_{t-1}=0$ and $\osym_{t-1}=1$, the
previous parity was odd (the bit was flipped), but the $1$ observed on output
makes the output parity even. Finally, if $\isym_{t-1}=1$ and $\osym_{t-1}=1$,
the output parity was even, and the $1$ observed makes the output parity odd.
Summarizing, we obtained the same presentation as the feedforward
presentation specified by Eq. (\ref{eq:OddNotChannel}).

\subsection{Infinite Feedforward-Feedback}

We just examined an example channel whose presentation depends upon infinite histories
when only feedforward or feedback is allowed, but only a finite
history when both feedforward and feedback are allowed. The following example
shows that finding a finite history presentation is not always possible.
Channels of this form are perhaps the most natural channel generalization of
infinite Markov order (strictly sofic) processes.

The \emph{Odd Random Channel} stores the parity of its input history just as
the Odd NOT Channel does, and it again behaves as the identity when the parity
is currently even. When the parity is odd, the channel outputs a $0$ or $1$
with equal probability. Like the Odd NOT Channel, this has an infinite
feedback presentation that stores the channel's output history parity. The
channel does \emph{not} have any finite history presentation, however. If one
attempts to construct a finite presentation via the recursion unrolling
procedure used for the Odd NOT Channel, it is simple to obtain the following
relationships:
\begin{align*}
\jsym_{t-1}&=(0,0) \Leftrightarrow \text{Even input history parity}, \\
\jsym_{t-1}&=(0,1) \Leftrightarrow \text{Even input history parity},
  ~\text{and} \\
\jsym_{t-1}&=(1,0) \Leftrightarrow \text{Even input history parity}.
\end{align*}
The problem arises from the fact that when $\isym_{t-1}=1$ and
$\osym_{t-1}=1$, the input history parity is uncertain. The channel could have
been operating as the identity (even parity) or giving random output (odd
parity). Looking at progressively longer histories can resolve this
uncertainty, but only once a $0$ has been observed (on either input, output,
or both). This ambiguity requires that we specify arbitrarily long joint
histories to determine the behavior of the channel in general. It is therefore
not possible to construct any finite-history presentation of the Odd Random
Channel.

\subsection{Irreducible Feedforward-Feedback}

As a final example, consider a channel that has a finite presentation when both
feedforward and feedback are allowed, but has \emph{no} pure-feedforward or
pure-feedback presentation. The \emph{Period-2 Identity NOT Channel}
alternates between the identity and bit flipped identity at each time step.
This channel's present behavior is completely determined by whether the
previous input and output bits, $\isym_{t-1}$ and $\osym_{t-1}$, match. When
the bits match, the channel was in its ``identity'' state and is therefore
now in its ``NOT'' state. The opposite is clearly true when the bits do not
match. The channel therefore has a recurrence relation of the form
$\OS_t \sim r(\IS_{t-1}, \OS_{t-1}, \IS_t)$.

Since the behavior does not depend on the particular values of the input or
output, but whether or not they \emph{match}, there is no way to construct a
pure-feedforward or pure-feedback presentation of the channel. This channel
therefore illustrates the notion of \emph{irreducible} output (or input)
memory---dependence upon past output (input) that cannot be eliminated even by
including dependence upon infinite past outputs (inputs).

\subsection{Causal Channel Markov Order Hierarchy}

It turns out that the set of examples above outlines a classification scheme
for causal channels in terms of their Markov orders \cite{Jame10a} that we now
make explicit. Just as Markov order plays a key role in understanding the
organization of processes, it is similarly helpful for channels. Channel Markov
orders are the history lengths required to \emph{completely} specify a causal
channel's behavior, given certain constraints on knowledge of other histories.

\begin{definition}
\noindent
\begin{enumerate}
\item The \emph{pure feedforward Markov order} $\MOrder_{\text{pff}}$ is the
	smallest $M$ such that $\forall L$\\
$\Prob \left( \OS_{0:L} |  \JPast_0, \IS_{0:L} \right) \!=\! \Prob \left( \OS_{0:L} |  \IS_{-M:0}, \IS_{0:L} \right)$.
\item The \emph{pure feedback Markov order} $\MOrder_{\text{pfb}}$ is the
	smallest $M$ such that $\forall L$\\
$\Prob \left( \OS_{0:L} |  \JPast_0, \IS_{0:L} \right) \!=\! \Prob \left( \OS_{0:L} |  \OS_{-M:0}, \IS_{0:L} \right)$.
\item The \emph{channel Markov order} $\MOrder$ is the smallest $M$ such
	that $\forall L$\\
$\Prob \left( \OS_{0:L} |  \JPast_0, \IS_{0:L} \right) \!=\! \Prob \left( \OS_{0:L} |  \JS_{-M:0}, \IS_{0:L} \right)$.
\item The \emph{irreducible feedforward Markov order} $\MOrder_{\text{iff}}$ is the
	smallest $M$ such that $\forall L$\\
$\Prob \left( \OS_{0:L} |  \JPast_0, \IS_{0:L} \right) \!=\! \Prob \left( \OS_{0:L} |  \IS_{-M:0}, \OPast_0, \IS_{0:L} \right)$.
\item The \emph{irreducible feedback Markov order} $\MOrder_{\text{ifb}}$ is the smallest $M$ such that $\forall L$\\
$\Prob \left( \OS_{0:L} |  \JPast_0, \IS_{0:L} \right) \!=\! \Prob \left( \OS_{0:L} |  \IPast_0, \OS_{-M:0}, \IS_{0:L} \right)$.
\end{enumerate}
\end{definition}

For example, we showed that the Odd NOT Channel's presentation requires an infinite history
when only feedforward or feedback is allowed. And so, it has
$\MOrder_{\text{pff}} = \MOrder_{\text{pfb}} = \infty$. However, it only
requires finite history when both are allowed: $\MOrder = 1$. If we have full knowledge of the output past, we still need one symbol of input history in order to characterize the channel, so $\MOrder_{\text{iff}}=1$. Similarly, we have $\MOrder_{\text{ifb}}=1$.

Note that the irreducible feedforward Markov order $\MOrder_{\text{iff}}$
will only be nonzero if the
pure feedback order $\MOrder_{\text{pfb}}$ is undefined. Similarly, the irreducible feedback Markov
order $\MOrder_{\text{ifb}}$ will only be nonzero when the pure feedforward
order $\MOrder_{\text{pff}}$ is undefined. In
words, if a channel has irreducible feedback (feedforward), the channel has no pure feedforward (feedback) presentation. Moreover, the channel Markov
order $\MOrder$ bounds the pure Markov orders from below and the smallest of
the two irreducible Markov orders bounds the channel Markov order from below:
\begin{align*}
\min ( \MOrder_{\text{iff}} , \MOrder_{\text{ifb}}) \leq \MOrder \leq \min ( \MOrder_{\text{pff}}, \MOrder_{\text{pfb}} )
  ~.
\end{align*}

In a sequel, we address input and output memory using information-theoretic
quantities. There, we characterize different \emph{amounts} of input and output
memory and how they relate, whereas here in discussing Markov orders we considered \emph{lengths} of input and output sequences. We also focused on causal channels, which allowed us to restrict our discussion to present behavior and memory lengths. In the case of anticipatory channels, we must explicitly consider \emph{future} behavior, as well as \emph{anticipation} lengths. However, this is best left to another venue, so that we do not deviate too far from the path to our goal.

\subsection{Causal-State Channels}

There is a natural and quite useful channel embedded in any \eM\ presentation
of a stationary process---\emph{the causal-state channel}---that identifies
structure embedded in a process via the \eM's causal states.

Consider a process and its \eM\ $M$. Previously, we described $M$ as a
generator of the process. An \eM, however, is also a \emph{recognizer} of its
process's sequences. Briefly, $M$ reads a sequence and follows the series of
transitions determined by the symbols it encounters. The output of the
causal-state channel is then the sequence of causal states. The operation of
this channel is what we call \emph{causal-state filtering}. Notably, the
induced mapping is a function due to the \eM's unifilarity.

In this way, the channel filters observed sequences, returning step-by-step
associated causal states. Given an \eM, the causal-state filter has the same
topology as the \eM, but input symbols match the \eM\ transition symbols and
output symbols are the state to which the transition goes.

The recurrence relation for causal-state filtering is:
\begin{align*}
\CS_t \sim r \left( \IS_{t-1} , \CS_{t-1} \right) ~.
\end{align*}
As just noted, $r(\cdot)$ is a (nonprobabilistic) function, determined by the
$\epsilon$-map:
\begin{align*}
\CS_t & = \epsilon(\IPast_t) \\
      & = \epsilon(\IS_{t-1} , \IPast_{t-1}) \\
      & = \epsilon(\IS_{t-1} , \CS_{t-1}) ~.
\end{align*}
Thus, the pure feedback order is $\MOrder_{\text{pfb}} = 1$, as is the channel Markov
order $\MOrder = 1$. The pure feedforward order $\MOrder_{\text{pff}}$, however, is the original process's Markov order. For methods to determine the latter see
Ref. \cite{Jame10a}.

In this way, an \eM\ can be used to detect the hidden structures captured
by the causal states. For example, causal-state filtering has been used to
great effect in detecting emergent domains, particles, and particle
interactions in spatially extended dynamical systems \cite{Crut93a},
single-molecule conformational states \cite{Li08a}, and polycrystalline and
fault structures in complex materials \cite{Varn12a}.

\section{Global \EM}

Before introducing channel presentations, we first frame the question in the
global setting of the \eM\ of joint (input-output) processes. This, then,
grounds the \eT\ in terms of an \eM.

Given a stationary process $\IAll$ and a channel $\OAll\big|\IAll$ with output
process $\OAll$, form a joint random variable $Z_t = \JS_t$ over input-output
symbol pairs with alphabet $\boldsymbol{\mathcal{Z}} = \JA$. For example, if
$\IA = \{a,b\}$ and $\OA = \{c,d\}$, then $\JA = \{ac,ad,bc,bd\}$.

\begin{definition}
The process $\olrharp{Z}$ over $\boldsymbol{\mathcal{Z}}$ defines the channel's
\emph{I/O process}.
\end{definition}

\begin{definition}
A stationary channel's \emph{global \eM} is the \eM\ of its I/O process.
\end{definition}

In this setting, the next section asks for a particular decomposition of the
global \eM, when one has selected a portion of $Z_t$ as ``input'' and another
as ``output''. The input process $\IAll$ is then described by the marginal
distribution of the joint process that projects onto ``input'' sequences. The
same also holds for the output process $\OAll$. In this way, the following
results not only provide an analysis for specified input and output processes,
but also an analysis of possible input-to-output mappings embedded in any
process or its \eM. Leveraging this observation and anticipating the sequels,
we also note here that the global \eM\ also provides the proper setting for
posing questions about information storage and flow \emph{within} any given
process.

\section{\ET}

Computational mechanics' fundamental assumption---only prediction
matters---applies as well to channels as to processes. In the case of channels,
though, we wish to predict the channel's future \emph{output} given the
channel's past inputs \emph{and} outputs \emph{and} the channel's future
\emph{input}. This leads to a new causal equivalence relation
$\sim_\epsilon$ over \emph{joint} pasts $\olharplow{z} = \jpast$:
\begin{align}
\jpast \sim_\epsilon \jpast^\prime \iff & \nonumber \\
\label{eq:ETEquivalence}
  \Prob \big( \OFuture \big| \IFuture, & \JPast = \jpast \big) \\
    & \quad = \Prob \big(\OFuture \big| \IFuture, \JPast = \jpast^\prime \big)
	\nonumber
	~.
\end{align}
Compare Eq. (\ref{eq:PredEqReln}) applied to the I/O process. The equivalence
classes of $\sim_\epsilon$ partition the set
$\olharplow{\boldsymbol{\mathcal{Z}}} = \JPastSet$ of all input-output pasts.
These classes are the channel's \emph{causal states}, denoted
$\CausalStateSet$. The \emph{$\epsilon$-map} is a function $\epsilon: \JPastSet
\to \CausalStateSet$ that maps each joint past to its corresponding channel
causal state or, equivalently, to the set of joint pasts to which it is
causally equivalent:
\begin{align*}
\epsilon \big( \jpast \big) = \cs_i =
  \big\{ \jpast^\prime: \jpast \sim_\epsilon \jpast^\prime \big\}
  ~.
\end{align*}

\begin{figure}
\centering
\includegraphics[scale=.75]{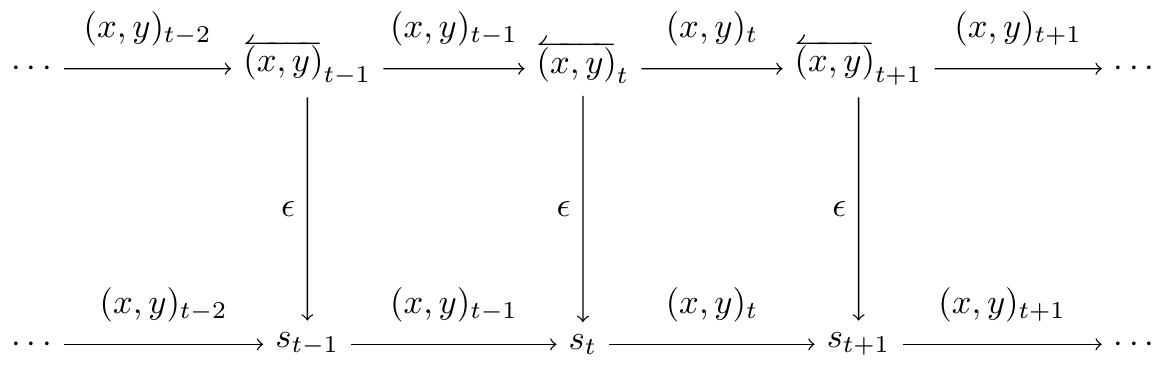}
\caption{\ET\ dynamic induced by the causal states
  $s_t = \cs_i \in \CausalStateSet$ from a channel's joint pasts via the
  $\epsilon$-map: $s_t = \epsilon (\ldots \jsym_{t-2} \jsym_{t-1} )
  \xrightarrow[]{\jsym_t} s_{t+1} = \epsilon (\ldots \jsym_{t-2} \jsym_{t-1} \jsym_t ) $.
  }
\label{fig:channellattice}
\end{figure} 

The dynamic over causal states is again inherited from the implicit dynamic
over joint pasts via the $\epsilon$-map, resulting from appending the joint
symbol $z_t = \jsym_t$, as shown in Fig. \ref{fig:channellattice}. Since state
transitions now depend upon the current input symbol, we specify the dynamic
by an indexed set of \emph{conditional-symbol transition matrices}:
\begin{align*}
\CSTSet \equiv
  \left\{ T^{(\osym|\isym)} \right\}_{\isym \in \IA, \osym \in \OA},
\end{align*}
where $T^{(\osym|\isym)}$ has elements:
\begin{align*}
T_{ij}^{(\osym|\isym)} = \Prob(\CS_1 = \cs_j, \OS_0 = \osym | \CS_0 = \cs_i, \IS_0 = \isym ).
\end{align*}

While the causal states for a stationary process have a \emph{unique}
stationary distribution, each stationary input to a channel can drive its
causal states into a \emph{different} stationary state distribution. The
$\epsilon$-map from joint histories to the channel's causal states can also be
seen as a function that maps a \emph{distribution} over joint histories to a
distribution over the channel's causal states. Since each input history
specifies a particular distribution over \emph{output} histories (via the
channel's conditional word probabilities), it follows that a
\emph{distribution} over input histories also specifies a distribution over
output histories. When this input history distribution is specified via a
particular input process, we obtain a unique distribution over causal states
via its $\epsilon(\cdot)$ function. We write this input-dependent state
distribution $\pi_\IS$ as:
\begin{align*}
\pi_\IS (i) & = \Prob_\IS (\CS_0 =\cs_i) \\
  & = \Prob_\IS \big( \epsilon \big( \JPast \big) = \cs_i \big)
  ~,
\end{align*}
where the subscript $\IS$ indicates that the input process has a specific,
known distribution. The distribution over joint histories is stationary by
assumption here and, since the $\epsilon$-map is time independent, $\pi_\IS$ is
stationary. We therefore refer to $\pi_\IS$ as the (input-dependent)
\emph{stationary distribution}.

When both the input process and channel are stationary and ergodic, we can
calculate this stationary distribution from the input and channel's
causal-state transition matrices using a generalization of the algorithm found
in Ref. \cite{Bill61a}. We save an in-depth discussion of this algorithm for a
sequel.

\begin{definition}
The tuple $(\IA, \OA, \CSSet, \CSTSet)$---consisting of the channel's input and
output alphabets, causal states, and conditional-symbol transition
probabilities, respectively---is the channel's \emph{\eT}.
\end{definition}

Note that in the causal equivalence relation for channels, we condition on the input future $\IFuture$, as a channel is
defined by its output behavior \emph{given} input. Requiring causal 
equivalence for the output future alone (or the joint future for that matter)
requires knowledge of a particular \emph{input process} as well. In
particular, if we \emph{do} have knowledge of the input process we can extend
the standard causal equivalence relation to an equivalence relation
involving joint pasts and joint future morphs, giving us the \emph{global}
(joint) $\eM$ of the preceding section.

\begin{figure}
\centering
\includegraphics[scale=1]{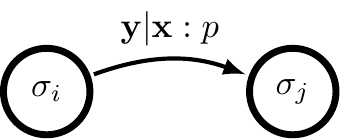}
\caption{Generic transition from \eT\ causal state $\cs_i$ to state $\cs_j$
  while accepting input symbol $\isym$ and emitting symbol $\osym$ with
  probability $p = T_{ij}^{(\osym|\isym)}$.
  }
\label{fig:transducernotation}
\end{figure}

As with \eMs, it is useful to consider channels whose \eTs\ have a finite (or
countable) number of causal states. This restriction again allows us to
represent an \eT\ as a labeled-directed graph. Since transitions between causal
states now depend on inputs as well as outputs, we represent a transition from
state $\cs_i$ to state $\cs_j$ while accepting input symbol $\isym$ and
emitting symbol $\osym$ as a directed edge from node $\cs_i$ to node $\cs_j$
with edge label $\mathbf{y} | \mathbf{x} : p$, where $p = T_{ij}^{(\osym|\isym)}$ is
the edge transition probability. This is illustrated in Fig.
\ref{fig:transducernotation}.

\begin{figure*}
\centering
\includegraphics[scale=0.55]{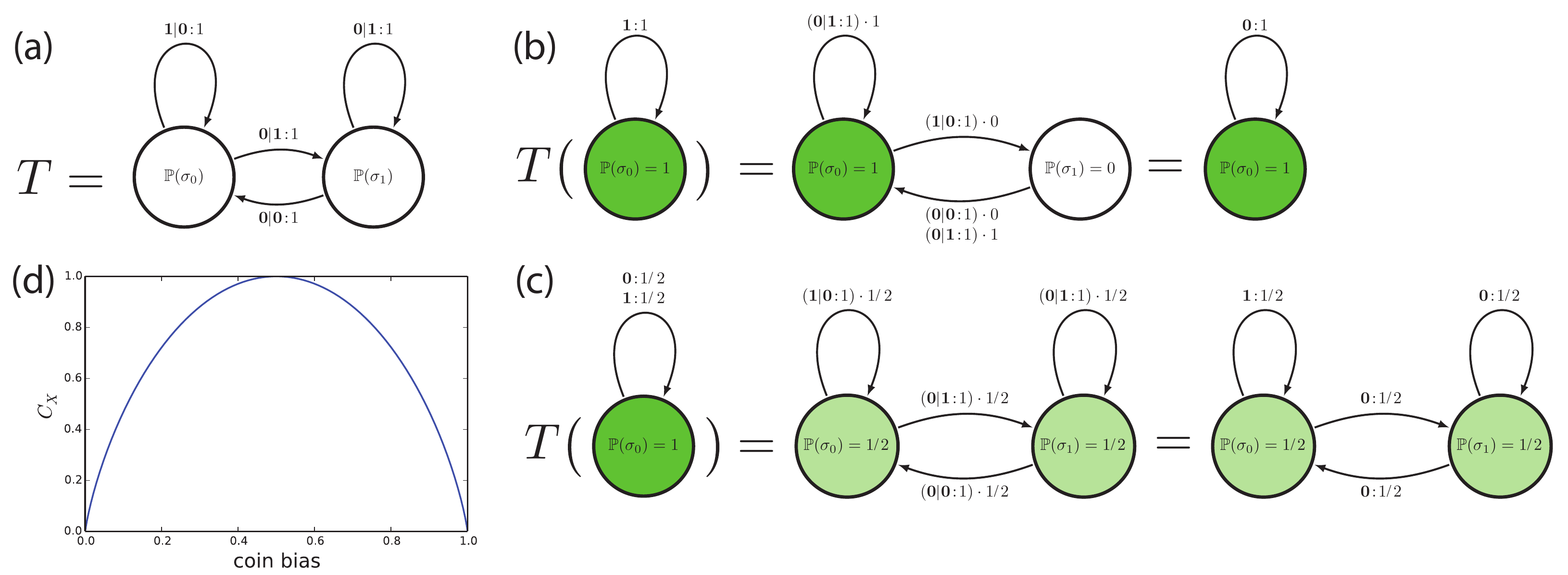}
%\subfigure[]
%  { \includegraphics[scale=0.6]{feedforwardNORTransducer} }
%\subfigure[]
%  { \includegraphics[scale=0.6]{feedforwardNORSimpleInput} }
%\subfigure[]
%  { \includegraphics[scale=0.6]{feedforwardNORFairCoinInput} }
%\subfigure[]
%  { \includegraphics[scale=.6]{varcomplexityFFNOR} }
\caption{A channel's statistical complexity $C_{\IS}$ depends on the input
  process $\IS$:
  (a) An example transducer $T$---the feedforward NOR channel described
  later---driven by two example inputs, showing the intermediate stage
  calculation $T(\IS)$ and the \eM\ of the output process $\OS = T(\IS)$.
  (b) When driven by an input process consisting of all $1$s, the statistical
  complexity vanishes: $C_{\IS} (T) = \log_2 (\Pr(\cs_0) = 1) = 0$ bits.
  (c) When driven by the Fair Coin Process, the distribution over causal states
  is uniform and there is a positive statistical complexity: $C_{\IS} (T) =
  \H(1/2) = 1$ bit.
  (d) Statistical complexity $C_{\IS}$ as a function of the bias of an input
  Biased Coin Process.
  }
\label{fig:TransducerAndComplexity}
\end{figure*} 

\section{Structural Complexity}

To monitor the degree of structuredness in an \eT, we use the Shannon
information captured by its causal states, paralleling the definition of
\eM\ statistical complexity. The stationary distribution over \eT\ states allows one to define an \eT's
\emph{input-dependent} statistical complexity:
\begin{align*}
C_{\IS} = \H[\pi_\IS]
  ~.
\end{align*}
(Note that $\IS$ replaces the previously subscripted measure $\mu$ in
\eM\ statistical complexity to specify the now-relevant measure.) While
quantifying structural complexity, $C_{\IS}$'s dependence on input requires
a new interpretation. Some processes drive a transducer into simple,
compressible behavior, while others will lead to complex behavior. Figure
\ref{fig:TransducerAndComplexity} illustrates this.

Input dependence can be removed by following the standard definition of channel
capacity \cite{Cove06a}, giving a single number characterizing an \eT. We take
the supremum of the statistical complexity over input processes. This gives an
upper bound on \eT\ complexity---the \emph{channel complexity}:
\begin{align*}
\overline{\Cmu} = \sup_{\IS} C_{\IS}
  ~,
\end{align*}
where the maximizing input measure $\mu$ is implicitly defined. Note that not
all transducers can be driven to a uniform distribution over states. Thus,
recalling that uniform distributions maximize Shannon entropy, in general
$\overline{\Cmu} \leq C_0 \equiv \log_2 |\CausalStateSet|$---the
\emph{topological state complexity}.

\section{Reproducing a Channel}
To establish that the \eT\ is an exact presentation of the causal channel it
models, we must show that it reproduces the channel's conditional word
probabilities. We first establish some needed notation.

Recall that the $\epsilon$-map takes a distribution over joint
histories to a distribution over the \eT's causal states and that a particular
input history defines a distribution over a channel's output history. It
follows that a particular input history defines a distribution over joint
histories and, therefore, also defines a distribution over an \eT's causal
states via the $\epsilon$-map. We call this distribution $\tau$:
\begin{align*}
\tau (i) & = \Prob_\IS \big(\CS_0 =\cs_i \big| \IPast_0 = \ipast \big) \\
  & = \Prob_\IS \big( \epsilon \big( \JPast_0 \big) = \cs_i \big| \IPast_0 = \ipast
  \big)
  ~.
\end{align*}
Note that while the $\epsilon$-map takes a particular \emph{joint} history to a
unique state, the distribution over states induced by a particular \emph{input}
history need not be concentrated on a single state.

When the input process is known and stationary, the \eT's stationary
distribution $\pi_\IS$ can be determined. This provides a starting distribution
for the \eT, which is updated by the \eT's symbol transition matrices as each
input symbol is observed and each output symbol is generated. We can therefore
calculate the final state distribution $\Prob(\CS_{L}=\cs_m | \isym_{0:L},
\CS_0 \sim \pi_\IS)$ that results from starting states in distribution
$\pi_\IS$ and observing finite input word $\isym_{0:L}$:
\begin{align*}
\Prob(&\CS_{L}=\cs_m | \isym_{0:L}, \CS_0 \sim \pi_\IS) \\
  & = \sum_{\osym_{0:L}} ~\sum_{i,j,k,\cdots,l}
  \pi_\IS(i) T_{ij}^{(\osym_0|\isym_0)} T_{jk}^{(\osym_1|\isym_1)} \cdots
  T_{lm}^{(\osym_{L-1}|\isym_{L-1})}
  ~.
\end{align*}

We then obtain $\tau$ from the \eT\ by shifting this distribution by $L$ and taking the limit as $L \to \infty$:
\begin{align*}
\tau (i) & = \Prob_\IS \big(\CS_0 =\cs_i \big| \IPast_0 = \ipast \big) \\
& = \lim_{L \to \infty} \Prob \big(\CS_{0}=\cs_i \big| \isym_{-L:0}, \CS_{-L} \sim \pi_\IS \big)
  ~.
\end{align*}

We can now establish that the \eT\ is an exact presentation of the causal channel that it models.

\begin{Prop}[Presentation]
A causal channel's \eT\ exactly (and only) reproduces the channel's conditional word probabilities $\Prob (\OS_{0:L}  | \IPast_{L} )$.
\end{Prop}

\begin{proof}
Recall that a \emph{causal} channel's output words do not depend on any inputs
occurring \emph{after} the output word. Therefore, the \eT\ must reproduce all
of a channel's conditional word probabilities of the form
$\Prob (\osym_{0:L}  | \ipast_{L} )$. As discussed above, an input history
induces a distribution $\tau$ over the \eT's causal states. So, we calculate
the word probabilities directly via repeated application of the \eT's symbol
transition matrices:
\begin{align}
\label{eq:ETConditionalWordDist}
\Prob (\osym_{0:L} & | \ipast_{L} ) \\
  & = \sum_{i,j,k,\cdots,l,m}
  \tau (i) T_{ij}^{(\osym_0|\isym_0)} T_{jk}^{(\osym_1|\isym_1)} \cdots
  T_{lm}^{(\osym_{L-1}|\isym_{L-1})}
  \nonumber
  ~.
\end{align}
\end{proof}

In fact, we can transduce a \emph{finite} input word even when the channel's
behavior depends upon arbitrarily long input histories. By simply starting
the \eT\ in its stationary distribution $\pi_\IS$, we have:
\begin{align*}
\Prob (\osym_{0:L} & | \isym_{0:L} ) \\
  & = \sum_{i,j,k,\cdots,l,m}
  \pi_\IS (i) T_{ij}^{(\osym_0|\isym_0)} T_{jk}^{(\osym_1|\isym_1)} \cdots
  T_{lm}^{(\osym_{L-1}|\isym_{L-1})}
  ~.
\end{align*}

In either case, we can multiply these conditional word probabilities by the
input's word probabilities to obtain joint word probabilities:
\begin{align*} 
\Prob ( \jsym_{0:L} ) = \Prob (\osym_{0:L} | \isym_{0:L} ) \Prob(\isym_{0:L})
  ~.
\end{align*}
Summing over input words then gives output word probabilities:
\begin{align*}
\Prob ( \osym_{0:L} ) = \sum_{\isym_{0:L} \in \IA^L} \Prob ( \jsym_{0:L} )
  ~.
\end{align*}

We can also start the \eT\ with other, arbitrary state distributions when
certain initial behavior is desired or if the current internal configuration of
the \eT\ is known. This can be very useful in practice, but the resulting
generated behavior is no longer guaranteed to be stationary or to match the
original channel's behavior. One use of arbitrary state distributions is
\emph{real-time} transduction of symbols, where a state distribution $\nu$ is
repeatedly updated each time-step after a single input symbol $\isym_t$ is
transduced to an output symbol $\osym_t$:
\begin{align*}
\Prob \big(\CS_{t+1}=\cs_i \big| \jsym_{t}, \CS_{t} \sim \nu \big)
  = \sum_{i,j} \nu (i) T_{ij}^{(\osym_t|\isym_t)}
  ~.
\end{align*}

\section{Optimality}

We now establish that the \eT\ is a channel's unique, maximally predictive,
minimal statistical complexity unifilar presentation. Other properties,
analogous to those of the \eM, are also developed.  Several proofs parallel
those in Refs.  \cite{Crut98d,Crut98d,Shal98a}, but are extended from \eMs\ to
the \eT.  For this initial development, we also adopt a caveat from there
concerning the use of infinite pasts and futures. For example, the
semi-infinite pasts' entropy $\H[\OPast]$ is typically infinite. And so, to
properly use such quantities, one first introduces finite-length chains (e.g.,
$\H[\OS_{0:L}]$) and at the end of an argument one takes infinite-length limits,
as appropriate. Here, as previously, we do not include these extra steps,
unless there is subtlety that requires attention using finite-length chains.

\begin{Prop}[Causal States Proxy the Past]
\label{proxy}
When conditioned on causal states, the future output given input is
independent of past input and past output: 
\begin{align*}
\Prob \big( \OFuture_0 \big| \IFuture_0, \JPast_0, \CS_0 \big)
  = \Prob \big( \OFuture_0 \big| \IFuture_0, \CS_0 \big)
  ~.
\end{align*}
\end{Prop}

\begin{proof}
By construction, the causal states have the same future morphs as their corresponding pasts:
\begin{align*}
\Prob \big(\OFuture_0 \big| \IFuture_0, \JPast_0 \big)
	= \Prob \big(\OFuture_0 \big| \IFuture_0, \CS_0 \big) ~.
\end{align*}
Since the causal states are a function of the
past---$\CS_0 = \epsilon \left(\JPast_0\right)$---we also have that:
\begin{align*}
\Prob \big(\OFuture_0 \big| \IFuture_0, \JPast_0, \CS_0 \big)
  = \Prob \big(\OFuture_0 \big| \IFuture_0, \JPast_0  \big) ~.
\end{align*}
Combining these two equalities gives the result.
\end{proof}

In other words, when predicting a channel's future behavior from its past behavior, it suffices to use the causal states instead.

\begin{Prop}[Causal Shielding]
Past output $\OPast_0$ and future output $\OFuture_0$ given future input
$\IFuture_0$ are independent given the current causal state $\CS_0$:
\begin{align*}
\Prob \big( \OAll_0 \big| \IFuture_0, \CS_0 \big) 
  = \Prob \big( \OPast_0 \big| \IFuture_0, \CS_0 \big)
  \Prob \big( \OFuture_0 \big| \IFuture_0, \CS_0 \big)
  ~.
\end{align*}
\end{Prop}

\begin{proof}
We directly calculate:
\begin{align*}
\Prob \big( \OAll_0 \big| \IFuture_0, \CS_0 \big)
  &= \Prob \big( \OPast_0 \big| \IFuture_0, \CS_0 \big)
  \Prob \big( \OFuture_0 \big| \IFuture_0, \OPast, \CS_0 \big)\\
  & = \Prob \big( \OPast_0 \big| \IFuture_0, \CS_0 \big)
  \Prob \big( \OFuture_0 | \IFuture_0, \CS_0 \big)
  ~.
\end{align*}
Where the second equality follows from applying Prop. \ref{proxy} to the second
factor.
\end{proof}

In the following, depending on use, we refer to either of the previous
propositions as \emph{causal shielding}.

\begin{Prop}[Joint Unifilarity]
The current causal state $\CS_0$ and current input-output symbol pair
$\JS_0$ uniquely determine the next causal state. In this case:
\begin{align*}
\H[\CS_1 | \JS_0,\CS_0] = 0 ~.
\end{align*}
\end{Prop}

\begin{proof}
If two pasts are causally equivalent, then either (i) appending a new
symbol pair $\jsym$ to both pasts results in two \emph{new} pasts that
are also causally equivalent or (ii) such a symbol pair is never observed
when in $\CS_0$. We
must show that the two new pasts have the same future morph:
\begin{align*}
& \jpast \sim_\epsilon \jpast^\prime \\
  & \implies \Prob \big(\OFuture_{1} \big| \IFuture_{1},\jpast(a,b) \big)
   = \Prob \big(\OFuture_{1} \big| \IFuture_{1},\jpast'(a,b) \big)
  ~,
\end{align*}
where we have $a \in \IA$ and $b \in \OA$ and
$\jpast (a,b) = ( \ipast a, \opast b )$, and the futures $\OFuture_{1}$ and $\IFuture_{1}$ denote those immediately
following the associated conditioning pasts $\jpast(a,b)$ and $\jpast^\prime(a,b)$, respectively.
Or, we must show that the input-output pair $\jsym$ is forbidden.

First, let $\jpast \sim_\epsilon \jpast'$. Since causal equivalence applies for \emph{any} joint future, it applies to the particular future beginning with symbol pair $(a,b)$:
\begin{align*}
\Prob \big(b\OFuture_{1} \big| a\IFuture_{1}, \jpast \big)
  = \Prob \big(b\OFuture_{1} \big| a\IFuture_{1}, \jpast^\prime \big)
  ~.
\end{align*}
Factoring:
\begin{align*}
\Prob \big(b\OFuture_{1} \big| \cdot \big)
	= \Prob \big(\OFuture_{1} \big| \OS_0 = b, \cdot \big)
	\Prob \big(\OS_0 = b \big| \cdot \big) 
\end{align*}
gives:
\begin{align*}
& \Prob \big(\OFuture_{1} \big| \OS_0 = b, a\IFuture_{1}, \jpast \big)
	\Prob \big(\OS_0 = b \big|  a\IFuture_{1}, \jpast \big) \\
  & \quad = \Prob \big(\OFuture_{1} \big| \OS_0 = b, a\IFuture_{1}, \jpast^\prime \big) 
	  \Prob \big(\OS_0 = b \big| a\IFuture_{1}, \jpast^\prime \big)
  ~.
\end{align*}
The second factors on both sides are equal by causal equivalence. So, there are
two cases: These factors either vanish or they do not. If they are positive,
then we have:
\begin{align*}
\Prob \big(\OFuture_{1} \big| \OS_0 = b, a\IFuture_{1}, \jpast \big)
	= \Prob \big(\OFuture_{1} \big| \OS_0 = b, a\IFuture_{1}, \jpast^\prime \big)
	~.
\end{align*}
Rewriting the conditional variables with the symbol pair $(a,b)$ attached to
the joint past then gives the first part of the result:
\begin{align*}
\Prob \big(\OFuture_{1} \big| \IFuture_{1},\jpast (a,b) \big)
  = \Prob \big(\OFuture_{1} \big| \IFuture_{1},\jpast^\prime(a,b) \big)
  ~.
\end{align*}
In the other case, when the factors vanish, we have:
\begin{align*}
\Prob \big(\OS_0 = b \big|  a\IFuture_{1}, \jpast \big)
   = \Prob \big(\OS_0 = b \big| a\IFuture_{1}, \jpast^\prime \big) 
   = 0 ~.
\end{align*}
This implies that:
\begin{align*}
  \Prob \big(\OS_0 = b \big| \IS_0=a, \jpast \big)
   & = \Prob \big(\OS_0 = b \big| \IS_0 = a, \jpast^\prime \big) \\
   & = 0
  ~.
\end{align*}
In other words, $\OS_0=b$ is never observed following either past,
given $\IS_0=a$. That is, $(a,b)$ is forbidden.

It then follows that:
\begin{align*}
\H[S_1|S_0,\JS_0] = 0 ~,
\end{align*}
which is equivalent to joint unifilarity when there is a finite number of causal states.
\end{proof}

Unifilarity guarantees that once we know the process is in a particular causal
state---we are ``synchronized'' \cite{Jame10a}---we do not lose synchronization
over time. This is an important property when using causal states to simulate
or predict a system's behavior. Using presentations that are nonunifilar,
typically it is necessary to keep track of a \emph{distribution} over states.

Unifilarity is also a useful property to have when attempting to infer an \eT\
from data. Inference of nonunifilar transducers can be challenging, partly
due to the existence of multiple possible state paths given a particular start
state. Unifilar transducers avoid this problem, effectively reducing the
difficulty to that of inferring a Markov chain from data \cite{Stre13a}.
Finally, unifilarity plays a key role, as a sequel shows, in calculating
channel information quantities.

The next theorem shows that \eTs\ are input-dependent hidden Markov models.

\begin{Prop}[Markovity]
A channel's causal states satisfy the conditional Markov property:
\begin{align*}
\Prob \big(\CS_t \big| \IS_{t-1}, \SPast_t \big)
  = \Prob(\CS_t|\IS_{t-1},\CS_{t-1}) ~.
\end{align*}
\end{Prop}

\begin{proof}
Since the causal-state transitions are unifilar, there is a well defined set
of output symbols $\mathcal{Z} \subseteq \OA$ that causes a transition from
state $\cs_j$ to state $\cs_k$. We therefore have:
\begin{align*}
\Prob \big(\CS_t=\cs_k & \big| \IS_{t-1}, \CS_{t-1}=\cs_j, \SPast_{t-1} \big) \\
  & = \Prob \big(\OS_{t-1} \in \mathcal{Z} \big| \IS_{t-1}, \CS_{t-1}=\cs_j ,
  \SPast_{t-1} \big)
  ~.
\end{align*}
Causal shielding applies to finite futures as well as infinite. This, combined
with the observation that $\SPast_{t-1}$ is purely a function of the past, allows
us to use $\CS_{t-1}$ to causally shield $\OS_{t-1}$ from $\SPast_{t-1}$, giving:
\begin{align*}
\Prob \big(\OS_{t-1} \in \mathcal{Z} & \big|\IS_{t-1}, \CS_{t-1}=\cs_j,
\SPast_{t-1} \big) \\
	& = \Prob(\OS_{t-1} \in \mathcal{Z}|\IS_{t-1}, \CS_{t-1}=\cs_j) \\
	& = \Prob(\CS_t=\cs_k|\IS_{t-1}, \CS_{t-1}=\cs_j)
  ~.
\end{align*}
The final equality is again possible due to unifilarity.
\end{proof}

The following theorem shows that the causal states store as much information as
possible (from the past) about a channel's future behavior---a desirable
property for any predictive model.

\begin{definition}
The \emph{prescience} of a set $\ASSet$ of rival states---equivalence classes
of an alternative partition of pasts---reflects how well the rival states predict a channel's future behavior. Quantitatively, this is monitored by the amount of information
they share with future output, given future input:
\begin{align*}
\I \big[ \AS ~; \OFuture \big| \IFuture \big] ~,
\end{align*}
where $\AS$ is the associated rival-state random variable.
\end{definition}

Note that it is sometimes simpler to prove statements about conditional entropy
than it is for mutual information. Due to this, we will transform statements
about prescience into statements about uncertainty in prediction in several
proofs that follow. Specifically, we will make use of the identity:
\begin{align*}
\I \big[ \AS_0 ~; \OFuture_0 \big| \IFuture_0 \big]
  & = \lim_{L \to \infty} \I \big[ \AS_0 ~; \OS_{0:L} \big| \IFuture_0 \big] \\
  & = \lim_{L \to \infty} \left( \H \big[\OS_{0:L} \big| \IFuture_0 \big]
  - \H \big[\OS_{0:L} \big| \IFuture_0, \AS_0 \big] \right)
  ,
\end{align*}
where $\H \big[\OS_{0:L} \big| \IFuture_0, \AS_0 \big]$ is the finite-future
\emph{prediction uncertainty}. Note that the infinite-future prediction
uncertainty $\H \big[\OFuture_0 \big| \IFuture_0, \AS_0 \big]$ typically will be
infinite, but rewriting the prescience in terms of the limit of finite-future
prediction uncertainty allows us to continue to work with finite
quantities.

\begin{theorem}[Maximal Prescience]
Among all rival partitions $\ASSet$ of joint pasts, the causal states have
maximal prescience and they are as prescient as pasts:
\begin{align*}
\I \big[\CS ~; \OFuture \big| \IFuture \big]
     & = \I \big[\JPast ~; \OFuture \big| \IFuture \big]  \\
	 & \geq \I[\AS ~; \OFuture| \IFuture ] ~.
\end{align*}
\end{theorem}

\begin{proof}
We will prove the equivalent statement that the causal states 
\emph{minimize} finite-future prediction \emph{uncertainty} for 
futures of any length $L$ and have the same finite-future prediction
uncertainty as pasts; i.e., that for all $L$:
\begin{align*}
\H \big[\OS_{0:L} \big| \IFuture_0, \CS_0 \big]
& = \H \big[\OS_{0:L} \big| \IFuture_0, \JPast_0 \big]  \\
 & \leq \H \big[\OS_{0:L} \big| \IFuture_0, \AS_0 \big]
 ~,
\end{align*}
Like the causal states, rival states---equivalence classes of an alternative
partition $\ASSet$ of pasts---are a function of the past:
\begin{align*}
\AS = \eta \big( \JPast \big)
  ~.
\end{align*}
By the Data Processing Inequality \cite{Cove06a}, we have:
\begin{align*}
\H \big[\OS_{0:L} \big| \IFuture_0, \JPast_0 \big]
	  \leq \H \big[\OS_{0:L} \big| \IFuture_0, \AS_0 \big] ~.
\end{align*}
The causal states share future morphs with their
corresponding pasts. By simple marginalization, the same is true for 
finite-future morphs:
\begin{align*}
\Prob \big(\OS_{0:L} & \big| \IFuture_0, \CS_0 \big)
  = \Prob \big(\OS_{0:L} \big| \IFuture_0, \JPast_0  \big) \\
  & \implies \H \big[\OS_{0:L} \big| \IFuture_0, \CS_0 \big]
  = \H \big[\OS_{0:L} \big| \IFuture_0, \JPast_0 \big]
  ~.
\end{align*}
\end{proof}

Note that this proof also shows that the causal states are maximally prescient for all length-$L$ futures: $\I [\CS ~; \OS_{0:L} | \IS_{0:L} ] = \I [\JPast_0 ~; \OS_{0:L} \big| \IS_{0:L}]$. 

The causal equivalence relation can also be applied directly to a
channel that anticipates its future inputs, but the resulting \eT\ has outputs
that depend not only upon the current state and input symbol, but some set of
\emph{future} input symbols. If the set is finite, then one uses the previous
construction that transforms finite anticipation into additional transducer memory (statistical complexity). In this case, the causal states still capture all of the
information from the past needed for prediction. That is, they have maximal
prescience. Any additional future input dependence, though, must be encoded in
the machine's transitions.

\begin{definition}
A \emph{prescient rival} is an indexed set $\PSSet$ of states (with elements
$\ps_i$ and random variable $\PS$) that is as predictive as any past:
\begin{align*}
\I \big[\PS ~; \OFuture \big| \IFuture \big]
     & = \I \big[\JPast ~; \OFuture \big| \IFuture \big]
  ~.
\end{align*}
\end{definition}

\begin{Lem}[Refinement]
The partition of a prescient rival $\PSSet$ is a refinement (almost everywhere)
of the causal-state partition of the joint input-output pasts.
\label{lem:Refine}
\end{Lem}

\begin{proof}
Since the causal states' future morphs include every possible future morph of a
channel, we can always express a prescient rival's future morph as a (convex)
combination of the causal states' future morphs. This allows us to rewrite the
entropy over a prescient rival (finite-length) future morph as:
\begin{align}
\H \big[ \OS_{0:L} \big| \ifuture, & \ps_k \big]
	= \H \big[ \Prob \big(\OS_{0:L} \big| \ifuture, \ps_k \big) \big]
	\nonumber \\
	& = \H \left[ \sum_j
		\Prob \big(\OS_{0:L} \big| \ifuture, \cs_j \big) \Prob(\cs_j | \ps_k)
		\right]
	~. 
\label{Hconvex1}
\end{align}
Since entropy is convex, we also have:
\begin{align}
\H & \left[ \sum_j
	\Prob \big(\OS_{0:L} \big| \ifuture, \cs_j \big)
	\Prob ( \cs_j | \ps_k) \right]
  \nonumber \\
  & \quad\quad \geq \sum_j \Prob(\cs_j | \ps_k)
  \H \big[\OS_{0:L} \big| \ifuture, \cs_j \big]
  ~ .
\label{Hconvex2}
\end{align}
Therefore:
\begin{align*}
\H \big[\OS_{0:L} \big| \IFuture_0,\PS_0 \big]
  & = \sum_k \Prob(\ps_k) \H \big[\OS_{0:L} \big| \IFuture_0,\ps_k \big] \\
  & \geq \sum_k \Prob(\ps_k) \sum_j
  	\Prob(\cs_j | \ps_k) \H \big[\OS_{0:L}  \big| \IFuture_0, \cs_j \big] \\
  & = \sum_{j,k} \Prob(\cs_j,\ps_k) \H \big[\OS_{0:L} \big|\IFuture_0,\cs_j \big] \\
  & = \sum_j \Prob(\cs_j) \H \big[\OS_{0:L}  \big| \IFuture_0, \cs_j \big] \\
  & = \H \big[\OS_{0:L}  \big| \IFuture_0, \CS_0 \big]
  ~ ,
\end{align*}
where the inequality follows from Eqs. (\ref{Hconvex1}) and (\ref{Hconvex2}).
Since the rival states $\PS$ are prescient, we know that equality must be
attained in this inequality for each $L$. Equality is only possible when
$\Prob(\cs_j | \ps_k) = 1$ for exactly one value of $j$ and vanishes for
every other $j$. That is, if a rival state is prescient, it is contained
entirely within a single causal state, aside from a set of measure zero.
Thus, the partition of the prescient rival states is a refinement of the
causal-state partition almost everywhere.
\end{proof}

\begin{theorem}[Minimality]
For any given input process $\IS$, causal states have the minimal conditional statistical complexity among all
prescient rival partitions $\PS$:
\begin{align*}
C_{\IS}(\CS) \leq C_{\IS} \big(\PS \big) ~.
\end{align*}
\end{theorem}

\begin{proof}
Since a prescient rival partition is a refinement almost everywhere, there
exists a function $f$ defined almost everywhere that maps each rival state
to the causal state that (almost everywhere) contains it:
\begin{align*}
f( \ps_i ) = \cs_j ~.
\end{align*}
Then, we have:
\begin{align*}
\H_\IS \big[\PS \big] & \geq \H_\IS \big[ f\big( \PS \big) \big] \\
            & = \H_\IS [\CS] ~.
\end{align*}
\end{proof}

\begin{Cor}
Causal states minimize the channel complexity $\overline{\Cmu}$.
\end{Cor}

\begin{proof}
Immediate from the preceding theorem.
\end{proof}

In words, we established the fact that the causal states store all of the
information contained in the past that is necessary for predicting a channel's future behavior
and as little of the remaining information ``overhead'' contained in the
past as possible. Given an input process, \eT\ causal states maximize
$\I[\OFuture;\CS|\IFuture]$ while minimizing $\I[\JPast;\CS]$.

The final optimality theorem shows that any states which have these properties are in fact the causal states.

\begin{theorem}[Uniqueness]
The \eT\ is the unique prescient, minimal partition of pasts. If $C_{\IS} \big(\PS \big) = C_{\IS}(\CS)$ for every input process $\IAll$, then the
corresponding states $\PS$ and $\CS$ are isomorphic to one another almost
everywhere. And, their equivalence relations $\sim_\eta$ and $\sim_\epsilon$
are the same almost everywhere.
\end{theorem}

\begin{proof}
Again, the Refinement Lemma (Lemma \ref{lem:Refine}) says that $\CS=f \big(\PS
\big)$ almost everywhere. It therefore follows that $\H_\IS  \big[\CS \big|\PS
\big] = 0$. Moreover, by assumption $\H_\IS [\CS] = \H_\IS  \big[\PS \big]$.
Combining these with the symmetry of mutual information gives:
\begin{align*}
\I_\IS  \big[\CS;\PS \big] & = \I_\IS  \big[\PS;\CS \big] \\
\H_\IS [\CS] - \H_\IS  \big[\CS \big|\PS \big]
    & = \H_\IS  \big[\PS \big] - \H_\IS \big[\PS \big|\CS \big] \\
\H_\IS [\CS] - 0 & = \H_\IS [\CS] - \H_\IS  \big[\PS \big|\CS \big] \\
\H_\IS  \big[\PS \big|\CS \big] &= 0 ~. 
\end{align*}
The latter holds if and only if there is a function $g$ such that $\PS = g(\CS)$
almost everywhere. By construction $g$ is the inverse $f^{-1}$ of $f$ almost everywhere. We have that $f \circ \eta = \epsilon$ and
$f^{-1} \circ \epsilon = \eta$.

Finally, the two equivalence relations $\sim_\epsilon$ and $\sim_\eta$ are the same almost everywhere:
\begin{align*}
\jpast & \sim_\epsilon \jpast^\prime \\
  \implies \epsilon \Big( \jpast \Big) & = \epsilon \Big( \jpast^\prime \Big) \\
  \implies f^{-1} \circ \epsilon \Big( \jpast \Big)
  	& = f^{-1} \circ \epsilon \Big( \jpast^\prime \Big) \\
  \implies \eta \Big(\jpast \Big)
  	& = \eta \Big( \jpast^\prime \Big) \\
  \implies \jpast & \sim_\eta
  	\jpast^\prime
  ~,
\end{align*}
and
\begin{align*}
\jpast & \sim_\eta \jpast^\prime \\
  \implies \eta \Big( \jpast \Big) & = \eta \Big( \jpast^\prime \Big) \\
  \implies f \circ \eta \Big( \jpast \Big)
  	& = f \circ \eta \Big( \jpast^\prime \Big) \\
  \implies \epsilon \Big( \jpast \Big)
  	& = \epsilon \Big( \jpast^\prime \Big) \\
  \implies \jpast & \sim_\epsilon \jpast^\prime
  ~.
\end{align*}
\end{proof}

\ET\ uniqueness means that $C_\IS$ is \emph{the} conditional complexity of a
channel and, therefore, justifies calling $\overline{\Cmu}$ \emph{the} channel
complexity.

\section{Global \EM\ versus \ET}

Given a particular joint process or its global $\eM$, it is possible (provided
that the input process satisfies certain requirements) to construct the $\eT$ that maps
input $\IAll$ to output $\OAll$, such that $\JAll = \big( \IAll,f(\IAll)
\big)$, where $f$ is the transducer and $f \big(\IAll \big)$ designates the
output of the transducer, given input process $\IAll$. Sequels address the
relationship between a joint process' global $\eM$ and corresponding $\eT$ at
both the process (channel) level and at the automata ($\eM$ and $\eT$) level.
There, we provide algorithms for ``conditionalizing'' a joint process or $\eM$
to obtain the corresponding channel or $\eT$, as well as algorithms for
obtaining input or output marginals, applying an $\eT$ to an input $\eM$,
composing multiple $\eT$s, and inverting an invertible $\eT$.

Note that the ability to construct an $\eT$ from a joint process can be useful
when attempting to  infer an $\eT$ from data, as such data will typically come
from a system driven by some \emph{particular} (possibly controllable) input;
i.e., the data is a sample of a \emph{joint process}. 

\section{History \ET\ versus Generator \ET}

The preceding focused on the \emph{history} specification of an \eT, where a
machine is obtained by partitioning a channel's histories (joint pasts). We can
also consider the \emph{generator} specification, where we instead \emph{start}
with a machine that produces a stationary, ergodic channel. Taking this
perspective, an \eT\ is an input-dependent, strongly connected, aperiodic
hidden Markov model with unifilar transitions and probabilistically distinct
states. The history and generator specifications of an \eT\ are likely
equivalent---as they are with \eMs\ \cite{Trav11a}---but we leave such a proof
to future work.

\section{Examples Revisited}

With the \eT\ defined, we can revisit the example channels examined above.
We display channel structure via its \eT's state transition diagram, as well
as a wordmap that now colors each joint history based on its corresponding
causal state. Input and output history projections are also shown. Histories 
that are not mapped to a single causal state are colored black. Recall that 
causal states partition \emph{joint} histories, so input or output histories 
alone need not correspond to a unique causal state. We will discuss Markov
orders for these channels, but we leave it to the reader to construct an 
exhaustive list of Markov orders for each channel.

\begin{figure}
 \centering
 \includegraphics[scale=1]{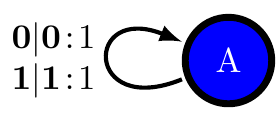}\\
 \includegraphics[scale=.5]{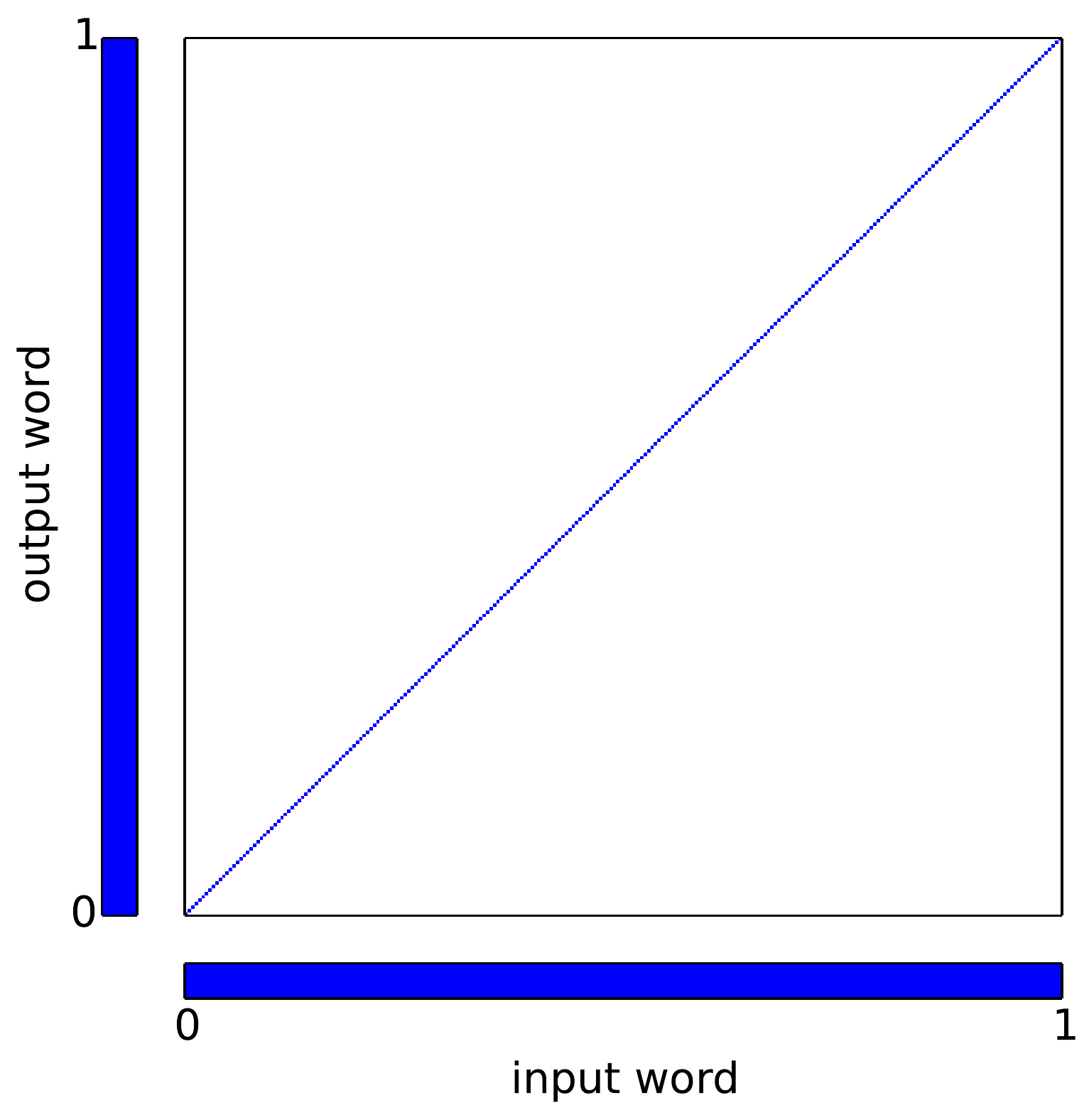}
\caption{Identity Channel: \eT\ and causal-state colored wordmap. See text
  for explanation.
  }
\label{fig:identitystate}
\end{figure} 

\begin{figure}
  \centering
\includegraphics[scale=1]{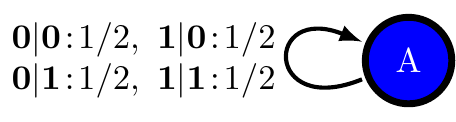}\\
 \includegraphics[scale=.5]{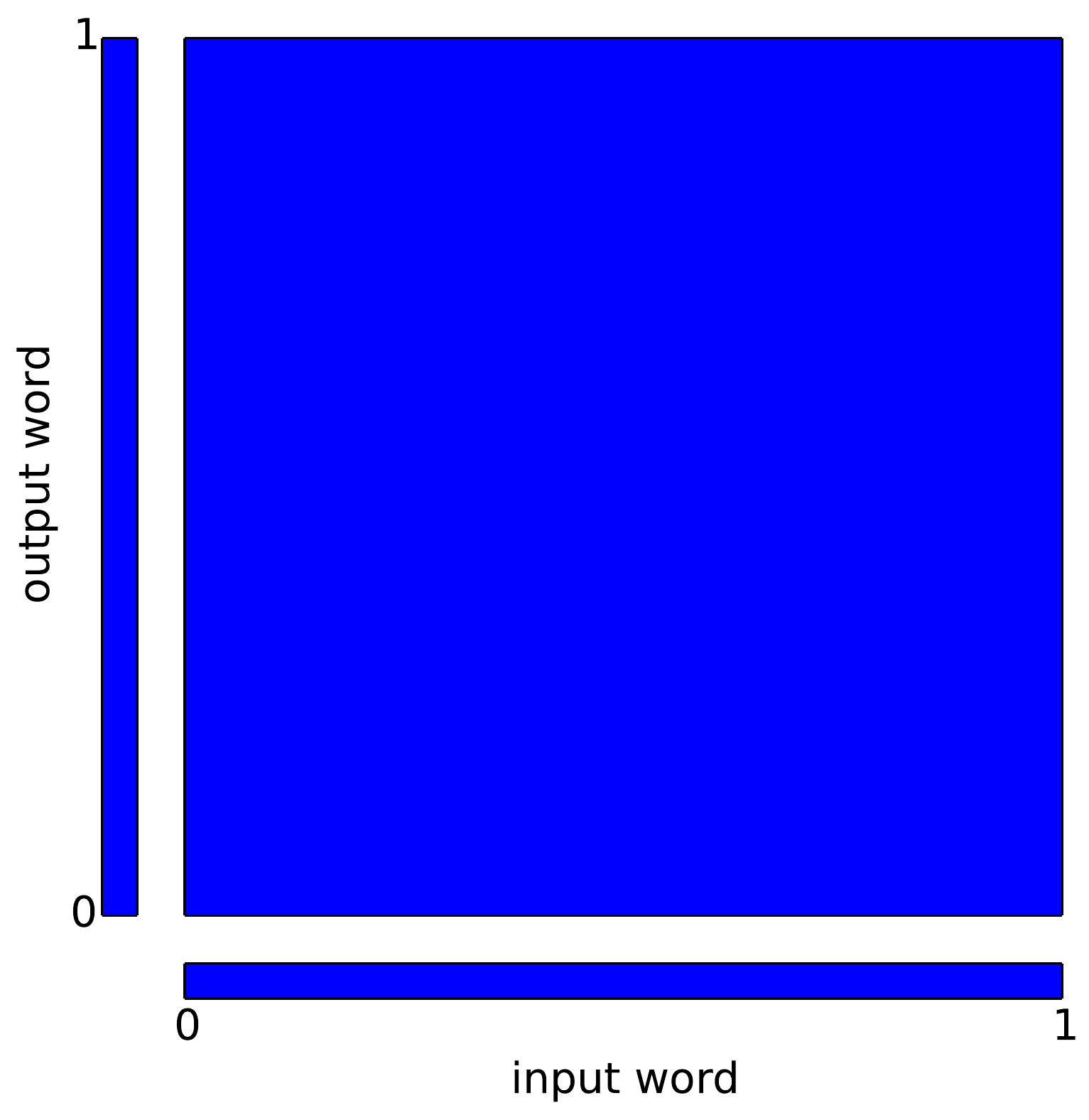}
\caption{All-Is-Fair Channel: \eT\ and causal-state colored wordmap.}
\label{fig:alltofairstate}
\end{figure} 

\begin{figure}
  \centering
\includegraphics[scale=1]{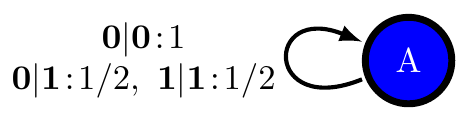}\\
\includegraphics[scale=.5]{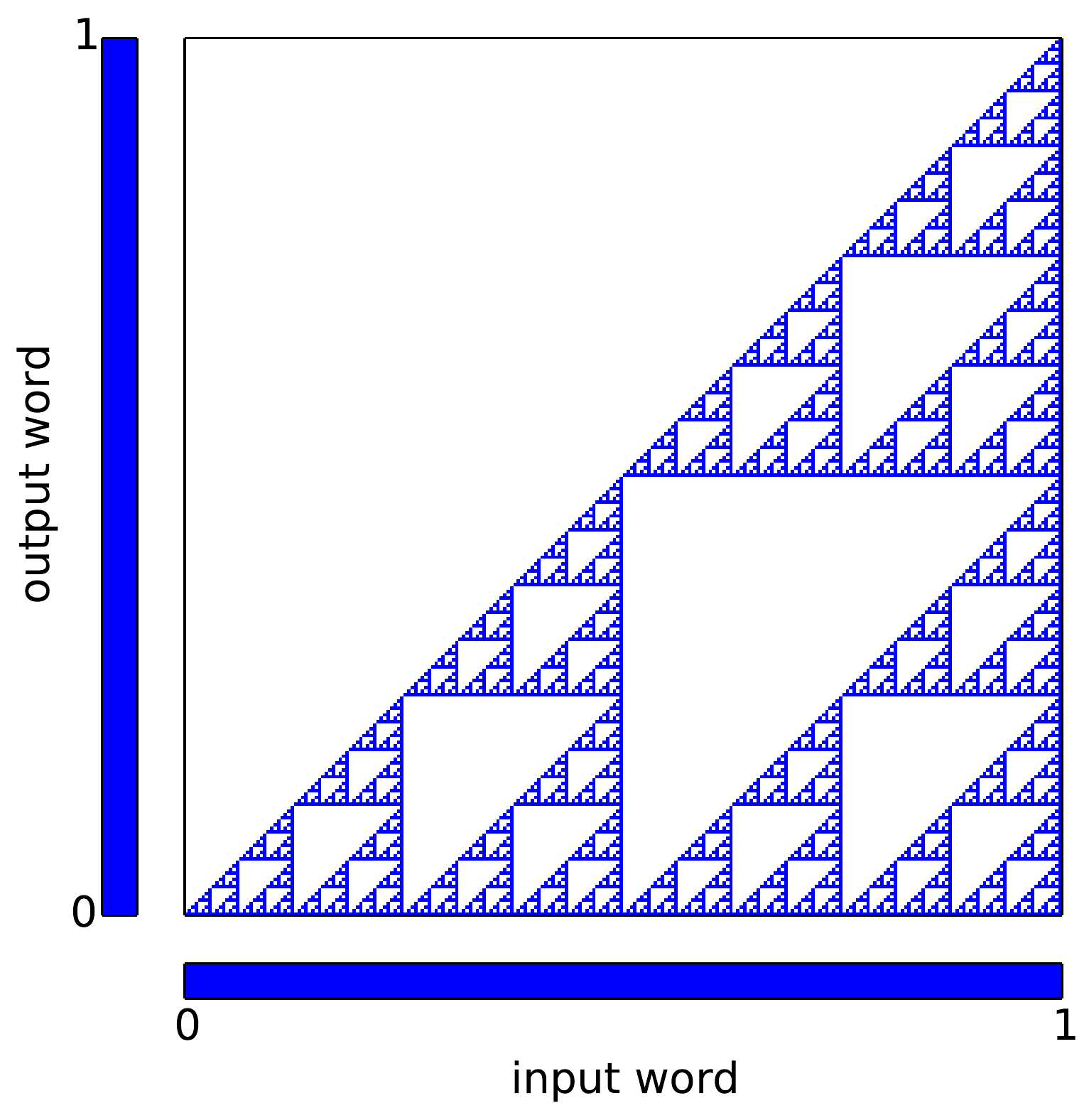}
\caption{Z Channel: \eT\ and causal-state colored wordmap.}
\label{fig:zchannelstate}
\end{figure} 

Since the Identity (Fig. \ref{fig:identitystate}), All is Fair (Fig.
\ref{fig:alltofairstate}), and Z (Fig. \ref{fig:zchannelstate}) Channels are
memoryless, their behavior does not depend on the past. As a result, there
is a single causal state containing every past and so their wordmaps are
monochromatic. Since they have a single causal state,
$C_X = \overline{\Cmu} = 0$.

\begin{figure}
  \centering
\includegraphics[scale=1]{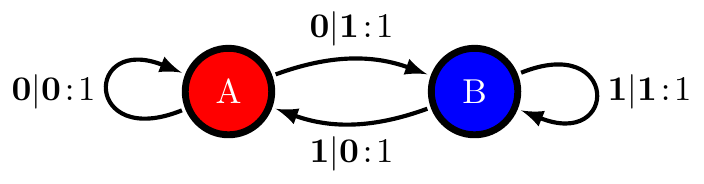}\\
\includegraphics[scale=.5]{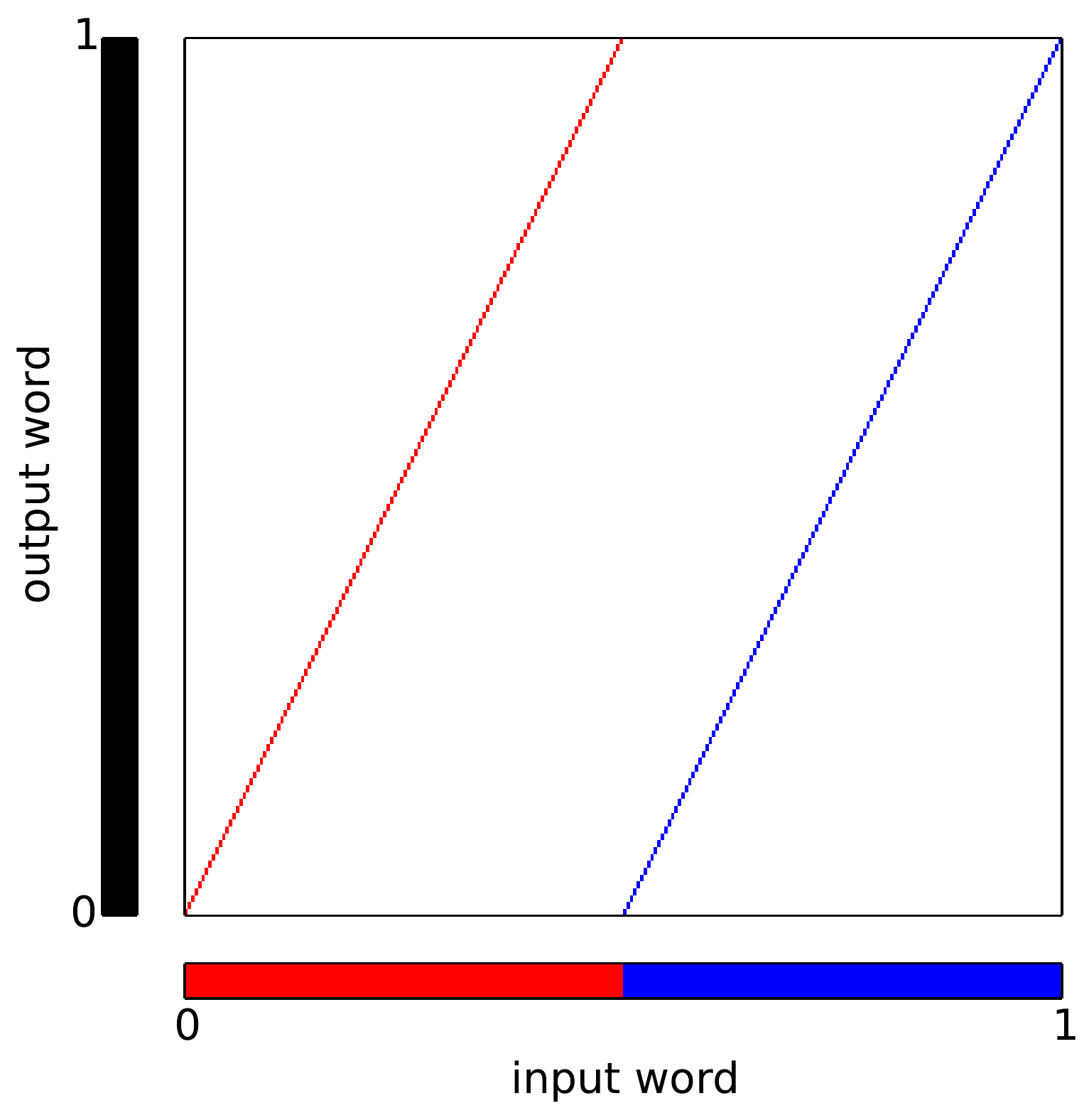}
\caption{Delay Channel: \eT\ and causal-state colored wordmap.}
\label{fig:delaystate}
\end{figure} 

In contrast, the Delay Channel (Fig. \ref{fig:delaystate}) has two causal states and so two colors
corresponding to pasts in which the input ends on a $0$ (left half of the
wordmap) or a $1$ (right half of the wordmap). This partitioning into halves is
a characteristic of channels with a pure feedforward Markov order $\MOrder_\text{pff}
= 1$.  We also see that the output words are colored black, illustrating the
fact that the output tells us nothing about the current causal state. The channel's pure feedforward Markov order of $1$ can be seen in the channel's \eT\ state-transition diagram by observing that all transitions on input
symbol $0$ lead to state $A$ and all transitions on input symbol $1$ lead to
state $B$. Since the Delay Channel is undefined for outputs alone ($\MOrder_\text{pfb}$ is undefined), it is the first example
channel with a nontrivial irreducible feedforward order: $\MOrder_\text{iff} = 1$. The causal states of the Delay Channel simply store a single
bit of input, their entropy therefore matches the length-$1$ block entropy of the input
process: $C_\IS = H[\IS_0]$. Maximizing the latter over input process gives us
$\overline{\Cmu}=1$, attained with Fair Coin Process input.

\begin{figure}
  \centering
\includegraphics[scale=1]{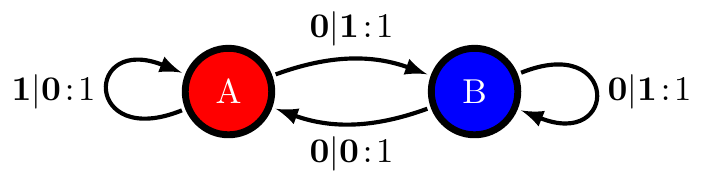}\\
  \includegraphics[scale=.5]{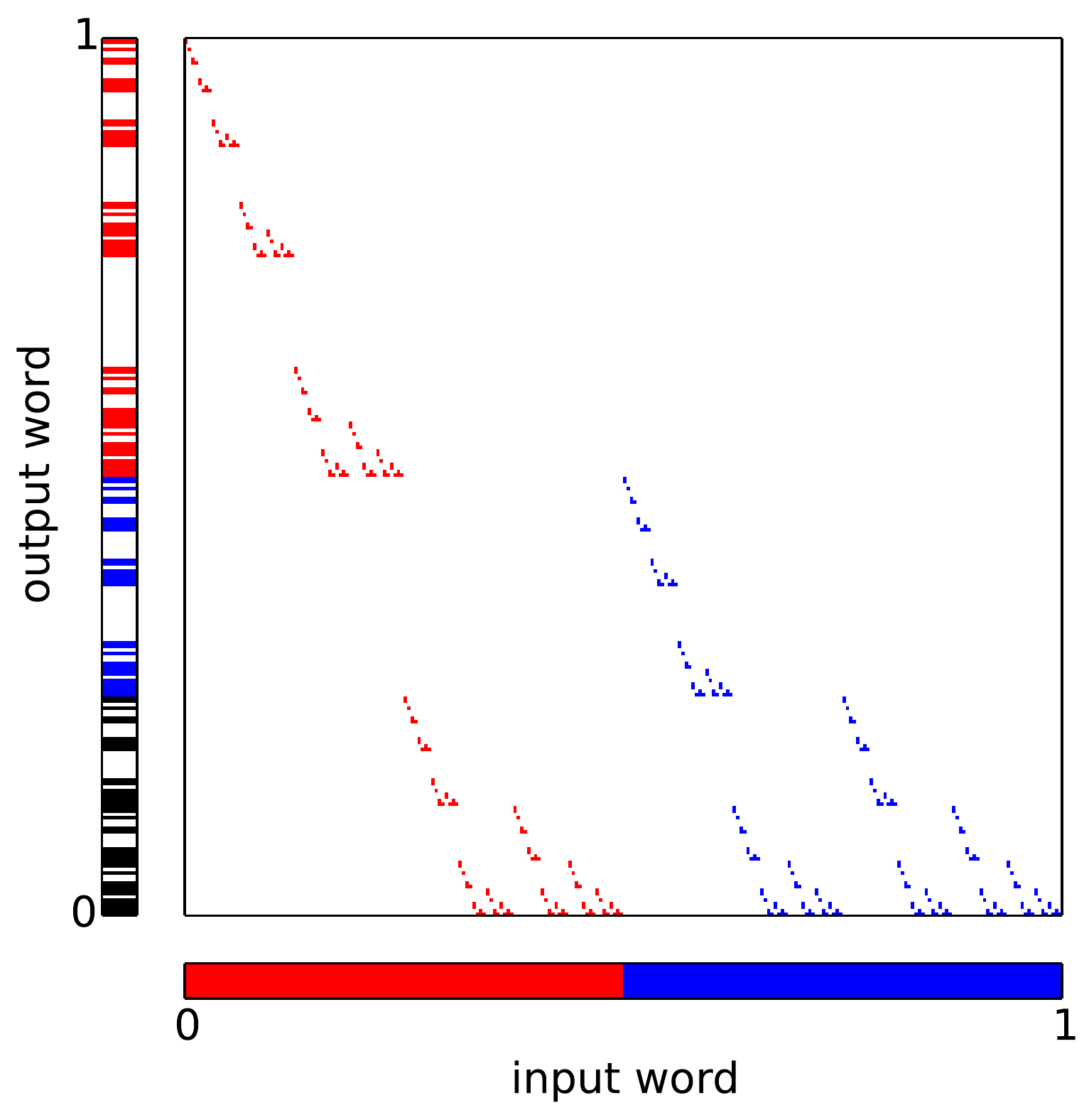}
\caption{Feedforward NOR Channel: \eT\ and causal-state colored wordmap.}
\label{fig:feedforwardNORstate}
\end{figure} 

The Feedforward NOR Channel (Fig. \ref{fig:feedforwardNORstate}) also stores
the previous input symbol and, therefore, partitions input histories the same
way ($\MOrder_\text{pff}= 1$). We see in the wordmap that there is again an ambiguity of causal state
given output histories alone and there is, therefore, no pure feedback
presentation for the channel. Specifically, we see that the ambiguity arises
for histories where the output ends on two $0$s (lower quarter of the wordmap).
This can be verified in the channel's \eT\ by observing that a $1$ on output
always leads to state $A$, but a $0$ on output only leads to a unique state if
it is followed by a $1$ on output. A single symbol of input is \emph{always} needed to guarantee well defined behavior ($\MOrder_\text{iff} = 1$). Since the Feedforward NOR Channel's causal
states store the same information as the Delay Channel, we can again drive the
channel with the Fair Coin Process to attain $\overline{\Cmu}=1$ bit.

\begin{figure}
  \centering
\includegraphics[scale=1]{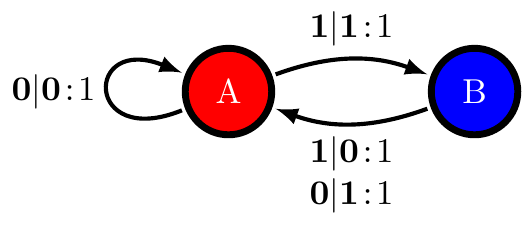}\\
  \includegraphics[scale=.5]{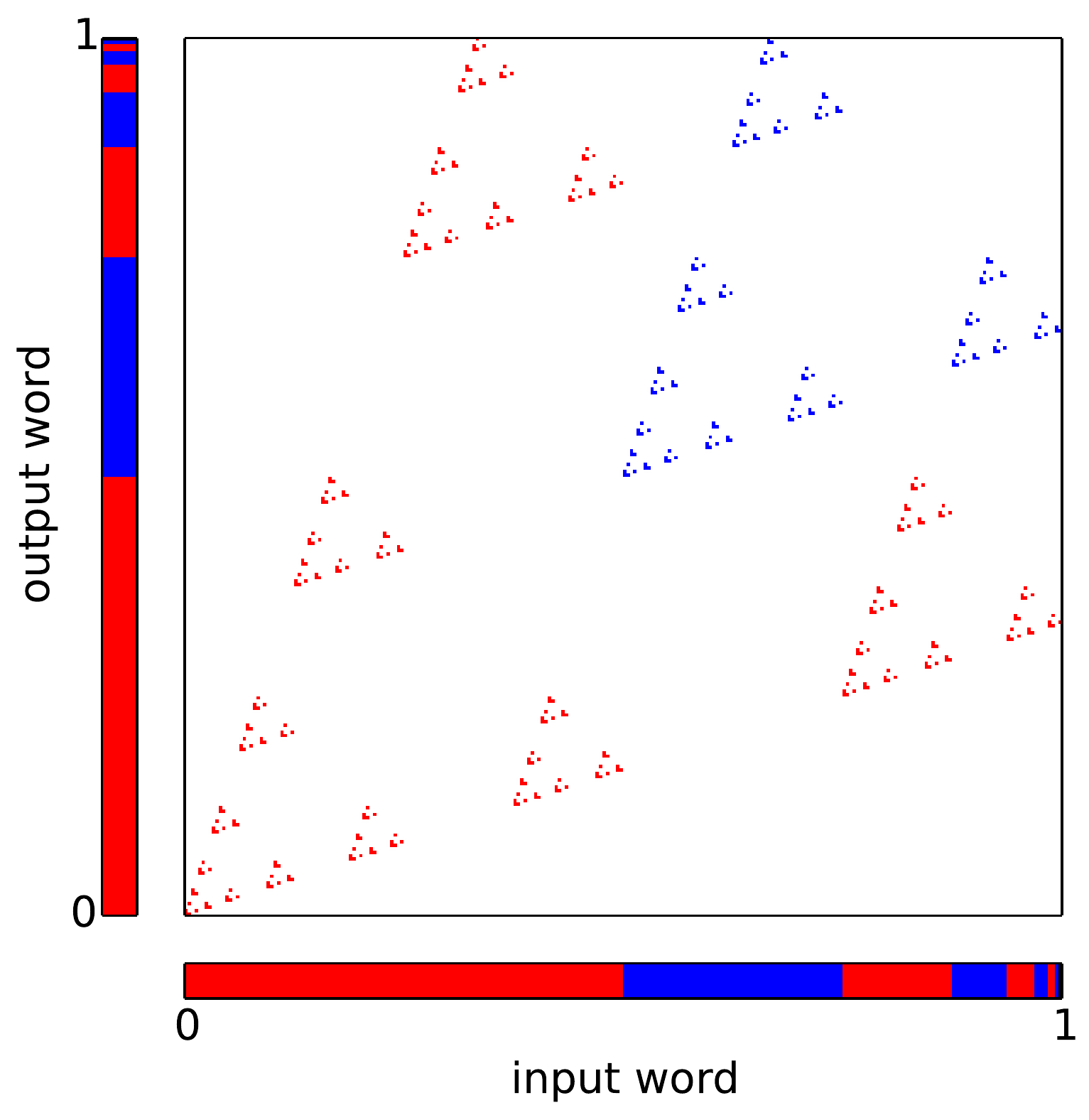}
\caption{Odd NOT Channel: \eT\ and causal-state colored wordmap.}
\label{fig:flipoddstate}
\end{figure} 

The wordmap for the Odd NOT Channel (Fig. \ref{fig:flipoddstate}) has projected
input (and output) partitions with structure at all scales. This is the
signature of states that depend upon infinite histories---one must provide an
arbitrarily long binary expansion to specify the location of the causal state
boundaries and, therefore, the causal states themselves. If we observe both
inputs and outputs, we only need to specify in which quadrant a joint history
lies in order to determine its causal state. That is, the Odd NOT Channel is
sofic on both input and output alone (infinite pure feedforward and pure feedback Markov
orders, $\MOrder_{\text{pff}}$ and $\MOrder_{\text{pfb}}$, respectively), but Markovian when
both input and output are considered (finite Markov order $\MOrder$). We also
see that the causal states store the \emph{same information} (parity) about
input histories as they do output histories, by observing the symmetry along
the diagonal. Since the Period-$2$ Process generates sequences that always
alternate between even and odd parity, we can drive the channel with this
process to induce a uniform distribution over its causal states.  Therefore, we
have $\overline{\Cmu} = 1$ bit, again.

\begin{figure}
  \centering
\includegraphics[scale=1]{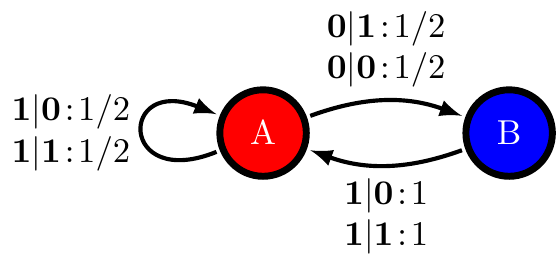}\\
  \includegraphics[scale=.5]{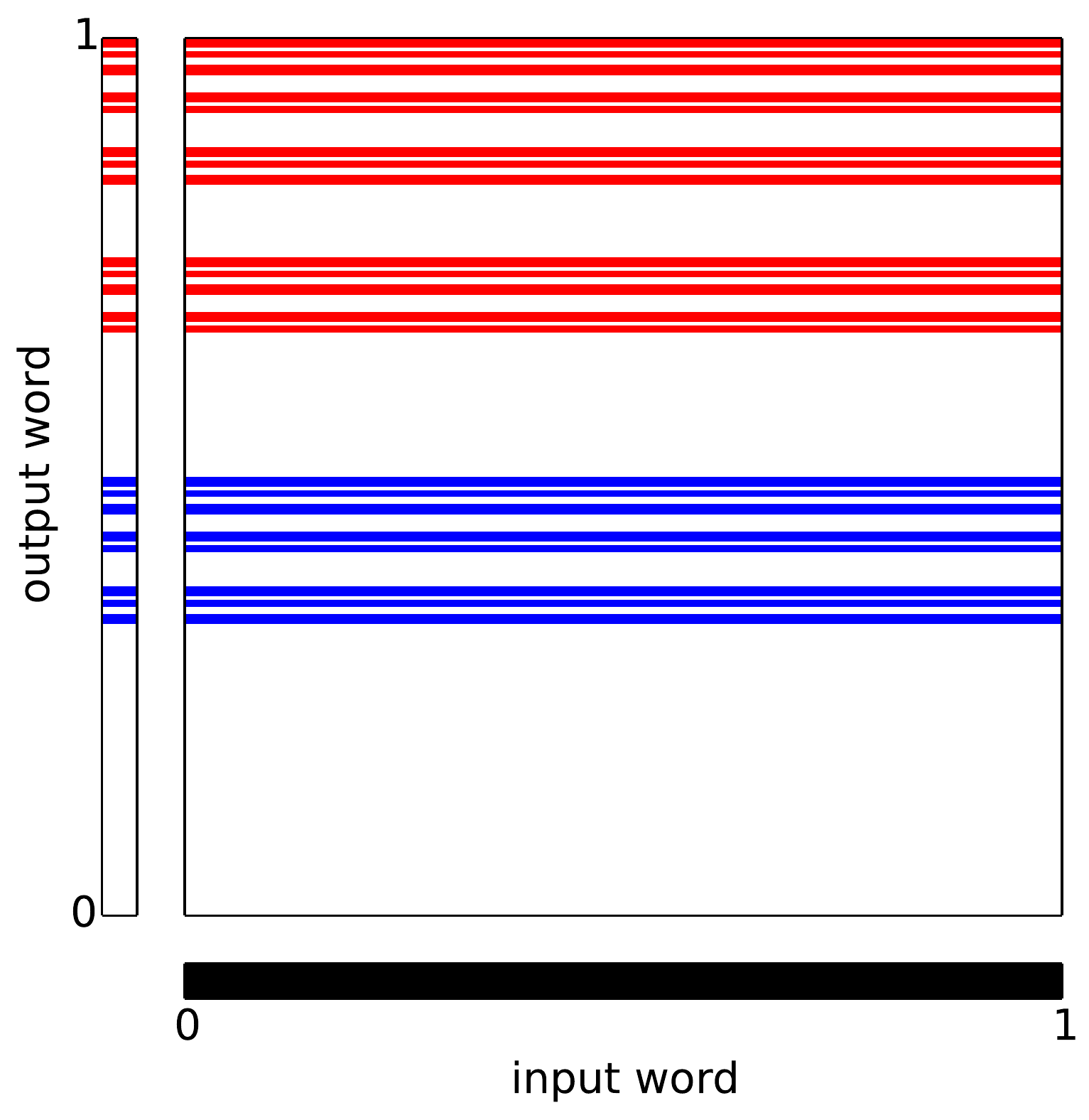}
\caption{All-Is-Golden Channel: \eT\ and causal-state colored wordmap.}
\label{fig:alltogmstate}
\end{figure} 

We see that the All-Is-Golden Channel (Fig. \ref{fig:alltogmstate}) has a
pure feedback Markov order of $\MOrder_\text{pfb} = 1$. Since it ignores its input,
however, input histories tell us nothing about in which causal state the
channel is. We also see that the wordmap is horizontally symmetric due to this
lack of input dependence. Since the state transitions depend only on output,
the state distribution and, therefore, the statistical complexity are
independent of input. In particular, the channel's statistical complexity is
that of the Golden Mean Process (GMP):
$C_X = \overline{\Cmu} = \Cmu(\text{GMP}) \approx 0.918$ bits.

\begin{figure}
  \centering
\includegraphics[scale=1]{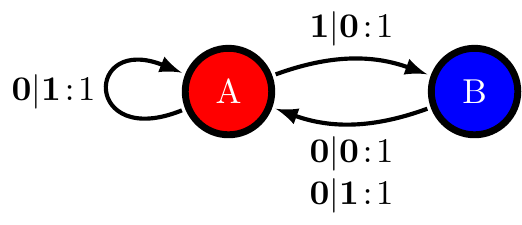}\\
  \includegraphics[scale=.5]{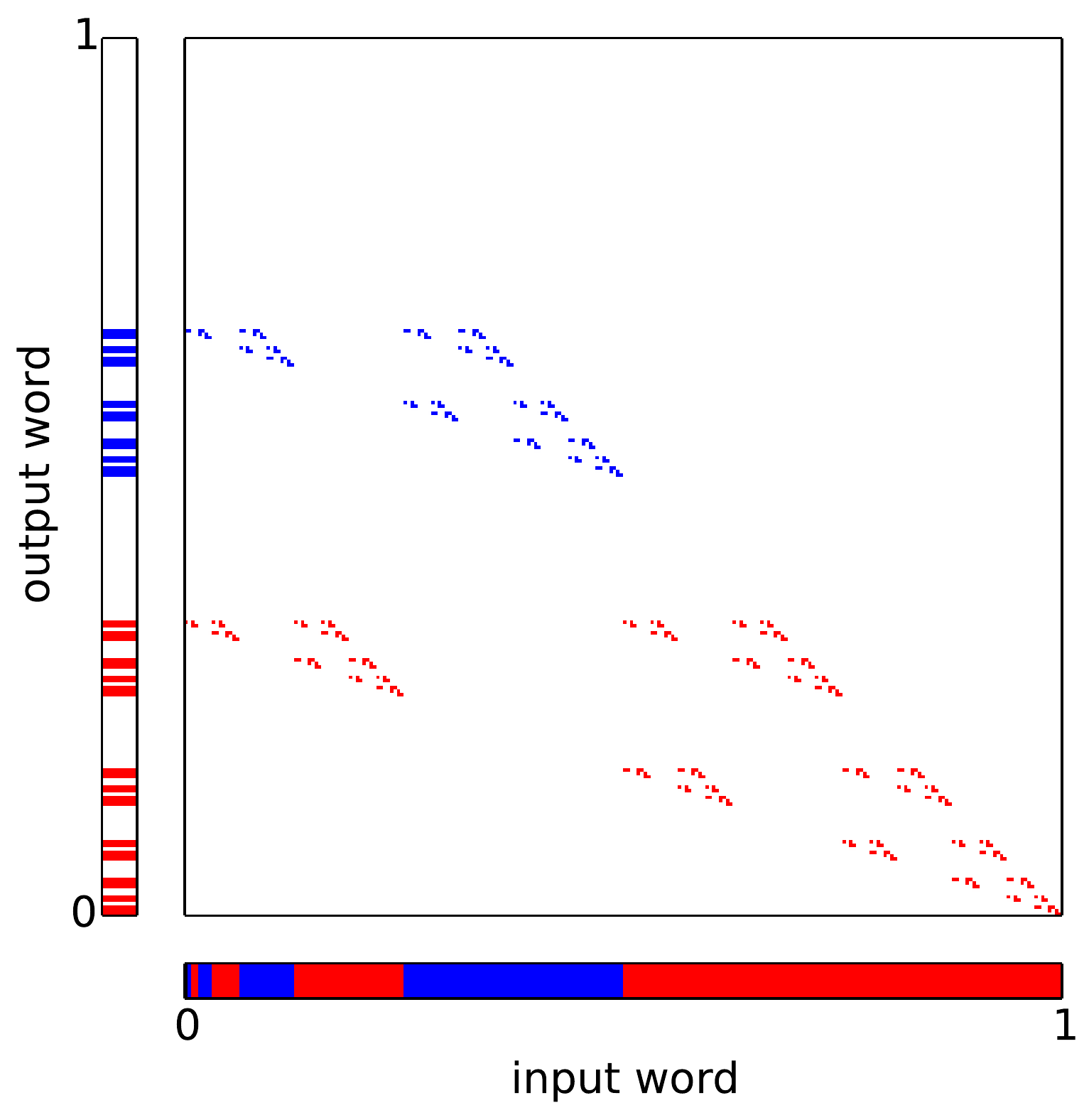}
\caption{Feedback NOR Channel: \eT\ and causal-state colored wordmap.}
\label{fig:feedbackNORstate}
\end{figure} 

\begin{figure}
  \centering
\includegraphics[scale=1]{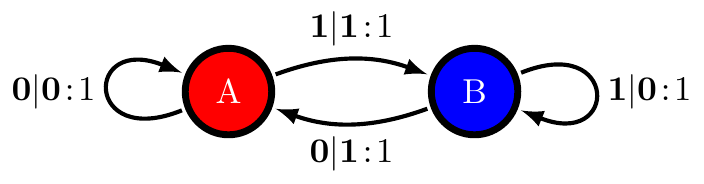}\\
  \includegraphics[scale=.5]{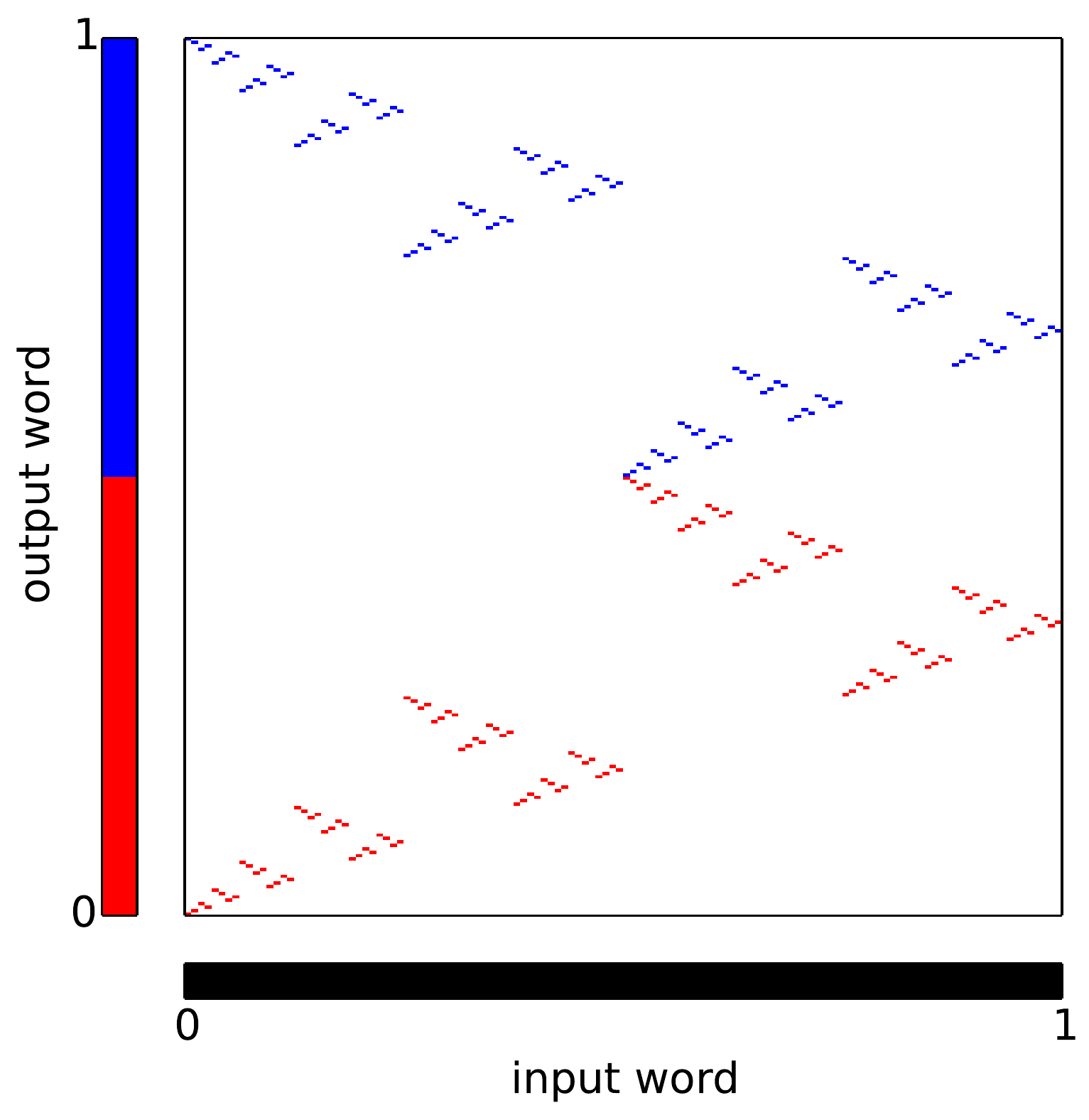}
\caption{Feedback XOR Channel: \eT\ and causal-state colored wordmap.}
\label{fig:feedbackXORstate}
\end{figure} 

The wordmap for the Feedback NOR Channel (Fig. \ref{fig:feedbackNORstate})
clearly shows that it has infinite pure feedforward Markov order
$\MOrder_{\text{pff}}$, but finite
Markov $\MOrder$ and pure feedback Markov $\MOrder_{\text{pfb}}$ orders. Contrast this with the wordmap for the
Feedback XOR Channel (Fig. \ref{fig:feedbackXORstate}) clearly showing
that the causal state, and so the channel's behavior, cannot be determined by
input alone. Observe that the Feedback NOR Channel is in state $A$ with
probability $1$ when a $1$ is observed on input and oscillates between states
$A$ and $B$ if the channel is driven with a period-$2$ cycle of $0$s and $1$s
from that point on. We can therefore induce a uniform distribution over causal
states by driving the channel with the Period-$2$ Process. We can also induce
a uniform distribution over the Feedback XOR Channel's causal states by driving
the channel with the Fair Coin Process, which causes all state transitions to occur with equal probability.
In both cases, $\overline{\Cmu} = 1$ bit.

\begin{figure}
  \centering
\includegraphics[scale=1]{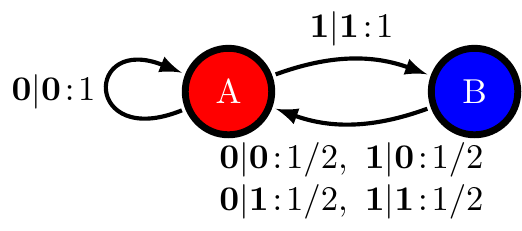}\\
  \includegraphics[scale=.5]{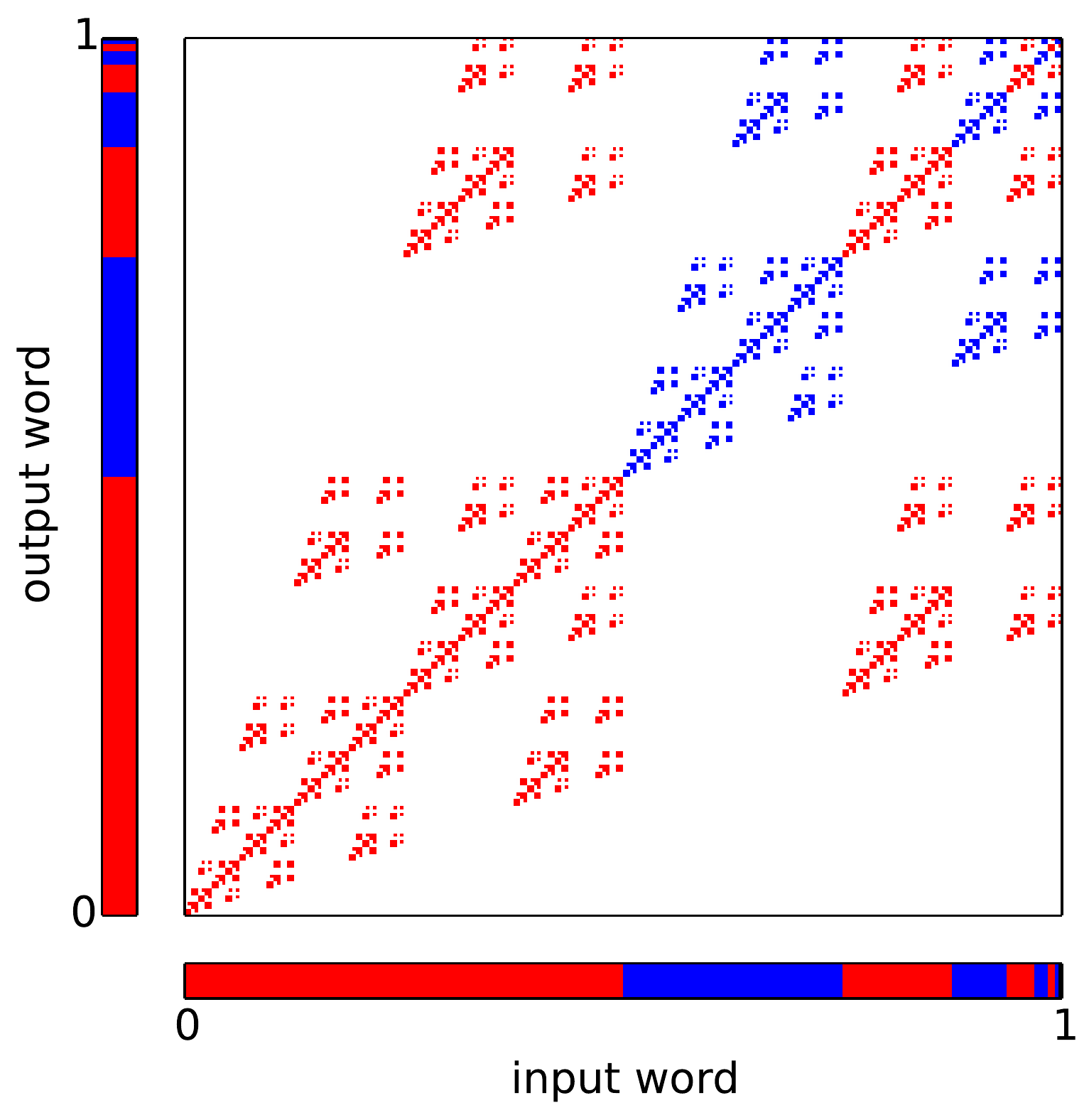}
\caption{Odd Random Channel: \eT\ and causal-state colored wordmap.}
\label{fig:oddrandomstate}
\end{figure} 

The Odd Random Channel (Fig. \ref{fig:oddrandomstate}) has infinite
pure feedforward and pure feedback Markov orders ($\MOrder_{\text{pff}}=\MOrder_{\text{pfb}}=\infty$), but unlike the Odd NOT channel,
the Markov order $\MOrder$ is infinite. In the Odd NOT channel, we saw structure
at all scales in the input and output projections of the wordmap, but a
partitioning into quadrants in the complete wordmap. Now, we see that there
is no such simple partition in the complete wordmap, and there is structure
at all scales in the coloring of histories. In other words, one must specify
arbitrarily long \emph{pairs} of binary expansions (input and output) in order
to specify a causal state. Specifically, knowing that a past ended in
$\jsym_t=(1,1)$ does not uniquely determine a causal state. This can be seen as
multiple colors appearing in the upper-right quadrant of the wordmap. If,
however, we know that the \emph{previous} pair was $(0,0)$, $(0,1)$, or
$(1,0)$, we see that the history maps to causal state $B$; corresponding to
the upper-left, lower-left, and lower-right subquadrants of the upper-right
quadrant, respectively. Similarly, if the previous pair was $(1,1)$, we are
left with an ambiguity in state; corresponding to the upper-right subquadrant
of the upper-right quadrant. Therefore, we require an arbitrarily
long past to determine the channel's causal state in general. Since the causal
states store the same parity as the Odd NOT channel, we can again drive the
channel with the Period-$2$ Process to induce a uniform distribution, giving
us $\overline{\Cmu} = 1$ bit.

\begin{figure}
  \centering
\includegraphics[scale=1]{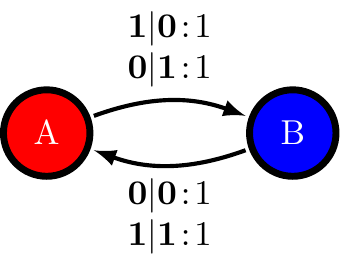}\\
  \includegraphics[scale=.5]{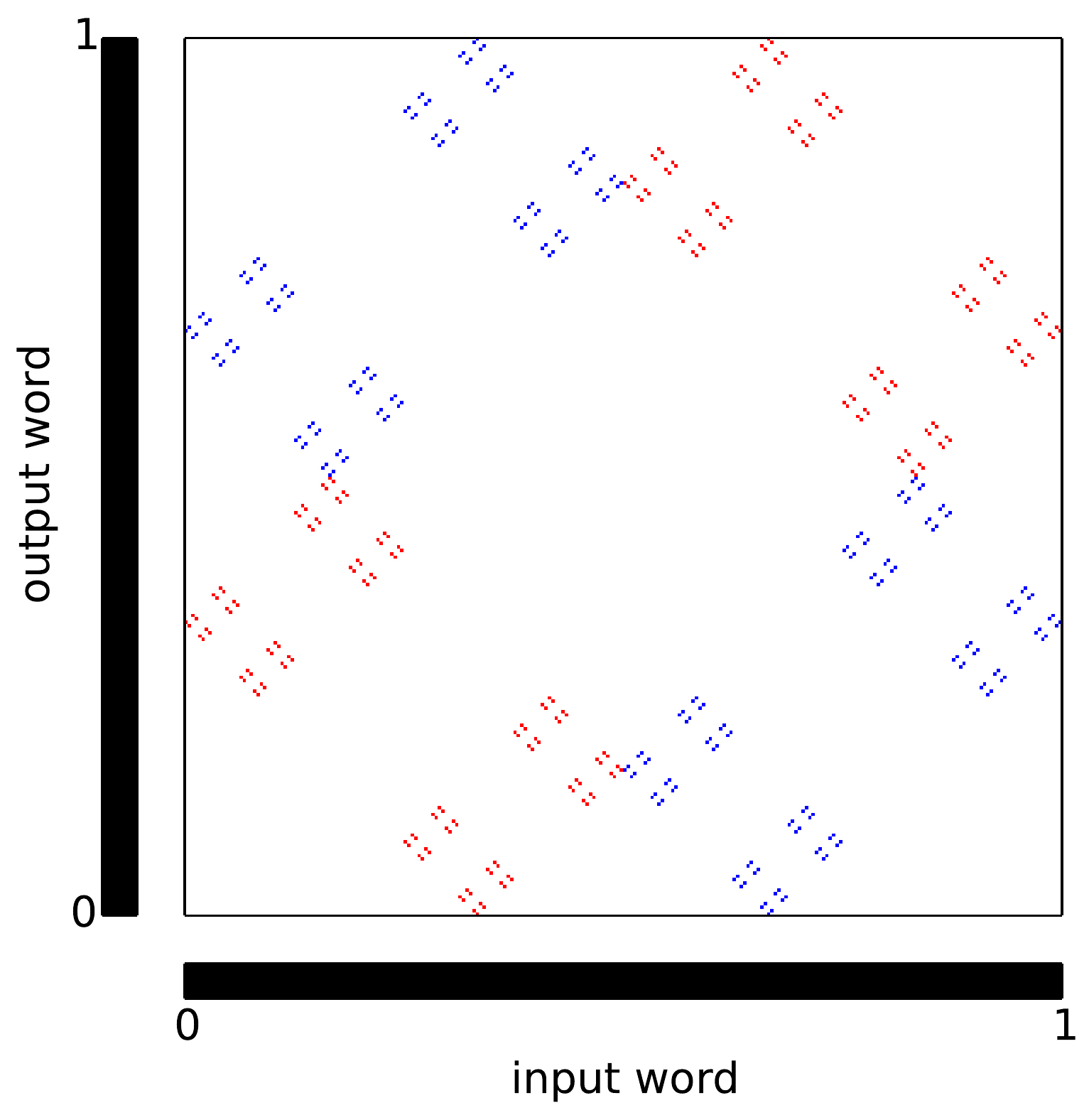}
\caption{Period-2 Identity NOT Channel: \eT\ and causal-state colored wordmap.}
\label{fig:flipidentitystate}
\end{figure} 

The Period-2 Identity NOT Channel (Fig. \ref{fig:flipidentitystate}) has a
Markov order of $\MOrder = 1$, but we clearly see that neither input nor output
alone determines the channel's causal state ($\MOrder_\text{iff}=\MOrder_\text{ifb}=1$). Since the states have a uniform
distribution regardless of input, we have $C_\IS = \overline{\Cmu} = 1$ bit.

\section{Infinite-state \ETs: The Simple Nonunifilar Channel}

\begin{figure}
  \centering
\includegraphics[scale=1]{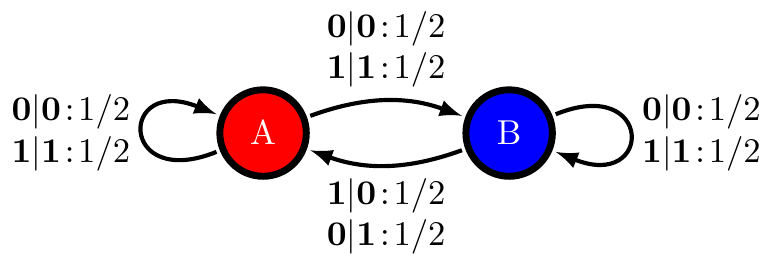}\\
  \includegraphics[scale=.5]{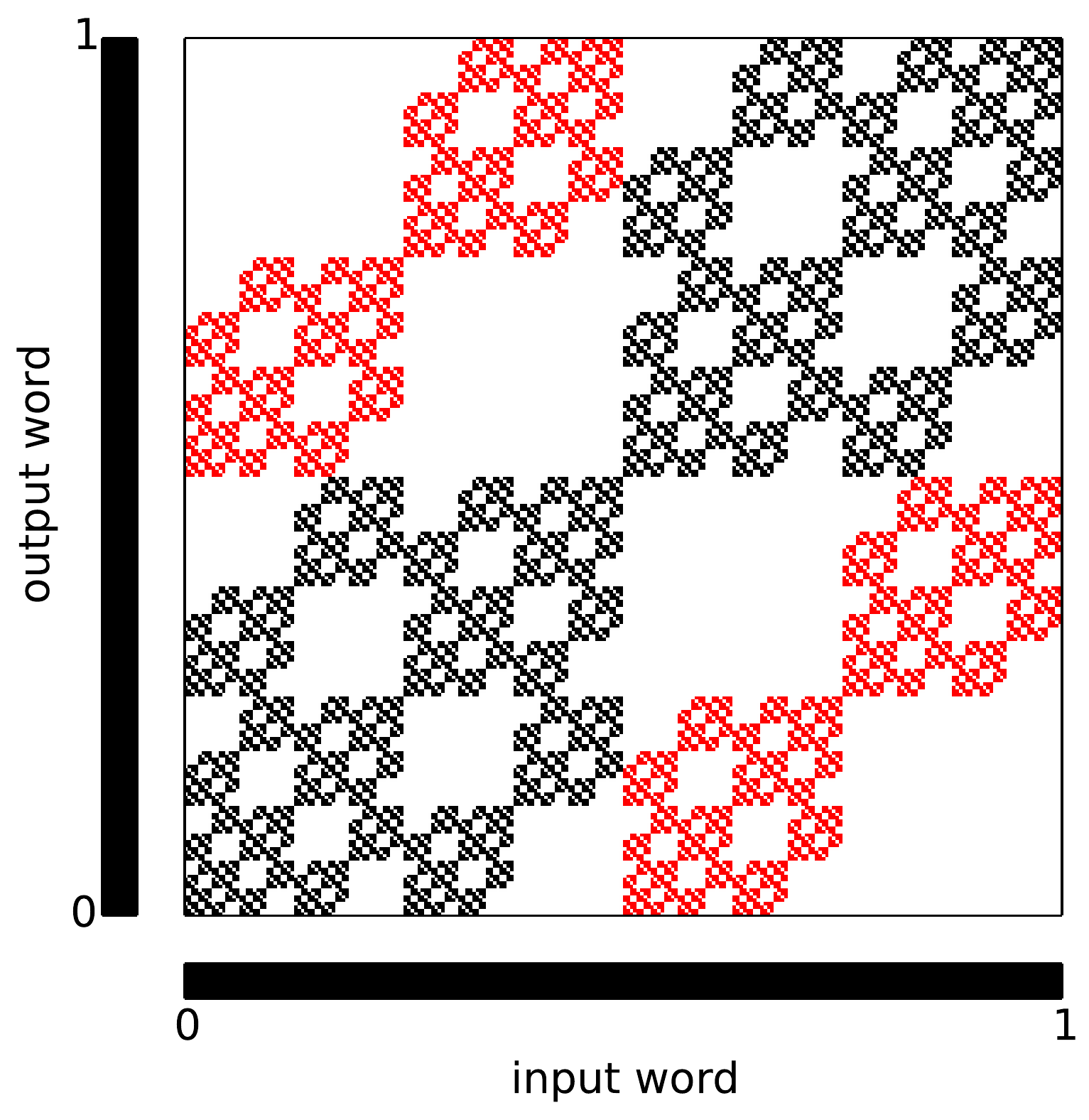}
\caption{Simple Nonunifilar Channel: \emph{Nonunifilar} transducer presentation and state colored wordmap.}
\label{fig:SNC}
\end{figure} 

The example channels were chosen to have a finite number of causal states
(typically two), largely to keep the analysis of their structure accessible. We
can see that even with a few states, \eTs\ capture a great deal of behavioral
richness. Nonetheless, many channels have an infinite number of causal states.
Consider, for example, the \emph{Simple Nonunifilar Channel}. This channel's
behavior is captured simply by the finite-state presentation shown in Fig.
\ref{fig:SNC}. When in state $A$, the channel behaves as the identity and has
an equal probability of staying in state $A$ or transitioning to state $B$.
When in state $B$, the channel will either behave as the identity and
transition back to state $B$ or behave as the bit flipped identity and
transition to state $A$, each with equal probability. 

Observe that the transducer shown is \emph{non}unifilar and is, therefore,
\emph{not} the \eT\ for the channel. For example, the joint symbol $(0,0)$ can
cause state $A$ to transition to either itself or state $B$. This
nonunifilarity manifests in the wordmap as large blocks of black points. These
indicate joint histories that lead to a \emph{mixture} of transducer states.
This illustrates the fact that an observer cannot typically retain
synchronization to a \emph{particular} state of a nonunifilar transducer---a
problem not present when using unifilar transducers.

\begin{figure}
  \centering
\includegraphics[scale=1]{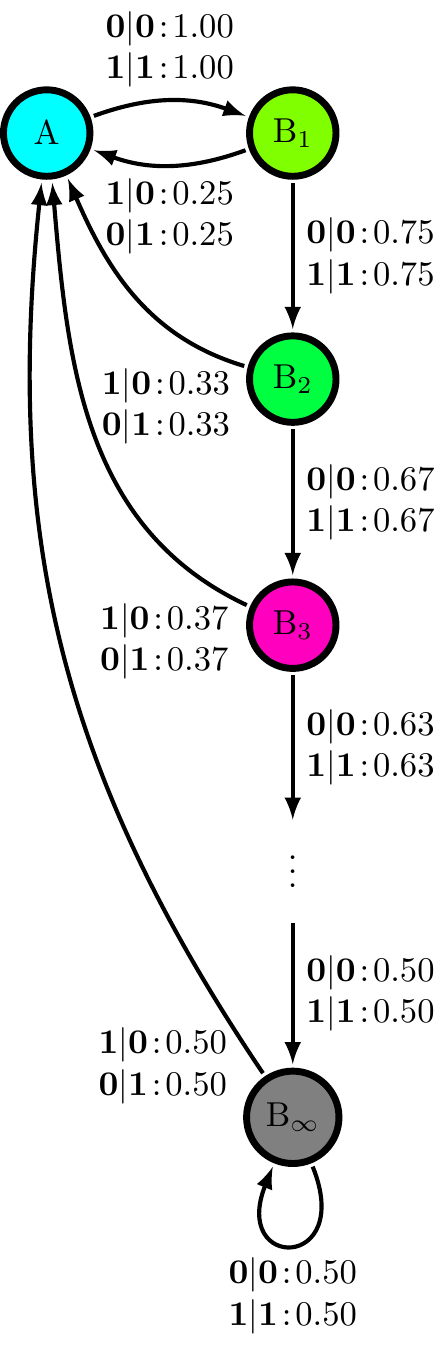}\\
  \includegraphics[scale=.5]{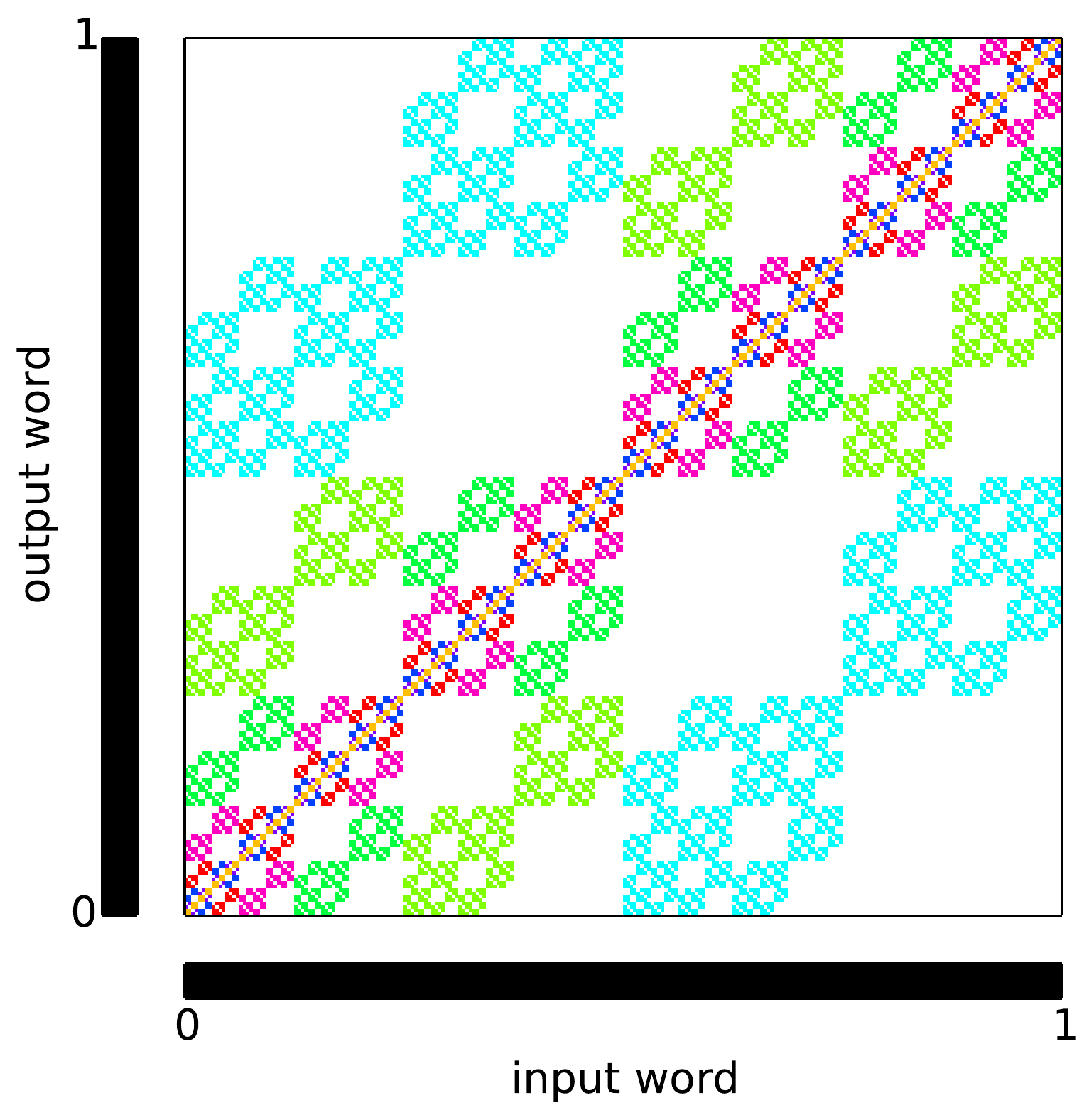}
\caption{Simple Nonunifilar Channel: \eT\ and causal-state colored wordmap.}
\label{fig:SNCem}
\end{figure} 

It is possible to construct the \eT\ for the Simple Nonunifilar Channel, but
doing so results in a transducer with a countably infinite set of states.
This minimal, unifilar \eT\ can be seen in Fig. \ref{fig:SNCem}. Since any
channel with a finite Markov order (and finite alphabet) will have a finite
number of causal states, a channel with an infinite number of causal states
will have infinite Markov order. This is also evident in the causal-state
wordmap for the Simple Nonunifilar Channel, as one needs infinite resolution
(in general) to determine to which causal state a joint past leads. In fact,
even if we know either the infinite input or output past, we still need to know
the full output or input past, respectively, in order to characterize the
channel's behavior. This is therefore the first example we have seen with
$\MOrder_\text{iff}=\MOrder_\text{ifb}=\infty$.

Observe that while the Simple Nonunifilar Channel's output clearly depends upon
its input, its state-to-state transitions do \emph{not}. Its statistical
complexity is therefore independent of the input process chosen. In fact, the
causal states and transitions between them are identical to the Simple
Nonunifilar Source \cite{Crut92c}. The statistical complexity is therefore
equal to the statistical complexity of the Simple Nonunifilar source: $C_\IS =
\overline{\Cmu} \approx 2.71$ bits. Even though there are an infinite number of
states, the $B_i$ states are occupied with probability that decreases quickly
with $i$, thus allowing for a finite Shannon state entropy. Note that if one
were to use Fig. \ref{fig:SNC}'s \emph{nonunifilar} presentation for the Simple
Nonunifilar Channel, the statistical complexity would be underestimated as
$C_\IS = \overline{\Cmu} = 1$ bit.

\section{Discussion}
\label{sec:Discussion}

Previously, we described computational mechanics in the setting of either
generating or controlling processes \cite{Crut10a}. As noted there, generation
and control are complementary. Here, we developed computational mechanics in
a way that merges both control (the input process) and generation (the output
process), extending the \eM\ to the \eT. With this laid out, we describe how
the \eT\ overlaps and differs from alternatives to modeling input-output
processes. We then turn to discuss applications, which incidentally elucidate
our original motivations, and suggest future directions.

\subsection{Related Work: Modeling}

Following the signposts of earlier approaches to modeling complex, nonlinear
dynamical systems \cite{Pack80,Crut87a}, we are ultimately concerned with
reconstructing a transducer when given a general channel or given a joint
process, either analytically or via statistical inference. And so, when
discussing related efforts, we distinguish between those whose goal is to
extract a model, which we review now, and those that analyze types of
transductions, which we review next. After this, we turn to applications.

For statistical estimation we note that the recently introduced Bayesian
Structural Inference (BSI) \cite{Stre13a} allows one to estimate the posterior
probability that \eMs\ generate a given, even relatively short, data series.
BSI's generality allows it to be readily adapted to infer \eTs\ from samples of
an input-output process. This turns on either developing an enumeration of \eTs\ which
parallels that developed for \eMs\ in Ref. \cite{John10a} or on developing a
list of candidate \eTs\ for a given circumstance. And, these are also
readily accomplished. A sequel provides the implementations. Previously,
inferring causal states, and so causal-state filters, had also been addressed; see, for example, Refs. \cite{Shal02a,Stil07b,Stre13a}.

Optimal transducers were originally introduced as structure-based filters to
define hierarchical \eM\ reconstruction \cite{Crut89j,Crut92c} in terms of
causal-state filtering, to detect emergent spatiotemporal patterns
\cite{Hans90a,Crut91d,Hans95a,McTa04a}, and to explore the origins of
evolutionary selection pressure \cite{Crut04a} and the evolution of language
structure \cite{Crut10c}. These causal-state transducers were first formalized
in Ref. \cite{Crut94a} and several of those results are reproduced in Ref.
\cite{Shal01a}. Appendix \ref{app:ETEquivalence} shows that the definition
there, which makes additional assumptions compared to that here, are
equivalent. The more-general development is more elegant, in that
it establishes unifilarity, for example, rather than assume such a powerful
property. Likely, in
addition, the generality will allow \eTs\ to be more widely used.

Throwing the net wider---beyond these, most directly related, prior
efforts---there have been many approaches to modeling input-output mappings. We
will use the fact that most do not focus on quantitatively analyzing the
mapping's intrinsic structure to limit the scope of our comments. We
mention a few and then only briefly. Hopefully the list, nonetheless,
suggests directions for future work in these areas.

Today, many fall under the rubric of learning, though they are rather more
accurately described as statistical parameter estimation within a fixed model
class. Probably, the most widely used and developed methods to model general
input-output mappings are found in artificial neural networks
\cite{Hopf82,Hert86a} and in the more modern approaches that employ kernel
methods \cite{Hast11a}, statistical physics \cite{Mack03a}, and information
theory \cite{Bial00a,Shal99a}. Often these methods require IID-drawn samples
and so do not directly concern mappings from one temporal process to another.
Unlike \eTs, they are also typically limited to model classes---e.g.,
feedforward and directed acyclic graph structures---that do not allow internal
feedback or dynamics.

That said, neural networks that are recurrent are universal approximators of
dynamical systems and, per force, are channels with feedback and feedforward
memory \cite{Albe05a}. They are well known to be hard to train and, in any
case, rarely quantitatively analyzed for the structures they capture when
successfully trained. In the mathematical statistics of time series, for
comparison, AutoRegressive-Moving-Average model with eXogenous inputs model
(ARMAX models) are channels with feedback and feedforward memory, but they are
\emph{linear}---current output is a linear combination of past inputs and
outputs. The nonlinear generalization is the Nonlinear AutoRegressive eXogenous
model (NARX), which is a very general memoryful causal channel. At some future
time, likely using \eTs\ extended to continuous variables as recently done for
\eMs\ in Ref. \cite{Marz14b}, we will understand better the kinds of structure
these channels can represent.

\subsection{Related Work: Classification}

Beyond developing a theoretical framework for structured transformations, one
that is sufficiently constructive to be of use in statistical inference, there
are issues that concern how they give a new view, if any, of the organization
of the space of structured processes itself.

Specifically, computational mechanics up to this point focused on processes and
developed \eMs\ to describe them as stochastic sets. \EMs\ are, most simply
stated, compact representations of distributions over sequences. With the \eTs\
introduced here, computational mechanics now has formalized stochastic mappings
of these stochastic sets. And, to get to the point, with sets and mappings one
finally has a framework capable of addressing the recoding equivalence notion
and the geometry of the space of processes proposed in Ref. \cite{Crut87f}. A
key component of this will be a measure of distance between processes that uses
a structural measure from the minimal optimal mapping (\eT) between them. This
would offer a constructive, in the sense we use the word, approach to the view
of process space originally introduced by Shannon \cite{Shan53a,Yeun08a,Li11a}.

This then leads to the historically prior question of structurally classifying
processes---paralleling schemes developed in computation theory \cite{Hopc79}.
Indeed, our development is much closer to input-output processes from the
earliest days of dynamical systems and automata theory---which were concerned
with exploring the range of behaviors of mechanical systems and the then-new
digital computers.

Briefly, \eTs\ are probabilistic endomorphisms of subshifts as studied in
symbolic dynamics \cite{Lind95a}. The (nonprobabilistic) endomorphisms there
were developed to explore the equivalence of processes via conjugacies.
Notably, this area grew out of efforts in the 1920s and 1930s by Hedlund,
Morse, Thue, and others to define symbolic dynamical systems that were more
analytically tractable than continuum-state systems \cite{Mors38a}. Their
efforts played a role in forming the logical foundations of mathematics and so
eventually in the emergence of a theory of computation via Church, G\"odel,
Post, and Turing \cite{Post21a,Gode92a,Turi37a,Chur36a}. This led eventually
to Moore's abstractions of sequential machines and transducers \cite{Moor56a}
and to Huffman's concept of a minimal implementation \cite{Huff54} and
information lossless automata \cite{Huff54b,Huff59a,Huff59b}. Today, though
expositions are increasingly rare, finite-state transducers are
covered by several texts on computation theory; see, for example, Ref.
\cite{Broo89a}.

Once one allows for distributions over sequences, though, then one shifts from the overtly structural approach of symbolic dynamics and automata to Shannon's information sources and communication channels \cite{Shan48a} and a strong emphasis on stochastic process theory. As noted in the introduction, one principal difference is that here we considered channels with memory, while the latter in its elementary treatments considers memoryless channels or channels with very restricted forms of memory. Finite-state channels have been developed in limited way, though; for example, see Ref. \cite[Ch. 7]{Ash65a} and for very early efforts see Refs. \cite{Blac58a} and \cite{Blac61a}. There are also overlaps, as we attempted to show in the selected examples, with classifications
developed in digital filter theory \cite{Hamm97a}.

There are also differences in focus and questions. Whereas information theory
\cite{Shan48a,Cove06a} studies \emph{quantities} of information such as
intrinsic randomness and informational correlation, computational mechanics
\cite{Crut12a} goes an additional step and attempts to quantify the information
\emph{itself}---the \emph{computational structure} or memory within a system.
This is achieved not by assuming a class of model directly, but by making a
simple assumption about modeling itself: The only relevant information is that
which contributes to prediction---the ``difference that makes a difference'' to
the future \cite{Bate79}. Via the causal equivalence relation, this
assumption leads directly to the \emph{unique}, \emph{maximally predictive},
and \emph{minimally complex} model of our measurement data---the \emph{\eM}.
Another way to express this is that \eTs\ give a constructive way to explore
the information theory of channels with and without memory.

\subsection{Applications}

Our development of \eTs\ was targeted to provide the foundation for several
related problems---problems that we will address elsewhere, but will briefly
describe here to emphasize general relevance and also to suggest future directions.

\subsubsection{Inference versus Experimentation}

If all data collected is produced by a measuring device, then any model formed
from that data captures both the structure of the system \emph{and sensor} in
combination. Is there a natural separation of measuring instrument from system
measured? We can now pose this precisely in terms of \eTs: Are there optimal
decompositions of a process's \eM\ into a possibly smaller
\eM\ (representing the hidden process) composed with a \eT\ (representing the
measuring instrument)?

\subsubsection{Information Flow within and between Systems}

In nonlinear dynamics and in information theory there has been a long-lived
interest in how information ``flows'' and how such flows relate to a system's
mechanical organization; see Refs. \cite{Shaw81,korner06,plw-11,barnett-09}, to
mention only a few. These employ specializations of Eq. (\ref{eq:EE})'s excess
entropy, being various forms of conditional mutual information. The
\emph{transfer entropy} \cite{Schr00a} and the earlier \emph{directional
information} \cite{Mark73a,Mass90a} are two such. The main issue concerns how
one process affects another and so this is a domain in which \eTs---as optimal
models of the structured transformations between processes---can help clarify
the issues.

In particular, there has been recent criticism of the use of these as measures
of information flow and, specifically, their relation to the structural
organization of the flows \cite{Sun14a}. We can now do better, we believe,
since \eTs\ give a canonical presentation with which to describe and extract
the structure of such mappings. And this, in turn, allows one to explicitly
relate how causal-state structure supports or precludes information flows. We
address this problem in a sequel \cite{Barn14a}.

\subsubsection{Process Decomposition}

Given a process, we can now analyze what internal components drive or are
driven by other internal components. As one example, Is a subset of the
measurement alphabet the ``output'' being driven by another subset that is
``input''? The question hints at the solution that one can now provide: Produce
the \eT\ for each bipartite input-output partitioning of the global \eM\
alphabet, giving a set of candidate input-output models. One can then invoke,
based on a notion of first principles (such as parsimony) or prior knowledge, a
way to choose the ``best'' input-output, driver-drivee decomposition.

\subsubsection{Perception-Action Cycles}

Probably one of the most vexing contemporary theoretical and practical
problems, one that occurs quite broadly, is how to describe long-term and
emergent features of \emph{dynamic learning} in which a system models its
input, makes a decision based on what it has gleaned, and takes an action that
affects the environment producing the inputs. In psychology and cognitive sciences
this problem goes under the label of the \emph{perception-action cycle}; in
neuroscience, under \emph{sensori-motor loop} \cite{Cuts11a,Gord11a}. The
problem transcends both traditional mathematical statistics and modern machine
learning, as their stance is that the data is not affected by what is
learned. And in this, it transcends the time-worn field of experiment design
\cite{Fedo72a,Atki01} and the more recent machine learning problem of
\emph{active learning} \cite{Mack03a}. Though related to computational
mechanics via Ref.  \cite{Stil07b}, the recent proposal \cite{Stil07c} for
\emph{interactive learning} is promising, but is not grounded in a systematic
approach to structure. It also transcends control theory, as the latter does
not address dynamically building models, but rather emphasizes how to
monitor and drive a given system into given states \cite{Astr08a}.

\ETs\ suggest a way to model the transduction of sensory input to a model and
from the model to a decision process that generates actions. Thus, the
computational mechanics representation of the perception-action cycle is two
cross-coupled \eTs---one's output is the other's input and vice versa.
Formulating the problem in this way promises progress in analyzing and in
quantifying structures in the space of models and strategies.

Physical applications of \eTs\ to analyze the information thermodynamics of
feedback control in Maxwellian Demons can be seen in Szilard's Engine
\cite{Boyd14b} and the Mandal-Jarzynski ratchet \cite{Mand012a,Boyd15a}.

\section{Conclusion}

Previously, computational mechanics focused on extracting and analyzing the
informational and structural properties of individual processes. The premise
being that once a process's \eM\ had been obtained, it can be studied in lieu
of other more cumbersome (or even inappropriate) process presentations. Since
the \eM\ is also unique and minimal for a process, its structure and quantities
were treated as being those of the underlying system that generated the
process. Strengthening this paradigm, virtually all of a process's
correlational, information, and structural quantities can now be calculated in
closed form using new methods of \eM\ spectral decomposition \cite{Crut13a}.

By way of explaining this paradigm, we opened with a review of stationary
processes and their \eMs, turning to broaden the setting to joint input-output
processes and communication channels. We then defined the (conditional)
causal equivalence relation, which led immediately to transducer causal
states and the \eT. A series of theorems then established their optimality. To
illustrate the range of possible transformations we considered a systematic set
of example channels that, in addition, provided an outline of a structural
classification scheme. As an aide in this, we gave a graphical way to view
structured transformations via causal-state wordmaps. With the framework
developed, one sees that the same level of computational mechanics' prior
analysis of individual processes can now be brought to bear on understanding
structural transformations between processes.

The foregoing, however, is simply the first in a series on the structural
analysis of mappings between processes. The next will address the
information-theoretic measures appropriate to joint input-output processes. We
then will turn to an analysis that blends the present results on the causal
architecture of structured transformations and the information-theoretic
measures, showing how the internal mechanism expressed in the \eT\ supports
information creation, loss, and manipulation during flow. From that point, the
sequels will branch out to address channel composition, decomposition, and
inversion.

Given the diversity of domains in which structured transformations (and their
understanding) appear to play a role, there looks to be a wide range of
applications.  In addition to addressing several of these applications, Sec.
\ref{sec:Discussion} outlined several future research directions. The \eT\
development leads, for example, to a number of questions that can now be
precisely posed and whose answers now seem in reach: How exactly do different
measuring devices change the \eM\ formed from measurements of a fixed system?
What precisely is lost in the measurement process, and how well can we model a
system using a given measuring device? When is it possible to see \emph{past} a
measuring device into a system, and how can we optimize our choice of measuring
device in practice?

\begin{acknowledgments}

We thank Cina Aghamohammadi, Alec Boyd, David Darmon, Chris Ellison, Ryan
James, John Mahoney, and Paul Riechers for helpful comments and the Santa Fe
Institute for its hospitality during visits. JPC is an SFI External Faculty
member. This material is based upon work supported by, or in part by, the U. S.
Army Research Laboratory and the U. S. Army Research Office under contract
numbers W911NF-12-1-0234, W911NF-13-1-0390, and W911NF-13-1-0340. NB was
partially supported by NSF VIGRE grant DMS0636297.

\end{acknowledgments}

\appendix

\section{Equivalence of Two $\epsilon$-Transducer Definitions}
\label{app:ETEquivalence}

We show the equivalence of two different \eT\ definitions, that presented in the
main paper and an earlier version requiring additional assumptions. Since the
\eT\ is determined by its causal equivalence relation, we show that the
respective equivalence relations are the same. The first is defined and
discussed at length above and duplicated here for convenience.

\begin{Def}
The \emph{causal equivalence relation} $\sim_\epsilon$ for channels is defined as follows:
\begin{equation}
\begin{aligned}
\jpast \sim_\epsilon \jpast^\prime \iff & \nonumber \\
  \Prob \big( \OFuture \big| \IFuture, & \JPast = \jpast \big) \\
    & \quad = \Prob \big(\OFuture \big| \IFuture, \JPast = \jpast^\prime \big)
\end{aligned}
\end{equation}
\end{Def}

The second definition is an implicit equivalence relation consisting of an
explicit equivalence relation, along with an additional unifilarity constraint
that, of course, is quite strong \cite{Crut94a,Shal01a}. Here, we make both requirements explicit.

\begin{Def}
The \emph{single-symbol unifilar equivalence relation} $\sim_\epsilon^1$ for channels is defined as follows:
\begin{align*}
& \jpast \sim_\epsilon^1 \jpast^\prime \iff \\
  & \quad (i) \quad \Prob \big( \OS_0 \big| \IS_0,  \JPast_0 = \jpast \big) \\
  & \quad\quad\quad = \Prob \big(\OS_0 \big| \IS_0, \JPast_0 = \jpast^\prime \big)
  \\
  & \quad \text{and:} \\
  & \quad (ii) \quad \Prob \big( \OS_1 \big| \IS_1, \JPast_0 = \jpast,
  \JS_0 = (a,b) \big) \\
  & \quad\quad\quad = \Prob \big(\OS_1 \big| \IS_1, \JPast_0 = \jpast^\prime ,
  \JS_0 = (a,b) \big)
  ~,
\end{align*}
for all $a \in \IA$ and $b \in \OA$ such that:
\begin{align*}
\Prob \big( \JS_0 = (a,b) | \jpast \big) > 0
\end{align*}
and:
\begin{align*}
\Prob \left( \JS_0 = (a,b) | \jpast^\prime \right) > 0
  ~.
\end{align*}
\end{Def}
The second requirement (ii) in the above definition requires that appending any
joint symbol to two single-symbol-equivalent pasts will also result in a pair
of pasts that are single-symbol-equivalent. This is unifiliarity. The second
part of the second requirement ensures that we are only considering
\emph{possible} joint symbols $(a,b)$---symbols that can follow $\jpast$ or
$\jpast^\prime$ with some nonzero probability.

\begin{Prop}
The single-symbol unifilar equivalence relation is identical to the causal equivalence relation.
\begin{proof}
Let $\jpast$ and $\jpast^\prime$ be two pasts, equivalent under $\sim_\epsilon^1$. This provides our base case for induction:
\begin{align}
\Prob \big( \OS_0 \big| \IS_0 , \jpast \big)
  = \Prob \big( \OS_0 \big| \IS_0 , \jpast^\prime \big)
  ~.
\end{align}
Now, let's assume that $\jpast$ and $\jpast^\prime$ are equivalent for
length-$L-1$ future morphs:
\begin{align}
\Prob \big( \OS_{0:L} \big| \IS_{0:L}, \jpast \big)
  = \Prob \big( \OS_{0:L} \big| \IS_{0:L}, \jpast^\prime \big)
  ~.
\end{align} 
We need to show that $\jpast$ and $\jpast^\prime$ are equivalent for length-$L$ future morphs by using the unifilarity constraint.
Unifilarity requires that appending a joint symbol $(a,b)$ to both $\jpast$ and
$\jpast^\prime$ results in two new pasts also equivalent to each other for length-$L-1$ future morphs:
\begin{equation} \label{eq:LMinusOneMorphsab}
\begin{aligned}
\Prob \big( \OS_{1:L+1} & \big| \IS_{1:L+1}, \jpast(a,b) \big) \\
  & = \Prob \big( \OS_{1:L+1} \big| \IS_{1:L+1}, \jpast^\prime(a,b) \big).
\end{aligned}
\end{equation}

Since this must be true for \emph{any} joint symbol, we replace $(a,b)$ with
$\JS_0$ in Eq. (\ref{eq:LMinusOneMorphsab}), giving:
\begin{equation} \label{eq:LMinusOneMorphs}
\begin{aligned}
\Prob \big( \OS_{1:L+1} & \big| \IS_{1:L+1}, \jpast, \JS_0 \big) \\
  & = \Prob \big( \OS_{1:L+1} \big| \IS_{1:L+1}, \jpast^\prime, \JS_0 \big) \\ 
\iff& \\
\Prob \big( \OS_{1:L+1} & \big| \IS_{1:L+1}, \jpast, \IS_0, \OS_0 \big) \\
  & = \Prob \big( \OS_{1:L+1} \big| \IS_{1:L+1}, \jpast^\prime, \IS_0, \OS_0 \big)
  ~.
\end{aligned}
\end{equation}
To arrive at our result, we need to multiply the left side of Eq. (\ref{eq:LMinusOneMorphs}) by $\Prob \big(
\OS_0 \big| \IS_{1:L+1} , \jpast, \IS_0 \big)$ and the right side by $\Prob
\big(\OS_0 \big| \IS_{1:L+1} , \jpast^\prime , \IS_0 \big)$, which we can do
when these quantities are equal. Since our channel is \emph{causal},
$\IS_{1:L+1}$ has no effect on $\OS_0$ when we condition on the infinite joint
past and present input symbol. The two quantities, $\Prob \big(
\OS_0 \big| \IS_{1:L+1} , \jpast, \IS_0 \big)$ and $\Prob
\big(\OS_0 \big| \IS_{1:L+1} , \jpast^\prime , \IS_0 \big)$, therefore reduce to $\Prob
\big( \OS_0 \big| \jpast, \IS_0 \big)$ and $\Prob \big(\OS_0 \big|
\jpast^\prime , \IS_0 \big)$, respectively. But these are equal by the single-symbol
unifilar equivalence relation---the base for induction.  Multiplying
each side of Eq. (\ref{eq:LMinusOneMorphs}) by these two terms (in their original
form) gives:
\begin{align*}
\Prob \big( \OS_{1:L+1} & \big| \IS_{1:L+1}, \jpast, \IS_0, \OS_0 \big) \\
  & \quad\quad\quad\quad\quad
  \cdot \Prob \big( \OS_0 \big| \IS_{1:L+1} , \jpast, \IS_0 \big) \\
  & = \Prob \big( \OS_{1:L+1} \big| \IS_{1:L+1}, \jpast^\prime, \IS_0,
  \OS_0 \big) \\
  & \quad\quad\quad\quad\quad
  \cdot \Prob \big( \OS_0 \big| \IS_{1:L+1} , \jpast^\prime, \IS_0 \big) \\
  & \iff \\
\Prob \big( \OS_{1:L+1}, \OS_0 & \big| \IS_{1:L+1}, \jpast, \IS_0 \big) \\
  & = \Prob \big( \OS_{1:L+1}, \OS_0 \big| \IS_{1:L+1}, \jpast^\prime, \IS_0
  \big) \\
  & \iff \\
\Prob \big( \OS_{0:L+1} & \big| \IS_{0:L+1}, \jpast \big)
  = \Prob \big( \OS_{0:L+1} \big| \IS_{0:L+1},\jpast^\prime \big)
  ~.
\end{align*}
The two pasts are therefore equivalent for length-$L$ future morphs. By
induction, the two pasts are equivalent for arbitrarily long future morphs. 
\end{proof}
\end{Prop}

%\bibliography{chaos,quant_synergy}

\end{document}

%% file: cmechabbrev.tex
\newcommand{\eM}     {\mbox{$\epsilon$-machine}}
\newcommand{\eMs}    {\mbox{$\epsilon$-machines}}
\newcommand{\EM}     {\mbox{$\epsilon$-Machine}}
\newcommand{\EMs}    {\mbox{$\epsilon$-Machines}}
\newcommand{\eT}     {\mbox{$\epsilon$-transducer}}
\newcommand{\eTs}    {\mbox{$\epsilon$-transducers}}
\newcommand{\ET}     {\mbox{$\epsilon$-Transducer}}
\newcommand{\ETs}    {\mbox{$\epsilon$-Transducers}}

\newcommand{\CausalState}	{ \mathcal{S} }

\newcommand{\causalstate}	{ \sigma }
\newcommand{\CausalStateSet}	{ \boldsymbol{\CausalState} }
\newcommand{\AlternateState}	{ R }
\newcommand{\AlternateStateSet}	{ \mathcal{R} }
\newcommand{\PrescientState}	{ \widehat{\AlternateState} }

\newcommand{\PrescientStateSet}	{ \boldsymbol{\PrescientState}}

\newcommand{\Prob}      {\Pr} % use standard command

\newcommand{\Cmu}		{C_\mu}
\newcommand{\hmu}		{h_\mu}
\newcommand{\EE}		{{\bf E}}

\newcommand{\PC}		{\chi}

\newcommand{\InAlphabet}	{ \mathcal{X}}
\newcommand{\insymbol}		{ x}
\newcommand{\OutAlphabet}	{ \mathcal{Y}}
\newcommand{\outsymbol}		{ y}
\newcommand{\InputSimple}	{ X}

\newcommand{\OutputSimple}	{ Y}

\newcommand{\forward}{+}
\newcommand{\reverse}{-}
\newcommand{\forwardreverse}{\pm} % \pm

\newcommand{\FutureCausalState}	{ {\CausalState}^{\forward} }

\newcommand{\PastCausalState}	{ {\CausalState}^{\reverse} }

\newcommand{\ForwardCausalState}	{ {\CausalState}^{\forward} }

\newcommand{\lastindex}[2]{
  \edef\tempa{0}
  \edef\tempb{#2}
  \ifx\tempa\tempb
    \edef\tempc{#1}
  \else
    \edef\tempa{0}
    \edef\tempb{#1}
    \ifx\tempa\tempb
      \edef\tempc{#2}
    \else
      \edef\tempc{#1+#2}
    \fi
  \fi
  \tempc
}

\newcommand{\MOrder}{\ensuremath{R}}

\newcommand{\CSjoint}[1][,]{
   \edef\tempa{:}
   \edef\tempb{#1}
   \ifx\tempa\tempb
      \ensuremath{\FutureCausalState\!#1\PastCausalState}
   \else
      \ensuremath{\FutureCausalState#1\PastCausalState}
   \fi
}

\newif\ifpm 
\edef\tempa{\forwardreverse}
\edef\tempb{\pm}
\ifx\tempa\tempb
   \pmfalse
\else
   \pmtrue  
\fi

\renewcommand{\Pr}{\mathbb{P}}

%% file: et1.bbl
\begin{thebibliography}{10}

\bibitem{Crut12a}
J.~P. Crutchfield.
\newblock Between order and chaos.
\newblock {\em Nature Physics}, 8(January):17--24, 2012.

\bibitem{Pack80}
N.~H. Packard, J.~P. Crutchfield, J.~D. Farmer, and R.~S. Shaw.
\newblock Geometry from a time series.
\newblock {\em Phys. Rev. Let.}, 45:712, 1980.

\bibitem{Mand012a}
D.~Mandal and C.~Jarzynski.
\newblock Work and information processing in a solvable model of {Maxwell's
  Demon}.
\newblock {\em Proc. Natl. Acad. Sci. USA}, 109(29):11641--11645, 2012.

\bibitem{Boyd14b}
A.~B. Boyd and J.~P. Crutchfield.
\newblock Demon dynamics: {Deterministic} chaos, the {Szilard} map, and the
  intelligence of thermodynamic systems.
\newblock 2015.
\newblock SFI Working Paper 15-06-019; arxiv.org:1506.04327 [cond-mat.stat-
  mech].

\bibitem{Boyd15a}
A.~B. Boyd, D.~Mandal, and J.~P. Crutchfield.
\newblock Identifying functional thermodynamics in autonomous {Maxwellian}
  ratchets.
\newblock In preparation, 2015.

\bibitem{Riek99}
F.~Rieke, D.~Warland, R.~de~Ruyter~van Steveninck, and W.~Bialek.
\newblock {\em Spikes: {Exploring} the Neural Code}.
\newblock Bradford Book, New York, 1999.

\bibitem{Cuts11a}
V.~Cutsuridis, A.~Hussain, and J.~G. Taylor.
\newblock {\em Perception-Action Cycle}.
\newblock Springer, New York, New York, 2011.

\bibitem{Gord11a}
G.~Gordon, D.~Kaplan, D.~Lankow, D.~Little, J.~Sherwin, B.~Suter, et~al.
\newblock Toward an integrated approach to perception and action: {Conference}
  report and future directions.
\newblock {\em Front. Syst. Neurosci.}, 5(20), 2011.

\bibitem{Padg03a}
J.~F. Padgett, D.~Lee, and N.~Collier.
\newblock Economic production as chemistry.
\newblock {\em Industrial and Corporate Change}, 12(4):843--877, 2003.

\bibitem{Crut88a}
J.~P. Crutchfield and K.~Young.
\newblock Inferring statistical complexity.
\newblock {\em Phys. Rev. Let.}, 63:105--108, 1989.

\bibitem{Gray09a}
R.~M. Gray.
\newblock {\em Probability, Random Processes, and Ergodic Theory}.
\newblock Springer-Verlag, New York, second edition, 2009.

\bibitem{Cove06a}
T.~M. Cover and J.~A. Thomas.
\newblock {\em Elements of Information Theory}.
\newblock Wiley-Interscience, New York, second edition, 2006.

\bibitem{Yeun08a}
R.~W. Yeung.
\newblock {\em Information Theory and Network Coding}.
\newblock Springer, New York, 2008.

\bibitem{Crut08b}
C.~J. Ellison, J.~R. Mahoney, and J.~P. Crutchfield.
\newblock Prediction, retrodiction, and the amount of information stored in the
  present.
\newblock {\em J. Stat. Phys.}, 136(6):1005--1034, 2009.

\bibitem{Crut10a}
J.~P. Crutchfield, C.~J. Ellison, J.~R. Mahoney, and R.~G. James.
\newblock Synchronization and control in intrinsic and designed computation:
  {An} information-theoretic analysis of competing models of stochastic
  computation.
\newblock {\em CHAOS}, 20(3):037105, 2010.

\bibitem{Gray74a}
R.~M. Gray and L.~D. Davisson.
\newblock The ergodic decomposition of stationary discrete random processses.
\newblock {\em IEEE Trans. Info. Th.}, 20(5):625--636, 1974.

\bibitem{Bill61a}
P.~Billingsley.
\newblock Statistical methods in {Markov} chains.
\newblock {\em Ann. Math. Stat.}, 32:12, 1961.

\bibitem{Crut98d}
J.~P. Crutchfield and C.~R. Shalizi.
\newblock Thermodynamic depth of causal states: {O}bjective complexity via
  minimal representations.
\newblock {\em Phys. Rev. E}, 59(1):275--283, 1999.

\bibitem{Shal98a}
C.~R. Shalizi and J.~P. Crutchfield.
\newblock Computational mechanics: Pattern and prediction, structure and
  simplicity.
\newblock {\em J. Stat. Phys.}, 104:817--879, 2001.

\bibitem{Crut08a}
J.~P. Crutchfield, C.~J. Ellison, and J.~R. Mahoney.
\newblock Time's barbed arrow: {Irreversibility}, crypticity, and stored
  information.
\newblock {\em Phys. Rev. Lett.}, 103(9):094101, 2009.

\bibitem{Trav11a}
N.~Travers and J.~P. Crutchfield.
\newblock Equivalence of history and generator $\epsilon$-machines.
\newblock 2011.
\newblock SFI Working Paper 11-11-051; arxiv.org:1111.4500 [math.PR].

\bibitem{Crut92c}
J.~P. Crutchfield.
\newblock The calculi of emergence: Computation, dynamics, and induction.
\newblock {\em Physica D}, 75:11--54, 1994.

\bibitem{Paz71a}
A.~Paz.
\newblock {\em Introduction to Probabilistic Automata}.
\newblock Academic Press, New York, 1971.

\bibitem{Hopc79}
J.~E. Hopcroft and J.~D. Ullman.
\newblock {\em Introduction to Automata Theory, Languages, and Computation}.
\newblock Addison-Wesley, Reading, 1979.

\bibitem{Lind95a}
D.~Lind and B.~Marcus.
\newblock {\em An Introduction to Symbolic Dynamics and Coding}.
\newblock Cambridge University Press, New York, 1995.

\bibitem{Shan48a}
C.~E. Shannon.
\newblock A mathematical theory of communication.
\newblock {\em Bell Sys. Tech. J.}, 27:379--423, 623--656, 1948.

\bibitem{Note1}
Though see Ref. \cite [see Ch. 7]{Ash65a} and for early efforts Refs. \cite
  {Blac58a} and \cite {Blac61a}.

\bibitem{Gray90a}
R.~M. Gray.
\newblock {\em Entropy and Information Theory}.
\newblock Springer-Verlag, New York, 1990.

\bibitem{oppenheim2010discrete}
A.~V. Oppenheim and R.~W. Schafer.
\newblock {\em Discrete-time signal processing}.
\newblock Prentice-Hall signal processing series. Prentice Hall, 2010.

\bibitem{Mand77a}
B.~Mandelbrot.
\newblock {\em Fractals: Form, Chance and Dimension}.
\newblock W. H. Freeman and Company, New York, 1977.

\bibitem{Rabi89a}
L.~R. Rabiner.
\newblock A tutorial on hidden {Markov} models and selected applications.
\newblock {\em IEEE Proc.}, 77:257, 1989.

\bibitem{Elli95a}
R.~J. Elliot, L.~Aggoun, and J.~B. Moore.
\newblock {\em Hidden Markov Models: {E}stimation and Control}, volume~29 of
  {\em Applications of Mathematics}.
\newblock Springer, New York, 1995.

\bibitem{Jame10a}
R.~G. James, J.~R. Mahoney, C.~J. Ellison, and J.~P. Crutchfield.
\newblock Many roads to synchrony: Natural time scales and their algorithms.
\newblock {\em Phys. Rev. E}, 89:042135, 2014.

\bibitem{Crut93a}
J.~P. Crutchfield and J.~E. Hanson.
\newblock Turbulent pattern bases for cellular automata.
\newblock {\em Physica D}, 69:279 -- 301, 1993.

\bibitem{Li08a}
C.-B. Li, H.~Yang, and T.~Komatsuzaki.
\newblock Multiscale complex network of protein conformational fluctuations in
  single-molecule time series.
\newblock {\em Proc. Natl. Acad. Sci. USA}, 105:536--541, 2008.

\bibitem{Varn12a}
D.~P. Varn, G.~S. Canright, and J.~P. Crutchfield.
\newblock {$\epsilon$-Machine} spectral reconstruction theory: {A} direct
  method for inferring planar disorder and structure from {X}-ray diffraction
  studies.
\newblock {\em Acta. Cryst. Sec. A}, 69(2):197--206, 2013.

\bibitem{Stre13a}
C.~C. Strelioff and J.~P. Crutchfield.
\newblock Bayesian structural inference for hidden processes.
\newblock {\em Phys. Rev. E}, 89:042119, 2014.

\bibitem{Crut87a}
J.~P. Crutchfield and B.~S. McNamara.
\newblock Equations of motion from a data series.
\newblock {\em Complex Systems}, 1:417--452, 1987.

\bibitem{John10a}
B.~D. Johnson, J.~P. Crutchfield, C.~J. Ellison, and C.~S. McTague.
\newblock Enumerating finitary processes.
\newblock 2012.
\newblock SFI Working Paper 10-11-027; arxiv.org:1011.0036 [cs.FL].

\bibitem{Shal02a}
C.~R. Shalizi, K.~L. Shalizi, and J.~P. Crutchfield.
\newblock Pattern discovery in time series, {Part} {I}: Theory, algorithm,
  analysis, and convergence.
\newblock {\em J. Machine Learning Research}, submitted, 2002.
\newblock Santa Fe Institute Working Paper 02-10-060;
  arXiv.org/abs/cs.LG/0210025.

\bibitem{Stil07b}
S.~Still, J.~P. Crutchfield, and C.~J. Ellison.
\newblock Optimal causal inference: {Estimating} stored information and
  approximating causal architecture.
\newblock {\em CHAOS}, 20(3):037111, 2010.

\bibitem{Crut89j}
J.~P. Crutchfield.
\newblock Reconstructing language hierarchies.
\newblock In H.~A. Atmanspracher and H.~Scheingraber, editors, {\em Information
  Dynamics}, pages 45--60, New York, 1991. Plenum.

\bibitem{Hans90a}
J.~E. Hanson and J.~P. Crutchfield.
\newblock The attractor-basin portrait of a cellular automaton.
\newblock {\em J. Stat. Phys.}, 66:1415--1462, 1992.

\bibitem{Crut91d}
J.~P. Crutchfield.
\newblock Discovering coherent structures in nonlinear spatial systems.
\newblock In A.~Brandt, S.~Ramberg, and M.~Shlesinger, editors, {\em Nonlinear
  Ocean Waves}, pages 190--216, Singapore, 1992. World Scientific.

\bibitem{Hans95a}
J.~E. Hanson and J.~P. Crutchfield.
\newblock Computational mechanics of cellular automata: An example.
\newblock {\em Physica D}, 103:169--189, 1997.

\bibitem{McTa04a}
C.~S. McTague and J.~P. Crutchfield.
\newblock Automated pattern discovery---{An} algorithm for constructing
  optimally synchronizing multi-regular language filters.
\newblock {\em Theo. Comp. Sci.}, 359(1-3):306--328, 2006.

\bibitem{Crut04a}
J.~P. Crutchfield and O.~G{\"o}rnerup.
\newblock Objects that make objects: {The} population dynamics of structural
  complexity.
\newblock {\em J. Roy. Soc. Interface}, 3:345--349, 2006.

\bibitem{Crut10c}
J.~P. Crutchfield and S.~Whalen.
\newblock Structural drift: The population dynamics of sequential learning.
\newblock {\em PLoS Computational Biology}, 8(6):e1002510, 2010.

\bibitem{Crut94a}
J.~P. Crutchfield.
\newblock Optimal structural transformations---the $\epsilon$-transducer.
\newblock {\em UC Berkeley Physics Research Report}, 1994.

\bibitem{Shal01a}
C.~R. Shalizi.
\newblock {\em Causal Architecture, Complexity and Self-Organization in Time
  Series and Cellular Automata}.
\newblock PhD thesis, University of Wisconsin, Madison, Wisconsin, 2001.

\bibitem{Hopf82}
J.~J. Hopfield.
\newblock Neural networks and physical systems with emergent collective
  behavior.
\newblock {\em Proc. Natl. Acad. Sci.}, 79:2554, 1982.

\bibitem{Hert86a}
J.~Hertz, A.~Krogh, and R.~G. Palmer.
\newblock {\em An Introduction to the Theory of Neural Networks}, volume~1 of
  {\em Lecture Notes, Studies in the Sciences of Complexity}.
\newblock Addison-Wesley, Redwood City, California, 1991.

\bibitem{Hast11a}
T.~Hastie, R.~Tibshirani, and J.~Friedman.
\newblock {\em The Elements of Statistical Learning: {Data} Mining, Inference,
  and Prediction}.
\newblock Spinger, New York, second edition, 2011.

\bibitem{Mack03a}
D.~J.~C. MacKay.
\newblock {\em Information Theory, Inference and Learning Algorithms}.
\newblock Cambridge, Cambridge, United Kingdom, 2003.

\bibitem{Bial00a}
W.~Bialek, I.~Nemenman, and N.~Tishby.
\newblock Predictability, complexity, and learning.
\newblock {\em Neural Computation}, 13:2409--2463, 2001.

\bibitem{Shal99a}
C.~R. Shalizi and J.~P. Crutchfield.
\newblock Information bottlenecks, causal states, and statistical relevance
  bases: {How} to represent relevant information in memoryless transduction.
\newblock {\em Adv. Comp. Sys.}, 5(1):91--95, 2002.

\bibitem{Albe05a}
D.~Albers, J.~C. Sprott, and J.~P. Crutchfield.
\newblock Persistent chaos in high dimensions.
\newblock {\em Phys. Rev. E}, 74(5):057201, 2006.

\bibitem{Marz14b}
S.~Marzen and J.~P. Crutchfield.
\newblock Informational and causal architecture of discrete-time renewal
  processes.
\newblock {\em Entropy}, to appear, 2015.
\newblock SFI Working Paper 14-08-032; arxiv.org:1408.6876
  [cond-mat.stat-mech].

\bibitem{Crut87f}
J.~P. Crutchfield.
\newblock Information and its metric.
\newblock In L.~Lam and H.~C. Morris, editors, {\em Nonlinear Structures in
  Physical Systems - Pattern Formation, Chaos and Waves}, pages 119--130, New
  York, 1990. Springer-Verlag.

\bibitem{Shan53a}
C.~E. Shannon.
\newblock The lattice theory of information.
\newblock {\em IEEE Trans. Info. Th}, 1:105--107, 1953.

\bibitem{Li11a}
H.~Li and E.~K.~P. Chong.
\newblock On a connection between information and group lattices.
\newblock {\em Entropy}, 13:683--798, 2011.

\bibitem{Mors38a}
M.~Morse and G.~A. Hedlund.
\newblock Symbolic dynamics.
\newblock {\em Am. J. Math.}, 60(4):815--866, 1938.

\bibitem{Post21a}
E.~Post.
\newblock Introduction to the general theory of elementary propositions.
\newblock {\em Am. J. Math.}, 43:163--185, 1921.

\bibitem{Gode92a}
K.~G{\"o}del.
\newblock {\em On Formally Undecidable Propositions of Principia Mathematica
  and Related Systems}.
\newblock Dover Publications, 1992.

\bibitem{Turi37a}
A.~Turing.
\newblock On computable numbers, with an application to the
  {Entschiedungsproblem}.
\newblock {\em Proc. Lond. Math. Soc.}, 42, 43:230--265, 544--546, 1937.

\bibitem{Chur36a}
A.~Church.
\newblock A note on the {Entscheidungsproblem}.
\newblock {\em J. Symbolic Logic}, 1:40--41, 1936.

\bibitem{Moor56a}
E.~F. Moore.
\newblock Gedanken-experiments on sequential machines.
\newblock In C.~Shannon and J.~McCarthy, editors, {\em Automata Studies},
  number~34 in Annals of Mathematical Studies, pages 129--153. Princeton
  University Press, Princeton, New Jersey, 1956.

\bibitem{Huff54}
D.~Huffman.
\newblock The synthesis of sequential switching circuits.
\newblock {\em J. Franklin Inst.}, 257:161--190, 275--303, 1954.

\bibitem{Huff54b}
D.~Huffman.
\newblock Information conservation and sequence transducers.
\newblock In {\em Proc. Symp. Information Networks}, pages 291--307.
  Polytechnic Institute of Brooklyn, Brooklyn, New York, 1954.

\bibitem{Huff59a}
D.~Huffman.
\newblock Canonical forms for information-lossless finite-state logical
  machines.
\newblock {\em IRE Trans. Circ. Th.}, 6:41--59, 1959.

\bibitem{Huff59b}
D.~Huffman.
\newblock Notes on information-lossless finite-state automata.
\newblock {\em Il Nuovo Cimento}, 13(2 Supplement):397--405, 1959.

\bibitem{Broo89a}
J.~G. Brookshear.
\newblock {\em Theory of computation: {Formal} languages, automata, and
  complexity}.
\newblock Benjamin/Cummings, Redwood City, California, 1989.

\bibitem{Ash65a}
R.~B. Ash.
\newblock {\em Information Theory}.
\newblock John Wiley and Sons, New York, 1965.

\bibitem{Blac58a}
D.~Blackwell, L.~Breiman, and A.~J. Thomasian.
\newblock Proof of {Shannon's} transmission theorem for finite-state
  indecomposable channels.
\newblock {\em Ann. Math. Stat.}, 29(4):1209--1220, 1958.

\bibitem{Blac61a}
D.~Blackwell.
\newblock Exponential error bounds for finite state channels.
\newblock In {\em Proc. Fourth Berkeley Symp. on Math. Statist. and Prob.},
  volume~1, pages 57--63. Univ. of Calif. Press, 1961.

\bibitem{Hamm97a}
R.~W. Hamming.
\newblock {\em Digital Filterns}.
\newblock Dover Publications, New York, third edition, 1997.

\bibitem{Bate79}
G.~Bateson.
\newblock {\em Mind and Nature: A Necessary Unity}.
\newblock E. P. Dutton, New York, 1979.

\bibitem{Shaw81}
R.~Shaw.
\newblock Strange attractors, chaotic behavior, and information flow.
\newblock {\em Z. Naturforsh.}, 36a:80, 1981.

\bibitem{korner06}
R.~Ahlswede and J.~K{\"o}rner.
\newblock Appendix: On common information and related characteristics of
  correlated information sources.
\newblock In R.~Ahlswede, Baumer, N.~Cai, H.~Aydinian, V.~Blinovsky, C.~Deppe,
  and H.~Mashurian, editors, {\em General Theory of Information Transfer and
  Combinatorics}, volume 4123 of {\em Lecture Notes in Computer Science}, pages
  664--677. Springer Berlin, 2006.

\bibitem{plw-11}
P.~L. Williams and R.~D. Beer.
\newblock Generalized measures of information transfer.
\newblock {\em arXiv:1102.1507}, abs/1102.1507, 2011.

\bibitem{barnett-09}
L.~Barnett, A.~B. Barrett, and A.~K. Seth.
\newblock Granger causality and transfer entropy are equivalent for {Gaussian}
  variables.
\newblock {\em Phys. Rev. Lett.}, 103, 2009.

\bibitem{Schr00a}
T.~Schreiber.
\newblock Measuring information transfer.
\newblock {\em Phys. Rev. Lett.}, 85:461--464, 2000.

\bibitem{Mark73a}
H.~Marko.
\newblock The bidirectional communication theory---{A} generalization of
  information theory.
\newblock {\em IEEE Trans. Comm.}, COM-21:1345--1351, 1973.

\bibitem{Mass90a}
J.~L. Massey.
\newblock Causality, feedback and directed information.
\newblock In {\em Proc. 1990 Intl. Symp. on Info. Th. and its Applications,
  Waikiki, Hawaii, Nov. 27-30}, pages 1--6. 1990.

\bibitem{Sun14a}
J.~Sun and E.~M. Bollt.
\newblock Causation entropy identifies indirect influences, dominance of
  neighbors and anticipatory couplings.
\newblock {\em Physica D}, 267:49--57, 2014.

\bibitem{Barn14a}
N.~Barnett and J.~P. Crutchfield.
\newblock Computational mechanics of input-output processes: {Shannon}
  information measures and decompositions.
\newblock {\em in preparation}, 2014.

\bibitem{Fedo72a}
V.~V. Fedorov.
\newblock {\em Theory of Optimal Experiments}.
\newblock Probability and Mathematical Statistics. Academic Press, New York,
  1972.

\bibitem{Atki01}
A.~Atkinson, B.~Bogacka, and A.~A. Zhigljavsky, editors.
\newblock {\em Optimum Design 2000}.
\newblock Nonconvex Optimization and Its Applications. Springer, New York,
  2001.

\bibitem{Stil07c}
S.~Still.
\newblock Information-theoretic approach to interactive learning.
\newblock {\em EuroPhys. Lett.}, 85:28005, 2009.

\bibitem{Astr08a}
K.~J. Astrom and R.~M. Murray.
\newblock {\em Feedback Systems: {An} Introduction for Scientists and
  Engineers}.
\newblock Princeton University Press, Princeton, New Jersey, 2008.

\bibitem{Crut13a}
J.~P. Crutchfield, P.~Riechers, and C.~J. Ellison.
\newblock Exact complexity: {Spectral} decomposition of intrinsic computation.
\newblock Santa Fe Institute Working Paper 13-09-028; arXiv:1309.3792 [cond-
  mat.stat-mech].

\end{thebibliography}
